\documentclass[USenglish]{lipics}

\usepackage{microtype}
\usepackage{graphicx}
\usepackage{booktabs} 
\usepackage{amssymb}
\usepackage{stmaryrd}
\usepackage{paralist}
\usepackage{amsmath}
\usepackage{mathrsfs}
\usepackage{colonequals}
\usepackage[linesnumbered,ruled,vlined]{algorithm2e}
\SetKwRepeat{Do}{do}{While}
\usepackage{url}
\usepackage{bm}
\usepackage{tikz}
\usepackage{float}

\usepackage{xcolor}

\newcommand{\postreviews}[1]{}

\makeatother

\renewcommand{\Pr}{\mathrm{Pr}}
\newcommand{\dom}{\mathrm{dom}}
\newcommand{\defeq}{\colonequals}

\newcommand{\sigmai}{\sigma_{\mathrm{int}}}
\newcommand{\sigmae}{\sigma}

\newcommand{\card}[1]{\left|#1\right|}

\newcommand{\calA}{\mathcal{A}}
\newcommand{\calQ}{\mathcal{Q}}
\newcommand{\calD}{\mathcal{D}}
\newcommand{\calI}{\mathcal{I}}
\newcommand{\calE}{\mathcal{E}}
\newcommand{\calS}{\mathcal{S}}
\newcommand{\calT}{\mathcal{T}}

\newcommand{\II}{\mathrm{I}}
\newcommand{\I}{\mathrm{I}}
\newcommand{\QQ}{\mathrm{Q}}
\newcommand{\PP}{\mathrm{P}}

\newcommand{\ki}{k_{\II}}
\newcommand{\kq}{k_{\QQ}}
\newcommand{\kp}{k_{\PP}}

\newcommand{\NN}{\mathbb{N}}

\newcommand{\decode}[1]{\mathrm{dec}(#1)}

\newcommand{\arity}[1]{\mathrm{arity}(#1)}

\renewcommand\i{\mathrm{i}}

\renewcommand\r{\mathrm{r}}
\newcommand\AND{\mathrm{AND}}
\newcommand\OR{\mathrm{OR}}
\newcommand\NOT{\mathrm{NOT}}

\newcommand\true{\top}
\newcommand\false{\bot}

\newcommand\neigh{\mathrm{Nbh}}
\newcommand\la{\langle}
\newcommand\ra{\rangle}
\newcommand\strat{\zeta}
\newcommand\dec{\mathrm{dec}}
\newcommand\enc{\mathrm{enc}}

\newcommand\vars{\mathrm{vars}}

\newcommand{\inp}{\mathsf{inp}}
\newcommand{\cst}{\mathsf{cst}}

\newcommand{\domd}{\dom^\downarrow}
\newcommand{\domu}{\dom^\uparrow}

\newcommand{\CST}{\mathit{CST}}

\newcommand{\guard}{\mathrm{guard}}

\newcommand{\goal}{\text{Goal}()}

\newcommand{\first}{\mathsf{Ch1}}
\newcommand{\second}{\mathsf{Ch2}}
\newcommand{\firstchild}{\mathsf{FirstChild}}
\newcommand{\secondchild}{\mathsf{SecondChild}}
\newcommand{\binary}{\mathsf{Bin}}
\newcommand{\schildself}{\mathcal{S}^{\binary}_{\first,\second,\child,\childself}}

\newcommand{\rlabel}{\mathsf{Label}}
\newcommand{\child}{\mathsf{Child}}
\newcommand{\childself}{\mathsf{Child^?}}

\newcommand{\f}{\mathrm{f}}

\newcommand{\mnull}{\mathsf{null}}

\newcommand{\wire}{\rightarrow}

\newcommand{\ins}{\mathrm{ins}}

\hypersetup{
    colorlinks,
    linkcolor={red!50!black},
    citecolor={blue!50!black},
    urlcolor={blue!30!black}
}

\theoremstyle{plain}
\newtheorem{property}[theorem]{Property}

\title{Combined Tractability of Query Evaluation\protect\\via Tree Automata and
Cycluits\protect\\(Extended Version)}
\titlerunning{Combined Tractability of Query Evaluation via Tree Automata and Cycluits}

\author[1]{Antoine Amarilli}
\author[2]{Pierre Bourhis}
\author[1,4]{Mikaël Monet}
\author[3,4]{Pierre~Senellart}
\affil[1]{LTCI, T\'el\'ecom ParisTech, Universit\'e Paris-Saclay; Paris, France}
\affil[2]{CRIStAL, CNRS \& Université Lille 1; Lille, France}
\affil[3]{DI, École normale supérieure, PSL Research University; Paris, France}
\affil[4]{Inria Paris; Paris, France}

\Copyright{Antoine Amarilli; Pierre Bourhis; Mikaël Monet; Pierre Senellart}

\renewcommand{\phi}{\varphi}
\renewcommand{\epsilon}{\varepsilon}
\renewcommand{\leq}{\leqslant}
\renewcommand{\geq}{\geqslant}

\usepackage{apxproof}
\def\appendixbibliographystyle{alphaabbr}
\newtheoremrep{theorem}{Theorem}
\newtheoremrep{proposition}[theorem]{Proposition}
\newtheoremrep{lemma}[theorem]{Lemma}

\makeatletter
\renewcommand\subparagraph{\@startsection{subparagraph}{5}{\z@}%
                                       {1.5ex \@plus0.5ex \@minus .2ex}%
                                       {-1em}%
                                      {\sffamily\normalsize\bfseries}}
\renewenvironment{proofsketch}{\noindent\emph{Proof sketch.} }%
{\hfill\qed\vskip3pt}

\begin{document}

\maketitle

\begin{abstract}
  We investigate parameterizations of both database instances and queries
that make query evaluation fixed-parameter tractable in combined
complexity. We introduce a new Datalog fragment with stratified negation,
intensional-clique-guarded Datalog (ICG-Datalog), with linear-time
evaluation on structures of bounded treewidth for programs of bounded
rule size. Such programs capture in particular conjunctive queries with
simplicial decompositions of bounded width, guarded negation fragment
queries of bounded CQ-rank, or two-way regular path queries. Our
result is shown by compiling to alternating two-way
automata, whose semantics is defined via cyclic provenance circuits
(cycluits) that can be tractably evaluated. Last, we prove that
probabilistic query evaluation remains intractable in combined
complexity under this parameterization.

\end{abstract}

\section{Introduction}
Arguably the most fundamental task performed by database systems
is \emph{query evaluation}, namely, computing the results
of a query over a
database instance. Unfortunately, this task is well-known to be intractable in
\emph{combined complexity}~\cite{vardi1982complexity} even for simple query languages.

To address this issue, two main directions have been investigated. The first
is to restrict the class of \emph{queries} to ensure tractability, for
instance, to $\alpha$-acyclic conjunctive queries \cite{yannakakis1981algorithms}, this
being motivated by the idea that many real-world queries are simple and usually
small. 
The second approach restricts the
structure of database instances, e.g., requiring them to have bounded
\emph{treewidth}~\cite{robertson1986graph} (we call them \emph{treelike}).
This has been notably studied by
Courcelle~\cite{courcelle1990monadic}, to show the tractability of
monadic-second order logic
on treelike instances, but in \emph{data complexity} (i.e., for fixed
queries); the combined complexity is
generally nonelementary~\cite{meyer1975weak}.

This leaves open the main question studied in this paper: \emph{Which queries
can be efficiently evaluated, in combined complexity, on treelike databases?}
This question has been addressed by Gottlob, Pichler, and Fei~\cite{gottlob2010monadic} by
introducing \emph{quasi-guarded Datalog};
however, an unusual feature of this language is that programs must 
explicitly refer to the tree decomposition
of the instance. Instead, we try
to follow Courcelle's approach and investigate which queries can be
efficiently \emph{compiled to automata}. Specifically, rather than restricting
to a fixed class of ``efficient'' queries, we study \emph{parameterized}
query classes,
i.e., we define an efficient class of queries for each value of the parameter.
We further make the standard assumption that the signature is fixed; in
particular, its arity is constant.
This allows us to
aim for low combined complexity for query evaluation, namely, fixed-parameter tractability with
linear time complexity in the product of the input query and instance, called
\emph{FPT-linear} complexity.

Surprisingly, we are not aware of further existing work on tractable 
combined query evaluation for parameterized instances and queries, except from
an unexpected angle: the compilation of restricted query fragments to tree
automata on treelike instances was used in the context of \emph{guarded logics}
and other fragments,
to decide
\emph{satisfiability}~\cite{benedikt2016step} and
\emph{containment}~\cite{barcelo2014does}.
To do this, one usually establishes a \emph{treelike
model property} to restrict the search to models of low treewidth (but
dependent on the formula), and then compiles the formula to an
automaton, so that the problems reduce to emptiness testing: expressive
automata formalisms, such as \emph{alternating two-way automata}, are typically
used. One contribution of our work is to notice this connection, and show how
query evaluation on treelike instances can benefit from these ideas: for
instance, as we show, some queries can only be compiled efficiently to such
concise automata, and not to the more common bottom-up tree automata.

From there, the first main contribution of
this paper is to define the language of \emph{intensional-clique-guarded
Datalog} (ICG-Datalog), and show an efficient FPT-linear compilation procedure
for this language, parameterized by the body size of rules: this implies
FPT-linear combined complexity on treelike instances. While we present it as a
Datalog fragment, our language shares some similarities with guarded logics;
yet, its design incorporates several features
(fixpoints, clique-guards, guarded negation, guarding positive subformulae) that
are not usually found together in guarded fragments, but are important for query
evaluation. We show how the tractability of this language captures the tractability of such query
classes as two-way regular path queries \cite{barcelo2013querying} and
$\alpha$-acyclic conjunctive queries.

Already for
conjunctive queries, we show that the treewidth of queries is not
the right parameter to ensure efficient compilability. In fact, a second contribution of
our work is a lower bound: we show that bounded
treewidth queries cannot be efficiently compiled to automata at all, so we cannot hope
to show
combined tractability for them via automata methods. By
contrast, ICG-Datalog implies the combined tractability of bounded-treewidth queries with
an additional requirement (interfaces between bags must be clique-guarded),
which is the notion of \emph{simplicial decompositions} previously studied by
Tarjan \cite{tarjan1985decomposition}. To our knowledge, our paper is the first
to introduce this query class and to show its 
tractability on treelike instances. ICG-Datalog can be
understood as an extension of this fragment to disjunction, clique-guarded
negation, and inflationary fixpoints, that preserves tractability.

To derive our main FPT-linear combined complexity result, we define an
operational semantics for our tree automata 
by introducing a notion of
\emph{cyclic provenance circuits},
that we call \emph{cycluits}. These cycluits, the third contribution of our paper,
are well-suited as a provenance representation
for alternating two-way automata encoding ICG-Datalog programs, as they naturally
deal with both recursion and two-way traversal of a treelike instance,
which is less straightforward with provenance
formulae~\cite{green2007provenance} or
circuits~\cite{deutch2014circuits}. 
While we believe that this
natural generalization of Boolean circuits may be of independent interest, it
does not seem to have been studied in detail, except in the context of
integrated circuit design \cite{malik1993analysis,riedel2012cyclic}, where the semantics often features
feedback loops that involve negation; we prohibit these by focusing on \emph{stratified}
circuits, which we show can be evaluated in linear time.
We
show that the provenance of alternating two-way automata can be
represented as a stratified 
cycluit in FPT-linear time, generalizing results on bottom-up automata
and circuits from~\cite{amarilli2015provenance}.

Since cycluits directly give us a provenance representation of the query,
we then investigate \emph{probabilistic query evaluation}, which we
showed in~\cite{amarilli2015provenance} to be linear-time in data
complexity through the use of provenance circuits.
We show how to remove cycles, 
so as to apply message-passing methods
\cite{lauritzen1988local}, yielding a 2EXPTIME upper bound for the combined
complexity of probabilistic query evaluation. 
While we do not obtain
tractable probabilistic query evaluation in combined complexity, we give lower
bounds showing that this is unlikely.

\subparagraph*{Outline.} We give preliminaries in
Section~\ref{sec:prelim}, and then position our approach relative to
existing work in Section~\ref{sec:existing}. We then present our tractable
fragment, first for bounded-simplicial-width conjunctive queries in
Section~\ref{sec:cq}, then for our ICG-Datalog language in
Section~\ref{sec:ICG}. We then define our automata and compile ICG-Datalog to
them in Section~\ref{sec:compilation}, before introducing cycluits and showing
our provenance computation result in Section~\ref{sec:provenance}. We last study
the conversion of cycluits to circuits, and probability evaluation, in
Section~\ref{sec:rmcycles}.
Full proofs are provided in appendix.

\section{Preliminaries}
\label{sec:prelim}
A \emph{relational signature} $\sigma$ is a finite set of relation
names written $R$, $S$,
$T$, \dots, each with its associated \emph{arity}
$\arity{R} \in \NN$.
Throughout this work, \emph{we always assume the signature $\sigma$ to be
fixed}:
hence, its \emph{arity} $\arity{\sigma}$ (the maximal arity of relations
in~$\sigma$) is constant, and we further assume it is $>0$.
A \emph{($\sigma$-)instance} $I$
is a finite set of \emph{ground facts} on~$\sigma$,
i.e., $R(a_1, \ldots, a_{\arity{R}})$ with $R \in
\sigma$.
The \emph{active domain} $\dom(I)$ consists
of the
elements occurring in~$I$.

We study query evaluation for several \emph{query
languages} that are subsets of first-order (FO) logic
(e.g., conjunctive queries) or of second-order (SO) logic (e.g., Datalog).
Unless otherwise stated, we
only consider queries that are \emph{constant-free}, and \emph{Boolean}, so that an instance $I$ either
\emph{satisfies} a query $q$ ($I \models q$), or \emph{violates} it
($I \not\models q$), with the standard semantics~\cite{abiteboul1995foundations}.

We study the \emph{query evaluation problem} (or \emph{model checking}) for a query class $\mathcal{Q}$ and instance
class $\mathcal{I}$:
given an instance $I \in
\mathcal{I}$ and query $Q \in \mathcal{Q}$, check if $I \models Q$.
Its \emph{combined complexity} 
for 
$\calI$ and $\calQ$ is 
a function of~$I$ and~$Q$, whereas
\emph{data complexity} 
assumes $Q$ to be fixed.
We 
also study cases where $\calI$ and
$\calQ$ are \emph{parameterized}: given infinite sequences
$\calI_1, \calI_2, \ldots$ and $\calQ_1, \calQ_2, \ldots$, the \emph{query evaluation
problem parameterized by $k_{\mathrm{I}}$, $k_{\mathrm{Q}}$} 
applies to~$\calI_{k_{\II}}$ and $\calQ_{k_{\mathrm{Q}}}$.
The parameterized problem is \emph{fixed-parameter tractable} (FPT), for $(\calI_n)$ and $(\calQ_n)$, if
there is a constant $c \in \mathbb{N}$ and computable function $f$ such that the problem
can be solved 
with combined complexity
$O\left(f(\ki, \kq) \cdot (\card{I} \cdot \card{Q})^c\right)$.
For $c =
1$, we call it \emph{FPT-linear} (in $\card{I} \cdot \card{Q}$).
Observe that calling the problem FPT is more
informative than saying that it is in PTIME for fixed $\ki$ and $\kq$, as we are
further imposing that the polynomial degree $c$ does not depend on $\ki$ and $\kq$:
this follows the usual distinction in parameterized complexity between FPT and
classes such as XP~\cite{flum2006parameterized}.

\subparagraph*{Query languages.}
We first study fragments of FO, in particular, \emph{conjunctive queries} (CQ), i.e., 
existentially quantified conjunctions of atoms.
The \emph{canonical model} of a CQ $Q$ is the instance built from $Q$ by
seeing variables as elements and atoms as facts.
The \emph{primal graph} of~$Q$ has its variables as vertices, and
connects all variable pairs that co-occur in some atom.

Second, we study \emph{Datalog with stratified negation}.
We summarize the definitions here, see
\cite{abiteboul1995foundations} for details. 
A \emph{Datalog program} $P$
(without negation) over $\sigma$ (called the \emph{extensional signature})
consists of an \emph{intensional signature} $\sigmai$ disjoint from $\sigma$
(with the arity of $\sigmai$ being possibly greater than that of~$\sigma$),
a 0-ary \emph{goal predicate} $\text{Goal}$ in
$\sigmai$, and a set of \emph{rules}: those are
of the form $R(\mathbf{x}) \leftarrow \psi(\mathbf{x}, \mathbf{y})$, where
the \emph{head} $R(\mathbf{x})$ is an atom with $R \in
\sigmai$, and the \emph{body} $\psi$ is a CQ over $\sigmai \sqcup \sigmae$ where
each variable of~$\mathbf{x}$ must occur.
The \emph{semantics} $P(I)$ of $P$ over
an input $\sigma$-instance $I$
is defined by a least fixpoint of the
interpretation of~$\sigmai$: we start with $P(I) \defeq I$, and for any rule $R(\mathbf{x})
\leftarrow \psi(\mathbf{x}, \mathbf{y})$ and tuple $\mathbf{a}$ of $\dom(I)$,
when $P(I) \models \exists \mathbf{y} \,\psi(\mathbf{a}, \mathbf{y})$,
then we \emph{derive} the fact $R(\mathbf{a})$ and
add it to~$P(I)$, where we can then use it to derive more facts.
We have $I \models P$ iff we derive the fact $\text{Goal}()$.
The \emph{arity} of~$P$ is $\max(\arity{\sigma}, \arity{\sigmai})$.
$P$ is \emph{monadic} if $\sigmai$ has arity~$1$.

\emph{Datalog with stratified
negation}~\cite{abiteboul1995foundations} allows
negated \emph{intensional} atoms
 in bodies, but
requires~$P$ to have a \emph{stratification}, i.e.,
an ordered partition $P_1
\sqcup \cdots \sqcup P_n$ of the rules where:
\begin{compactenum}[(i)]
\item Each $R \in \sigmai$ has a \emph{stratum} $\strat(R) \in \{1, \ldots, n\}$ such that all rules
  with $R$ in the head are in~$P_{\strat(R)}$;
 \item For any $1 \leq i \leq n$ and $\sigmai$-atom $R(\mathbf{z})$ in a body of a rule of~$P_i$,
   we have $\strat(R) \leq i$;
 \item For any $1 \leq i \leq n$ and negated $\sigmai$-atom $R(\mathbf{z})$ in a
   body of~$P_i$, we have 
  $\strat(R) < i$. 
\end{compactenum}
The stratification ensures that we can define the semantics of a stratified
Datalog program by computing its interpretation for strata $P_1, \ldots, P_n$ in order:
atoms in bodies always depend on a lower stratum, and negated atoms depend on
strictly lower strata, whose
interpretation was already fixed. Hence, there is a unique least fixpoint and
$I \models P$ is well-defined.

\begin{example}
  The following stratified Datalog program, with $\sigma = \{R\}$ and $\sigmai =
  \{T, \text{Goal}\}$, and strata $P_1$,
  $P_2$, tests if there are two elements that are not connected by a directed
  $R$-path:\\
  \null\hfill\(
    P_1: T(x,y) \leftarrow R(x,y), \quad
      T(x,y) \leftarrow R(x,z) \land T(z,y) \qquad\qquad P_2: 
    \text{Goal}() \leftarrow \lnot T(x,y)
    \)\hfill\null
\end{example}

\subparagraph*{Treewidth.}
Treewidth is a measure quantifying how far a graph is to being a tree, which we
use to restrict instances and conjunctive queries.
The \emph{treewidth} of a CQ is that of its canonical instance, and
the \emph{treewidth} of an instance $I$ is the smallest $k$ such that $I$ has a 
\emph{tree decomposition} of \emph{width} $k$,
i.e., a finite, rooted, unranked tree $T$, whose nodes~$b$ (called
\emph{bags}) are labeled by a subset $\dom(b)$ of $\dom(I)$ with $\card{\dom(b)}
\leq k+1$, and which satisfies:
\begin{compactenum}[(i)]
\item for every
fact $R(\mathbf{a}) \in I$, there is a bag $b \in T$ with $\mathbf{a} \subseteq
\dom(b)$; 
\item for all $a \in \dom(I)$, the set of bags $\{b \in T \mid a \in
  \dom(b)\}$ is a connected subtree of~$T$.
\end{compactenum}
A family of instances is \emph{treelike} if their treewidth is bounded by a
constant.

\section{Approaches for Tractability}
\label{sec:existing}
We now review existing approaches to ensure the tractability of query
evaluation, starting by query languages whose evaluation is tractable
in combined complexity on \emph{all} input instances. We then study more
expressive query languages which are tractable on \emph{treelike} instances, but
where tractability only holds in data complexity. We then present the
goals of our work.

\subsection{Tractable Queries on All Instances}
\label{sec:tractableall}

The best-known query language to ensure tractable query complexity is
\emph{$\alpha$-acyclic queries}~\cite{fagin1983degrees}, i.e., those that have a tree
decomposition where the domain of each bag corresponds exactly to an atom: this
is called a \emph{join tree} \cite{gottlob2002hypertree}.
With Yannakakis's algorithm
\cite{yannakakis1981algorithms}, we can evaluate an $\alpha$-acyclic
conjunctive query $Q$
on an arbitrary instance $I$ in time $O(\card{I} \cdot \card{Q})$.

Yannakakis's result was generalized in two main directions. One
direction~\cite{gottlob2014treewidth}, moving from linear time to PTIME,
has investigated more general CQ classes, in particular
CQs of bounded treewidth~\cite{flum2002query},
\emph{hypertreewidth}~\cite{gottlob2002hypertree},
and \emph{fractional hypertreewidth}~\cite{grohe2014constraint}.
Bounding these query parameters to some
fixed $k$ makes query evaluation run in time
$O((\card{I} \cdot \card{Q})^{f(k)})$ for some function $f$, hence in
PTIME; for
treewidth, since the decomposition can be computed
in FPT-linear time~\cite{bodlaender1996linear}, this goes down to
$O(\card{I}^k\cdot\card{Q})$. However, query
evaluation on arbitrary instances is unlikely to be FPT when parameterized by
the query treewidth, since it would imply that the exponential-time hypothesis
fails (Theorem 5.1 of~\cite{marx2010can}). 
Further, even for treewidth~2 (e.g.,
triangles), it is not known if 
we can achieve linear data complexity~\cite{alon1997finding}.

In another direction, $\alpha$-acyclicity has been generalized to queries
with more expressive operators, e.g.,
disjunction or negation. The result on $\alpha$-acyclic CQs thus extends
to the \emph{guarded fragment} (GF) of first-order logic, which can be evaluated
on arbitrary instances
in time $O(\card{I} \cdot \card{Q})$
\cite{leinders2005semijoin}.
Tractability is independently known for FO$^k$, the fragment of
FO where subformulae use at most $k$ variables, with
a simple evaluation algorithm in $O(\card{I}^k \cdot \card{Q})$~\cite{vardi1995complexity}.

Another important operator are
\emph{fixpoints}, which can be used to express, e.g., reachability
queries.
Though 
FO$^k$ is no longer tractable when adding fixpoints~\cite{vardi1995complexity},
query evaluation is tractable
\cite[Theorem~3]{berwanger2001games} 
for
$\mu$GF, 
i.e., GF with some restricted least and greatest fixpoint operators, when
\emph{alternation depth} is bounded; without alternation, the combined
complexity
is in $O(\card{I} \cdot \card{Q})$.
We could alternatively express fixpoints in Datalog, 
but, sadly, most known tractable fragments are nonrecursive:
nonrecursive stratified Datalog is tractable 
\cite[Corollary~5.26]{flum2002query}
for rules with restricted bodies
(i.e., strictly acyclic, or bounded strict treewidth). 
This result was generalized in \cite{gottlob2003robbers} when bounding the number
of guards: this nonrecursive fragment is shown to be equivalent to the
$k$-guarded fragment of FO, with connections to the bounded-hypertreewidth
approach.
One recursive tractable fragment is Datalog LITE, which is equivalent to alternation-free
$\mu$GF \cite{gottlob2002datalog}.
Fixpoints were independently studied for graph query languages
such as
reachability queries and \emph{regular path queries}
(RPQ), which enjoy linear combined complexity on arbitrary input instances: this
extends to two-way RPQs (2RPQs) and even strongly acyclic conjunctions of 2RPQs (SAC2RPQs),
which are expressible in alternation-free $\mu$GF. Tractability also extends to acyclic RPQs
but with PTIME complexity~\cite{barcelo2013querying}.

\subsection{Tractability on Treelike Instances}
\label{sec:treelike}

We now study another approach for tractable query evaluation:
this time, we restrict the shape of the \emph{instances}, using treewidth. This
ensures that we can
translate them to a tree for efficient query evaluation.
Informally, having fixed the signature $\sigma$, for a fixed treewidth $k \in
\NN$, there is a finite tree alphabet $\Gamma^k_\sigma$ such that
$\sigma$-instances of treewidth $\leq k$ can be translated in FPT-linear time
(parameterized by $k$), following the structure of a tree decomposition, to a
\emph{$\Gamma^k_\sigma$-tree}, i.e., a rooted full ordered binary tree with nodes labeled by
$\Gamma^k_\sigma$, which we call a \emph{tree encoding}. We omit the formal
construction: see Appendix~\ref{apx:tree-encodings} for more details.

We can then evaluate queries on treelike instances by running \emph{tree
automata} on the tree encoding that represents them.
Formally, given an alphabet $\Gamma$, a \emph{bottom-up
nondeterministic tree automaton}
on $\Gamma$-trees (or $\Gamma$-bNTA) is a tuple $A = (Q,F,\iota,\Delta)$,
where:

\begin{compactenum}[(i)]
	\item $Q$ is a finite set of \emph{states};
	\item $F \in Q$ is a subset of \emph{accepting states};
	\item $\iota : \Gamma \to 2^Q$ 
          is an \emph{initialization function} 
          determining the state of a leaf from its label;
	\item $\Delta :  \Gamma \times Q^2
          \to 2^Q$ 
          is a \emph{transition function} 
		determining the possible states for an internal 
                node from its label and the states of its two
                children.
\end{compactenum}
Given a $\Gamma$-tree $\la T,\lambda\ra$ 
(where $\lambda : T \to \Gamma$ is the \emph{labeling function}),
we define a \emph{run} of $A$ on~$\la T,\lambda\ra$
as a function $\phi : T \to Q$ such that
(1)~$\phi(l) \in \iota(\lambda(l))$  for every leaf $l$ of~$T$; and
(2)~$\phi(n) \in \Delta( \lambda(n), \phi(n_1),\phi(n_2))$
for every internal node $n$ of~$T$ with children $n_1$ and $n_2$.
The bNTA $A$ \emph{accepts} $\la T,\lambda\ra$
if
it has a run on~$T$ mapping
the root of~$T$ to a state of~$F$.

We say that a bNTA $A$ \emph{tests} a query $Q$ for treewidth~$k$ if, for any
$\Gamma^k_\sigma$-encoding $\la E,\lambda\ra$ coding an instance $I$ (of
treewidth $\leq k$), 
$A$ accepts $\la E,\lambda\ra$ iff $I \models Q$.
By a well-known result of Courcelle \cite{courcelle1990monadic} on graphs
(extended to higher-arity in \cite{flum2002query}),
we can use bNTAs to evaluate all queries in \emph{monadic second-order
logic} (MSO), i.e., first-order logic with second-order variables of arity~$1$.
MSO subsumes in particular CQs 
and monadic
Datalog (but not general Datalog).
Courcelle showed that MSO
queries can be compiled to a bNTA that tests them:

\begin{theorem}[\cite{courcelle1990monadic,flum2002query}]
  \label{thm:courcelle}
  For any MSO query $Q$ and treewidth $k \in \NN$, we can
  compute a bNTA that tests $Q$ for treewidth $k$.
\end{theorem}

This implies that evaluating any MSO query $Q$ has FPT-linear
\emph{data complexity} when parameterized by $Q$ 
and the instance treewidth~\cite{courcelle1990monadic,flum2002query},
i.e., it is in $O\left(f(\card{Q}, k) \cdot
  \card{I}\right)$ for some computable function~$f$.
However, this tells little about the combined complexity, as 
$f$ is generally nonelementary 
in $Q$ \cite{meyer1975weak}.
A better combined complexity bound is known for unions of conjunctions of two-way regular path queries (UC2RPQs) that are further required to be acyclic and to have a constant number
of edges between pairs of variables: these can be compiled into
polynomial-sized alternating two-way automata \cite{barcelo2014does}.

\subsection{Restricted Queries on Treelike Instances}

Our approach combines both ideas: we use instance treewidth as a parameter, but also restrict the queries to ensure tractable
compilability.
We are only aware of two approaches in this spirit.
First, 
Gottlob, Pichler, and Wei~\cite{gottlob2010monadic}
have proposed a
\emph{quasiguarded} Datalog fragment on \emph{relational structures
and their tree decompositions}, 
with query evaluation is in $O(|I|\cdot|Q|)$. However,
this formalism requires queries to be expressed in terms of the tree
decomposition, and not just in terms of the relational signature.
Second, Berwanger and Grädel~\cite{berwanger2001games} remark (after Theorem~4) that, when
alternation depth and \emph{width} are bounded,
$\mu$CGF (the 
\emph{clique-guarded} fragment of FO with fixpoints) enjoys FPT-linear query evaluation when
parameterized by instance treewidth.
Their approach does not rely on automata methods, and subsumes the tractability
of $\alpha$-acyclic CQs and alternation-free $\mu$GF (and hence SAC2RPQs), on treelike instances.
However, $\mu$CGF is a restricted query language (the only CQs
that it can express are those with a chordal primal graph),
whereas we want a richer language, with a
parameterized definition.

Our goal is thus to develop an expressive parameterized
query language, which can be compiled in \emph{FPT-linear time} to an automaton
that tests it (with the treewidth of instances also being a parameter). We can then evaluate
the automaton, and obtain FPT-linear combined complexity for query evaluation.
Further, as we will show, the use of tree
automata will yield \emph{provenance representations} for the
query as in~\cite{amarilli2015provenance} (see
Section~\ref{sec:provenance}).

\section{Conjunctive Queries on Treelike Instances}
\label{sec:cq}
To identify classes of queries that can be efficiently compiled to tree
automata, we start by the simplest queries: \emph{conjunctive
queries}.

\subparagraph*{$\bm{\alpha}$-acyclic queries.}
A natural candidate for a tractable query class via automata methods would be
\emph{$\alpha$-acyclic} CQs, which, as we explained in
Section~\ref{sec:tractableall},
can be evaluated in time $O(\card{I} \cdot \card{Q})$ on all instances. Sadly,
we show that such queries cannot be compiled efficiently to bNTAs, so
our compilation result (Theorem~\ref{thm:courcelle}) does not extend directly:
\begin{propositionrep}
  \label{prp:bntalower}
  There is an arity-two signature $\sigma$ and an infinite family $Q_1, Q_2, \ldots$ of
  $\alpha$-acyclic CQs
  such
  that, for any $i \in \NN$, any bNTA that tests~$Q_i$ for treewidth $1$ must have 
  $\Omega(2^{\card{Q_i}^{1-\epsilon}})$ states for any $\epsilon>0$.
\end{propositionrep}

\begin{proof}
  We fix the signature $\sigma$ to consist of binary relations $S$, $S_0$, $S_1$, and
  $C$. We will code binary numbers as gadgets on this fixed signature. The
  coding of $i \in \NN$ at length $k$, with $k \geq 1 + \lceil \log_2 i
  \rceil$,
  consists of an $S$-chain $S(a_1, a_2), \ldots,
  S(a_{k-1}, a_k)$, and facts $S_{b_j}(a_{j+1},
  a'_{j+1})$ for $1\leq j\leq k-1$ where $a'_{j+1}$ is a fresh element and $b_j$ is the $j$-th bit in the binary expression of $i$ (padding
  the most significant bits with 0). We now
  define the query family $Q_i$: each $Q_i$ is formed by picking a root variable
  $x$ and gluing $2^i$ chains to~$x$; for $0\leq j\leq 2^i-1$,
  we have one chain that is the 
  concatenation of a chain of~$C$ of length $i$ and the length-$(i+1)$ coding of
  $j$ using a gadget. Clearly the size of $Q_i$ is $O(i \times
  2^i)$ and thus
  $2^i=\Omega(\card{Q_i}^{1-\epsilon/2})$.

  Fix $i > 0$. Let $A$ be a bNTA testing $Q_i$ on instances of treewidth
  $1$. We will show that $A$~must have at least ${2^{i} \choose
  2^{i-1}}=\Omega\left(2^{2^{i}-\frac{i}{2}}\right)$ states (the lower
  bound is obtained from Stirling's formula), from which
  the claim follows. In fact, we will consider a specific subset $\calI$ of the instances
  of treewidth $\leq 1$, and a specific set $\calE$ of tree encodings of
  instances of $\calI$, and show the claim on $\calE$, which suffices to
  conclude.

  To define $\calI$, let $\calS_i$ be the set of subsets of
  $\{0, \ldots, 2^i-1\}$ of cardinality $2^{i-1}$, so that $\card{\calS_i}$ is
  $2^{i} \choose 2^{i-1}$. 
  We will first define a family~$\calI'$ of instances indexed by $\calS_i$
  as follows. Given $S \in \calS_i$, the
  instance $I'_S$ of $\calI'$ is obtained by constructing a full binary tree of
  the $C$-relation of height $i-1$, and identifying, for all $j$, the $j$-th
  leaf node with element $a_1$ of the length-$(i+1)$ coding of the $j$-th
  smallest number in~$S$.
  We now define the instances of $\calI$ to consist of a root element
  with two $C$-children, each of which are the root element of an instance of
  $\calI'$ (we call the two the \emph{child instances}). It is clear
  that instances of $\calI$ have treewidth $1$, and we can check quite easily
  that an instance of $\calI$ satisfies $Q_i$ iff the child instances $I'_{S_1}$ and
  $I'_{S_2}$ are such that $S_1 \cup S_2 = \{1, \ldots, 2^i\}$.
  
  We now define $\calE$ to be tree encodings of instances of $\calI$: refer to
  Appendix~\ref{apx:tree-encodings} for details of how they are defined. First,
  define $\calE'$ to consist of tree encodings of instances of~$\calI'$, which
  we will also index with $\calS_i$, i.e., $E_S$ is a tree encoding of $I'_S$. We
  now define $\calE$ as the tree encodings $E$ constructed as follows:
  given an instance $I \in \calI$, we encode it as a root bag with
  domain $\{r\}$, where $r$ is the root of the tree $I$, and no fact, the
  first child $n_1$ of the root bag having
  domain $\{r, r_1\}$ and fact $C(r, r_1)$, the second child $n_2$
  of the root being defined in the same way. Now, $n_1$ has one dummy
  child with empty domain and no fact, and one child
  which is the root of some tree encoding in~$\calE$ of one child instance of
  $I$. We define $n_2$ analogously with the other child instance.

  For each $S \in \calS_i$, letting $\bar S$ be the complement of~$S$ relative
  to $\{0, \ldots, 2^i-1\}$, we call $I_S \in \calI$ the instance where the
  first child instance is $I'_S$ and the second child instance is $I'_{\bar S}$,
  and we call $E_S \in \calE$ the tree encoding of~$I_S$ according to the
  definition above. We then call $\calQ_S$ the set of states $q$ of~$A$ such that
  there exists a run of~$A$ on~$E_S$ where the root of the encoding of the first child
  instance is mapped to~$q$. As each $I_S$ satisfies~$Q$, each $E_S$
  should be accepted by the automaton, so each $\calQ_S$ is non-empty.
  
  Further,
  we show that the $\calQ_S$ are pairwise disjoint: for any $S_1 \neq S_2$ of~$\calS_i$,
  we show that $\calQ_{S_1} \cap \calQ_{S_2} = \emptyset$. Assume to the
  contrary the existence of $q$ in the intersection, and let $\rho_{S_1}$ and
  $\rho_{S_2}$ be runs of~$A$ respectively on $I_{S_1}$ and $I_{S_2}$ that
  witness respectively that $q \in \calQ_{S_1}$ and $q \in \calQ_{S_2}$. Now,
  consider the instance $I \in \calI$ where the first child instance is~$I_1$,
  and the second child instance is $\bar{I_2}$, and let $E \in \calE$ be the tree
  encoding of~$I$. We can construct a run~$\rho$ of~$A$ on~$E$ by defining
  $\rho$ according to~$\rho_{S_2}$ except that, on the subtree of~$E$ rooted at
  the root $r'$ of the tree encoding of the first child instance, $\rho$ is defined
  according to~$\rho_{S_1}$: this is possible because $\rho_{S_1}$ and
  $\rho_{S_2}$ agree on~$r_1'$ as they both map $r'$ to~$q$. Hence, $\rho$
  witnesses that $A$ accepts $E$. Yet, as $I_1 \neq I_2$, we know that $I$ does
  not satisfy~$Q$, so that, letting $E \in \calE$ be its tree encoding, $A$
  rejects~$E$. We have reached a contradiction, so indeed the $\calQ_S$ are
  pairwise disjoint.

  As the $\calQ_S$ are non-empty, 
  we can construct an mapping from $\calS_i$ to the state
  set of $A$ by mapping each $S \in \calS_i$ to some state of~$\calQ_S$: as the
  $\calQ_S$ are pairwise disjoint, this mapping is injective.
  We deduce that the state set of~$A$ has size at least $\card{\calS_i}$, which
  concludes from the bound on the size of~$\calS_i$ that we showed previously.
\end{proof}

The intuition of the proof is that bNTAs can only make one traversal of the
encoding of the input instance. Faced by this, we propose to use different tree automata formalisms,
which are generally more concise than bNTAs.
There are two classical generalizations of nondeterministic automata, on words~\cite{birget1993state}
and on trees~\cite{tata}: one goes from the inherent existential
quantification of nondeterminism to
\emph{quantifier alternation}; the other allows
\emph{two-way} navigation instead of
imposing a left-to-right (on
words) or bottom-up (on trees) traversal. On words, both of these extensions
independently allow for exponentially more compact
automata~\cite{birget1993state}. In this work, we combine both extensions
and use
\emph{alternating two-way
tree automata}~\cite{tata, cachat2002two}, formally introduced in
Section~\ref{sec:compilation}, which leads to tractable 
combined complexity for evaluation.
Our general results in the next section will then imply:

\begin{propositionrep}\label{prp:alphaatwa}
  For any treewidth bound $\ki \in \NN$, given an $\alpha$-acyclic CQ $Q$, we can
  compute in FPT-linear time in $O(\card{Q})$ (parameterized by $\ki$) an alternating two-way tree automaton that tests
  it for treewidth~$\ki$.

  Hence, if we are additionally given a relational instance $I$ of treewidth
  $\leq \ki$, one can determine whether $I\models Q$ in FPT-linear time in
  $\card{I}\cdot\card{Q}$ (parameterized by $\ki$).
\end{propositionrep}

\begin{proof}
  Given the $\alpha$-acyclic CQ $Q$, we can compute in linear time in~$Q$ a
  chordal decomposition~$T$ (equivalently, a join tree) of~$Q$ (Theorem~5.6
  of~\cite{flum2002query}, attributed to~\cite{tarjan1984simple}). As $T$ is in
  particular a simplicial decomposition of~$Q$ of width $\leq \arity{\sigma} -
  1$, i.e., of constant width, we use Proposition~\ref{prp:simplicial} to
  compile in linear time in $\card{Q}$ an
  ICG-Datalog program $P$ with body size bounded by a constant $\kp$.
  
  We now use Theorem~\ref{thm:maintheorem} to construct, in FPT-linear time in
  $\card{P}$ (hence, in $\card{Q}$), parameterized by $\ki$ and the constant
  $\kp$, a SATWA $A$ testing $P$ for treewidth $\ki$.

  We now observe that, thanks to the fact that $Q$ is monotone, the SATWA $A$
  does not actually feature any negation: the translation in the proof of
  Proposition~\ref{prp:simplicial} does not produce any negated atom, and
  the compilation in the proof of Theorem~\ref{thm:maintheorem} only
  produces a negated state within a Boolean formula when there is a corresponding
  negated atom in the Datalog program. Hence, $A$ is
  actually an alternating two-way tree automaton, which proves the first part of
  the claim.

  For the second part of the claim, we use Theorem~\ref{thm:main} to evaluate
  $P$ on $I$ in FPT-linear time in $\card{I} \cdot \card{P}$, parameterized by
  the constant $\kp$ and $\ki$. This proves the claim.
\end{proof}

\subparagraph*{Bounded-treewidth queries.}
Having re-proven the combined tractability of $\alpha$-acyclic queries (on
bounded-treewidth instances),
we naturally try to extend to \emph{bounded-treewidth} CQs. Recall from
Section~\ref{sec:tractableall} that these queries have PTIME combined
complexity on all instances, but are unlikely to be FPT when parameterized by
the query treewidth \cite{marx2010can}. Can they be efficiently evaluated on
treelike instances by compiling them to automata? We answer in the negative:
that bounded-treewidth CQs \emph{cannot} be
efficiently compiled to automata to test them, even when using the expressive
formalism of alternating two-way tree automata \cite{tata}:

\begin{theoremrep}\label{thm:nocompile}
  There is an arity-two signature $\sigma$ for which there is no
  algorithm $\calA$ 
  with exponential running time and polynomial output size for the 
  following task:
  given a conjunctive query $Q$ of treewidth $\leq 2$,
  produce an alternating two-way tree automaton $A_Q$ on $\Gamma^5_\sigma$-trees
  that tests $Q$ on
  $\sigma$-instances of treewidth $\leq 5$.
\end{theoremrep}

\begin{toappendix}
To prove this theorem, we need some notions and lemmas from
  \cite{unpublishedbenediktmonadic}, an extended version
  of~\cite{benedikt2012monadic}. Since \cite{unpublishedbenediktmonadic}
  is currently unpublished, we
  provide in Appendix~\ref{app:mdl} the complete relevant material from
  this paper, in particular Lemma~\ref{lem:caninst},
  Theorem~\ref{thm:containment}, and their proofs.
\end{toappendix}

\begin{toappendix}
\begin{proof}[Proof of Theorem~\ref{thm:nocompile}]
  Let $\sigma$ be $\schildself$ as in Theorem~\ref{thm:containment}. We
  pose $c=3$, $\ki=2\times 3-1=5$.
  Assume by way of
  contradiction that there exists an algorithm $\calA$ satisfying the prescribed
  properties.
  We will describe an algorithm to solve any instance of the
  containment problem of Theorem~\ref{thm:containment} in singly exponential
  time. As Theorem~\ref{thm:containment} states that it is 2EXPTIME-hard, this
  yields a contradiction by the time hierarchy theorem.

  Let $P$ and $Q$ be an instance of the containment problem of
  Theorem~\ref{thm:containment}, where $P$ is a monadic Datalog program of
  var-size $\leq 3$,
  and $Q$ is a CQ of treewidth $\leq 2$. We will show how to solve the
  containment problem, that is, decide whether there exists some instance $I$
  satisfying $P \land \neg Q$.
  
  Using
  Lemma~\ref{lem:caninst}, compute in singly exponential time the
  $\Gamma_\sigma^{\ki}$-bNTA
  $A_P$. Using the putative algorithm $\calA$ on~$Q$, compute in singly
  exponential time an alternating two-way automaton $A_Q$ of polynomial size. As $A_P$ describes a family $\calI$ of canonical
  instances for~$P$, there is an instance satisfying $P \wedge \neg Q$ iff there is an
  instance in $\calI$ satisfying $P \wedge \neg Q$. Now, as $\calI$ is described
  as the decodings of the language of~$A_P$, all instances in $\calI$ have
  treewidth~$\leq \ki$. Furthermore, the instances in~$\calI$ satisfy $P$ by definition of $\calI$.
  Hence, there is an instance satisfying $P\wedge\neg Q$ iff there is an encoding $E$ in the
  language of $A_P$ whose decoding satisfies $\neg Q$. Now, as $A_Q$ tests $Q$ on
  instances of treewidth $\ki$, this is the case iff there is an encoding $E$ in the
  language of $A_P$ which is not accepted by $A_Q$. Hence, our problem is
  equivalent to the problem of deciding whether there is a tree accepted by
  $A_P$ but not by $A_Q$.
  
  We now use Theorem~A.1 of~\cite{cosmadakis1988decidable} to compute in EXPTIME
  in~$A_Q$ a bNTA $A'_Q$ recognizing the complement of the language
  of~$A_Q$. Remember that $A_Q$ was computed in EXPTIME and is of polynomial
  size, so the entire process so far is EXPTIME. Now we know that we can solve
  the containment problem by testing whether $A_P$ and $A'_Q$ have non-trivial
  intersection, which can be done in PTIME by computing the product automaton
  and testing emptiness~\cite{tata}. This solves the containment problem in
  EXPTIME. As we explained initially, we have reached a contradiction, because
  it is 2EXPTIME-hard.
\end{proof}
\end{toappendix}

This result is obtained from a variant of the 2EXPTIME-hardness of monadic
Datalog containment \cite{benedikt2012monadic}.
We show that efficient compilation of bounded-treewidth CQs to
automata would yield an EXPTIME containment test, and conclude by the
time hierarchy theorem.

\subparagraph*{Bounded simplicial width.}
We have shown that we cannot compile bounded-treewidth queries to automata
efficiently.
We now show that efficient compilation can be ensured with an additional
requirement on tree decompositions.
As it turns out, the resulting decomposition notion
has been independently introduced for graphs:

\begin{definition}[\cite{diestel1989simplicial}]
  A \emph{simplicial decomposition of a graph $G$} is a tree decomposition $T$ of~$G$
  such that, for any bag $b$ of $T$ and child bag $b'$ of~$b$, letting $S$
  be the intersection of the domains of $b$ and $b'$, then the subgraph of $G$ induced by $S$ is a complete
  subgraph of~$G$.
\end{definition}

We extend this notion to CQs, and introduce the \emph{simplicial width} measure:

\begin{definition}
  \label{def:simplicialcq}
  A \emph{simplicial decomposition of a CQ $Q$} is a simplicial
  decomposition of its primal graph.
  Note that any
  CQ has a simplicial
  decomposition (e.g., the trivial one that puts all variables in one
  bag).
  The \emph{simplicial width} of~$Q$
  is the minimum, over all simplicial tree decompositions, of the size of
  the largest bag minus~$1$.
\end{definition}

Bounding the simplicial width of CQs is of course more restrictive than
bounding their treewidth, and
this containment relation is strict: cycles have
treewidth $\leq 2$ but 
have unbounded simplicial width. This
being said, bounding the simplicial width is less restrictive than imposing
$\alpha$-acyclicity:
the join tree of an $\alpha$-acyclic CQ is
in particular a simplicial decomposition, so $\alpha$-acyclic CQs have
simplicial width at most $\arity{\sigma}-1$, which is constant as $\sigma$ is
fixed. Again, the containment is strict: a triangle has simplicial width $2$ but is not $\alpha$-acyclic.

To our knowledge, 
simplicial width for CQs has not been studied before.
Yet, we show that bounding the
simplicial width ensures that CQs can be efficiently compiled to
automata. This is unexpected, because the same is not
true of treewidth, by Theorem~\ref{thm:nocompile}. Hence:

\begin{theoremrep}\label{thm:simplicial}
  For any $\ki, \kq \in \NN$, given a CQ $Q$ and a simplicial decomposition $T$
  of simplicial width $\kq$
  of~$Q$,
  we can compute in FPT-linear in~$\card{Q}$ (parameterized by~$\ki$ and~$\kq$)
  an alternating two-way tree automaton that tests
  $Q$ for treewidth~$\ki$.

  Hence, if we are additionally given a relational instance~$I$ of treewidth
  $\leq \ki$,
  one can determine
  whether $I\models Q$ in FPT-linear time in $|I|\cdot(|Q|+|T|)$ 
  (parameterized by~$\ki$ and~$\kq$).
\end{theoremrep}

\begin{proof}
  We use Proposition~\ref{prp:simplicial} to compile
  the CQ $Q$ to an
  ICG-Datalog program $P$ with body size at most $\kp \defeq f_\sigma(\kq)$,
  in FPT-linear time in $\card{Q} + \card{T}$ parameterized by~$\kq$. 
  
  We now use Theorem~\ref{thm:maintheorem} to construct, in FPT-linear time in
  $\card{P}$ (hence, in $\card{Q}$), parameterized by $\ki$ and~$\kp$, hence
  in~$\ki$ and~$\kq$, a SATWA $A$ testing $P$ for treewidth $\ki$. For the same
  reasons as in the proof of Proposition~\ref{prp:alphaatwa}, it is actually an
  two-way alternating tree automaton, so we have shown the first part of the result.

  To prove the second part of the result,
  we now use Theorem~\ref{thm:main} to evaluate $P$ on $I$ in FPT-linear time
  in $\card{I} \cdot \card{P}$, parameterized by~$\kp$ and~$\ki$, hence again
  by~$\kq$ and~$\ki$. This proves the claim.
\end{proof}

Notice the technicality that the simplicial decomposition $T$ must be provided
as input to the procedure, because it is not known
to be computable in FPT-linear time, unlike tree decompositions.
While we are not aware of results on the
complexity of this specific task,
quadratic time algorithms are known for the related problem of computing the
\emph{clique-minimal separator decomposition}
\cite{leimer1993optimal,berry2010introduction}.

The intuition for the efficient compilation of bounded-simplicial-width CQs
is as follows. The \emph{interface} variables shared between any bag and its parent must
be ``clique-guarded'' (each pair is covered by an atom). Hence, consider
any subquery rooted at a bag of the query decomposition, and see it as a
non-Boolean CQ
with the interface variables as free variables. Each result of this
CQ must then be covered by a clique of facts of the instance, which ensures
\cite{gavril1974intersection} that it occurs in some bag in the instance tree decomposition and can be
``seen'' by a tree automaton.
This intuition can be generalized, beyond conjunctive queries, to design an
expressive query language featuring disjunction, negation, and fixpoints, with
the same properties of efficient compilation to automata and FPT-linear
combined complexity of evaluation on treelike instances. We introduce such a
Datalog variant in the next section.

\section{ICG-Datalog on Treelike Instances}
\label{sec:ICG}
To design a Datalog fragment with efficient compilation to automata, we must of
course impose some limitations, as we did for CQs. In fact, we can even show
that the full Datalog language (even without negation) \emph{cannot} be compiled to automata, no matter
the complexity:

\begin{propositionrep}\label{prp:dlnotftar}
  There is a signature $\sigma$ and Datalog program $P$ 
  such that the language of $\Gamma_\sigma^1$-trees that encode instances
  satisfying $P$ is not a regular tree language.
  \end{propositionrep}

\begin{toappendix}
	We will use in this proof the notion of tree encoding and the associated
        concepts, which can be found in Appendix~\ref{apx:tree-encodings}.
\end{toappendix}

\begin{proof}
Let $\sigma$ be the signature containing two binary relations $Y$ and $Z$ and
  two unary relations $\mathrm{Begin}$ and $\mathrm{End}$.
Consider the following program $P$:
\begin{align*}
  \mathrm{Goal}() &\leftarrow S(x,y), \mathrm{Begin}(x), \mathrm{End}(y)\\
                      S(x,y) &\leftarrow Y(x,w), S(w,u), Z(u,y) \\ 
                      S(x,y) &\leftarrow Y(x,w), Z(w,y)
\end{align*}
Let $L$ be the language of the tree encodings of instances of treewidth~$1$ that
satisfy~$P$. We will show that $L$ is not a regular tree language, which clearly
  implies the second claim, as a bNTA or an alternating two-way tree
  automaton can
  only recognize regular tree languages \cite{tata}. To show this, let us assume by
  contradiction that $L$ is a regular tree language, so that there exists a
  $\Gamma^1_\sigma$-bNTA $A$ that accepts $L$, i.e., that tests $P$.

We consider instances which are chains of facts which are either $Y$- or
$Z$-facts, and where the first end is the only node labeled $\mathrm{Begin}$ and
the other end is the only node labeled $\mathrm{End}$. This condition on instances
  can clearly be expressed in MSO, so that by Theorem~\ref{thm:courcelle} there exists
a bNTA $A_{\mathrm{chain}}$ on~$\Gamma_\sigma^1$ that tests this property. In particular, we can build the
bNTA $A'$ which is the intersection of $A$ and $A_{\mathrm{chain}}$, which tests whether
instances are of the prescribed form and are accepted by the program $P$.

We now observe that such instances must be the instance
\begin{align*}
  I_k = {} & \{
  \mathrm{Begin}(a_1), \allowbreak
  Y(a_1, a_2), \ldots, \allowbreak
  Y(a_{k-1}, a_k), \allowbreak
  Y(a_k, a_{k+1}), \\
  & \quad Z(a_{k+1}, a_{k+2}), \ldots, \allowbreak
  Z(a_{2k-1}, a_{2k}), \allowbreak
  Z(a_{2k}, a_{2k+1}), \allowbreak
\mathrm{End}(a_{2k+1})\}\end{align*} for
some $k \in \mathbb{N}$. Indeed, it is clear that $I_k$ satisfies $P$ for all $k
\in \mathbb{N}$, as we derive the facts \[S(a_k, a_{k+2}), S(a_{k-1}, a_{k+3}), \ldots,
S(a_{k-(k-1)}, a_{k+2+(k-1)})\text{, that is, }S(a_1, a_{2k+1}),\] and finally
$\mathrm{Goal}()$. Conversely, for any instance $I$ of the prescribed shape that
satisfies $P$, it is easily seen that the derivation of $\mathrm{Goal}$
justifies the existence of an chain in $I$ of the form~$I_k$, which by the
restrictions on the shape of $I$ means that $I = I_k$.

We further restrict our attention to tree encodings
that form a single branch of a specific form, namely, their contents are as
  follows (given from leaf to root) for some integer $n\geq 0$:
$(\{a_1\}, \mathrm{Begin}(a_1))$, $(\{a_1, a_2\}, X(a_1, a_2))$, $(\{a_2, a_3\},
X(a_2, a_3))$, $(\{a_3, a_1\}, X(a_3, a_1))$, 
\dots, $(\{a_{n}, a_{n+1}\}, X(a_n, a_{n+1}))$, $(\{a_{n+1}\},
\mathrm{End}(a_{n+1}))$, 
  where we write $X$ to mean that we may match either $Y$ or $Z$,
  where addition is modulo $3$, and where we add dummy nodes $(\bot,
\bot)$ as left children of all nodes, and as right children of the leaf node
  $(\{a_1\},\mathrm{Begin}(a_1))$,
to ensure that the tree is full.
It is clear that we can design a bNTA $A_{\mathrm{encode}}$ which recognizes
  tree encodings of this form, and we define $A''$ to be the intersection of
  $A'$ and $A_{\mathrm{encode}}$. In other words,
  $A''$ further enforces that the $\Gamma^1_\sigma$-tree
encodes the input instance as a chain of consecutive facts with a certain
prescribed alternation pattern for elements, with the $\mathrm{Begin}$ end of
the chain at the top and the $\mathrm{End}$ end at the bottom.

Now, it is easily seen that there is exactly one tree encoding of every $I_k$ which is
accepted by~$A''$, namely, the one of the form tested by $A_{\mathrm{encode}}$
where $n=2k$, the first $k$ $X$ are matched to $Y$ and the last $k$ $X$
are matched to~$Z$. 

Now, we observe that as $A''$ is a bNTA which is forced to operate on
chains (completed to full binary trees by a specific
addition of binary nodes). Thus, we can translate it to a deterministic automaton
$A'''$ on words on the alphabet $\Sigma = \{B, Y, Z, E\}$,
by looking at
its behavior in terms of the $X$-facts. Formally, $A'''$ has same state
space as $A''$, same final states, initial state
$\delta(\iota((\bot, \bot)), \iota((\bot, \bot)))$ and
transition function $\delta(q, x) = \delta(\iota((\bot, \bot)), q,
(s, f))$ for every domain $s$, where $f$ is a fact corresponding to the letter $x
\in \Sigma$ ($B$ stands here for $\mathrm{Begin}$, and $E$ for
$\mathrm{End}$).
By definition of $A''$, the automaton $A'''$ on words recognizes the language $\{BY^kZ^kE \mid k \in
\mathbb{N}\}$. As this language is not regular, we have reached a contradiction.
This contradicts our hypothesis about the existence of automaton $A$, which
establishes the desired result.
\end{proof}

Hence, there is no bNTA or alternating two-way tree automaton that tests $P$ for
  treewidth~$1$.
To work around this problem and ensure that compilation is possible and efficient,
the key condition that we impose on Datalog programs, pursuant to the
intuition of simplicial decompositions, is that
intensional predicates in rule bodies must be \emph{clique-guarded}, i.e., their
variables must co-occur in \emph{extensional} predicates of the rule body. 
We can then use the \emph{body size} of the program rules as a parameter,
and will show that the fragment can then be compiled to automata in FPT-linear
time.

\begin{definition}
  \label{def:ICG}
  Let $P$ be a stratified Datalog program.
  An intensional literal $A(\mathbf{x})$ or $\lnot A(\mathbf{x})$ in a 
  rule body~$\psi$ of~$P$
  is \emph{clique-guarded} if, for any two variables
  $x_i \neq x_j$ of~$\mathbf{x}$, $x_i$ and $x_j$
  co-occur in some extensional atom of~$\psi$.
  $P$ is \emph{intensional-clique-guarded} (ICG) if, for any rule
  $R(\mathbf{x}) \leftarrow \psi(\mathbf{x}, \mathbf{y})$, every
  \emph{intensional} literal in $\psi$ is clique-guarded in
  $\psi$.
  The \emph{body size} of~$P$ is the maximal number of atoms in
  the body of its rules, multiplied by its arity.
\end{definition}

The main result of this paper is that evaluation of ICG-Datalog is
\emph{FPT-linear} in \emph{combined complexity}, when parameterized by the body
size of the program and the instance treewidth.

\begin{theoremrep}
  \label{thm:main}
  Given an ICG-Datalog program~$P$ of body size
  $\kp$ and a relational instance~$I$ of treewidth $\ki$, checking
  if $I\models P$ is FPT-linear time in $|I|\cdot|P|$ 
  (parameterized by~$\kp$ and~$\ki$).
\end{theoremrep}

\begin{proof}
  Anticipating on the results of later sections, we use
  Theorem~\ref{thm:mainprov} to compute a representation $C$ of the provenance
  of~$P$ on~$I$ as a stratified cycluit, in FPT-linear time (parameterized by
  $\kp$ and $\ki$). We let $\nu$ be the valuation of~$C$ that sets every input
  to true (reflecting that all the facts of~$I$ are indeed there), and we then
  use Proposition~\ref{prp:cycluitlinear} to evaluate $\nu(C)$ in linear time.
  By definition of the provenance, we have $\nu(C) = 1$ iff $I \models P$, which
  concludes the proof.
\end{proof}

We will show this result in the next section by compiling ICG-Datalog programs in FPT-linear time to
a special kind of tree automata (Theorem~\ref{thm:maintheorem}), and showing
in Section~\ref{sec:provenance} that we can efficiently evaluate such automata
and even compute \emph{provenance representations}.
The rest of this section
presents consequences of our main result for various languages.

\subparagraph*{Conjunctive queries.}
Our tractability result for
bounded-simplicial-width CQs (Theorem~\ref{thm:simplicial}), including $\alpha$-acyclic
CQs, is shown by rewriting
to ICG-Datalog of bounded body size:

\begin{propositionrep}\label{prp:simplicial}
  There is a function $f_\sigma$ (depending only on~$\sigma$)
  such that for all $k \in \mathbb{N}$,
  for any conjunctive query $Q$ and simplicial tree decomposition $T$ 
  of $Q$ of
  width at most $k$, we can compute in $O(\card{Q} + \card{T})$ an equivalent
  ICG-Datalog program with body size at most $f_{\sigma}(k)$.
\end{propositionrep}

\begin{toappendix}
  We first prove the following lemma about simplicial tree decompositions:

  \begin{lemma}
    \label{lem:rewritesimplicial}
    For any simplicial decomposition $T$ of width~$k$ of a query $Q$, we can compute in
    linear time a simplicial decomposition $T_{\mathrm{bounded}}$ of $Q$ such
    that each bag has degree at most~$2^{k+1}$.
  \end{lemma}

  \begin{proof}
	Fix $Q$ and $T$.
        We construct the simplicial decomposition $T_{\mathrm{bounded}}$ of $Q$
        in a process which shares some similarity with the routine rewriting of
        tree decompositions to make them binary, by creating copies of bags.
        However, the process is more intricate because we need to preserve the fact that
        we have a \emph{simplicial} tree decomposition, where interfaces are
        guarded.
        
        We go over $T$ bottom-up: for each bag $b$ of $T$, we create a bag
        $b'$ of $T_{\mathrm{bounded}}$ with same domain as $b$.
        Now, we partition the children of $b$ depending on their intersection
        with $b$: for every subset $S$ of the domain of~$b$ such that $b$
        has some children whose intersection with $b$ is equal to $S$, we
        write these children
        $b_{S,1},\ldots,b_{S,n_S}$ (so we have $S = \dom(b) \cap
        \dom(b_{S,j})$ for all $1 \leq j \leq n_S$), and we write $b'_{S,1},
        \ldots, b'_{S, n_S}$ for the copies that we already created for these bags in
        $T_{\mathrm{bounded}}$. Now, for each~$S$, we create
        $n_S$ fresh bags $b'_{=S,j}$ in $T_{\mathrm{bounded}}$ (for $1 \leq
        j \leq n_S$) with domain equal to~$S$,
        and we set $b'_{=S,1}$ to be a child of~$b'$,
        $b'_{=S,j+1}$ to be a child of $b'_{=S,j}$ for all $1 \leq j < n_S$,
        and we set each $b'_{S,i}$ to be a child of $b'_{=S,i}$.

        This process can clearly be performed in linear time. Now, the degree of
        the fresh bags in $T_{\mathrm{bounded}}$ is at most~$2$, and the degree
        of the copies of the original bags is at most $2^{k+1}$, as stated.
        Further, it is clear that the result is still a tree decomposition (each
        fact is still covered, the occurrences of each element still form a
        connected subtree because they are as in~$T$ with the addition of some
        paths of the fresh bags), and the interfaces in~$T_{\mathrm{bounded}}$
        are the same as in~$T$, so they still satisfy the requirement of
        simplicial decompositions.
  \end{proof}

    We can now prove Proposition~\ref{prp:simplicial}. In fact, as will be easy
    to notice from the proof, our construction further ensures that the
    equivalent ICG-Datalog program is positive, nonrecursive, and
    conjunctive.

  \begin{proof}[Proof of Proposition~\ref{prp:simplicial}]
        Using Lemma~\ref{lem:rewritesimplicial}, we can start by rewriting in
        linear time the
        input simplicial decomposition to ensure that each bag has degree at
        most $2^{k+1}$. Hence, let us assume without loss of generality that $T$
        has this property. We further add an empty root bag if necessary to
        ensure that the root bag of~$T$ is empty and has exactly one child.

	We start by using Lemma~3.1 of \cite{flum2002query} to annotate in linear time each node $b$ of $T$ by
	the set of atoms $\calA_b$ of $Q$ whose free variables are in the domain of $b$ and such that 
	for each atom $A$ of $\calA_b$, $b$ is the topmost bag of $T$ which contains all the variables of $A$.
        As the signature $\sigma$ is fixed, note that we have
        $\card{\calA_b} \leq g_\sigma(k)$ for some function $g_\sigma$ depending only on~$\sigma$.

        Once this is done, we \emph{scrub} $T$, namely, we remove variables from
        bags of~$T$ to ensure that, for each bag $b$, for each $x \in \dom(b)$,
        either $\calA_b$ contains an atom where $x$ appears, or there is a
        child of~$b$ where $x$ appears. We can do so in linear time by
        traversing $T$ bottom-up, keeping track of which variables must be kept,
        and removing all other variables.
        Intuitively, this scrubbing transformation will ensure that, when constructing our
        ICG-Datalog program, all variables in the head of a rule also appear in
        the body of this rule. 

	We now perform a process similar to Lemma~3.1 of \cite{flum2002query}. 
	We start by precomputing in linear time a mapping $\mu$ that associates, to each
        pair $\{x, y\}$ of variables of~$Q$, the set of all atoms in~$Q$ where
        $\{x, y\}$ co-occur. We can compute $\mu$ in linear time by
        processing all atoms of~$Q$ and adding each atom as an image of~$\mu$ for
        each pair of variables that it contains (remember that the arity of~$\sigma$ is
        constant). Now, we do the following computation:
        for each bag $b$ which is not the root of~$T$, letting
        $S$ be its interface with its parent bag, we annotate $b$ by a set of
        atoms $\calA^\guard_b$ defined as follows:
        for all $x,y \in S$ with $x \neq y$, letting $A(\mathbf{z})$ be an atom
        of~$Q$ where $x$ and $y$ appear (which must exist, by the requirement on
        simplicial decompositions, and which we retrieve from~$\mu$), we
        add $A(\mathbf{w})$ to $\calA^\guard_b$, where, for $1 \leq i \leq
        \card{\mathbf{z}}$, we set $w_i \defeq z_i$ if $z_i \in \{x, y\}$, and
        $w_i$ to be a fresh variable otherwise. In other words,
        $\calA^\guard_b$ is a set of atoms that ensures that the interface $S$
        of $b$ with its parent is covered by a clique, and we construct it by
        picking atoms
        of~$Q$ that witness the fact that it is guarded (which it is, because
        $T$ is a simplicial decomposition), and replacing their irrelevant
        variables to be fresh. Note that $\calA^\guard_b$
        consists of at most $k\times(k+1)/2$ atoms, but the domain of these
        atoms is not a subset of $\dom(b)$ (because they include fresh
        variables). This entire computation is performed in linear time.

        We now define the function $f_\sigma(k)$ as follows, remembering that
        $\arity{\sigma}$ denotes the arity of the \emph{extensional} signature:
        \[
          f_\sigma(k) \defeq \max(\arity{\sigma}, k+1) \times \left( g_\sigma(k)
          + 2^{k+1}(k(k+1)/2 +1) \right).
        \]

	We now build our ICG-Datalog program $P$ of body size $f_\sigma(k)$
        which is equivalent to $Q$. We define the intensional signature
        $\sigmai$ by creating one intensional predicate
        $P_b$ for each non-root bag $b$ of~$T$,
        whose arity is the size of the intersection of $b$ with its parent.
        As we
        ensured that the root bag $b_\r$ of~$T$ is empty and has exactly one
        child $b_\r'$, we use $P_{b_\r'}$ as our
        0-ary $\goal$ predicate (because its interface with its parent $b_\r$
        is necessarily empty).
        We now define
        the rules of~$P$ by
        processing $T$ bottom-up: for each bag $b$ of~$T$, we add one rule
        $\rho_b$ with head $P_b(\mathbf{x})$, defined as follows:

	\begin{itemize}
		\item If $b$ is a leaf, then
                  $\rho_b$ is $P_b \leftarrow \bigwedge \calA_b$.
		\item If $b$ is an internal node with children $b_1,\ldots,b_m$
                  (remember that $m \leq 2^{k+1}$),
			then $\rho_b$ is $P_b \leftarrow
                        \bigwedge\calA_b \land \bigwedge_{1 \leq i \leq m}
                        (\bigwedge\calA^{\text{guard}}_{b_i} \land P_{b_i} )$.
	\end{itemize}
        
        We first check that $P$ is intensional-clique-guarded, but this is the
        case because the only use of intensional atoms in rules is the
        $P_{b_i}$, whose free variables are the intersection $S$ of $b$ and
        $b_i$, but by construction the conjunction of atoms
        $\bigwedge\calA^{\guard}_{b_i}$ is a suitable guard for $S$: for each
        $\{x, y\} \in S$, it contains an atom where both $x$ and $y$ occur).

        Second, we check that, pursuant to the definition of Datalog, each head
        variable occurs in the body, but this is the case thanks to the fact
        that $T$ is scrubbed: any variable of an intensional predicate $P_b$
        is a variable of $\dom(b)$, which must occur either in $\calA_b$ or
        in the intersection of $b$ with one of its children $b_i$.

        Third, we check that the body size of $P$ is indeed
        $f_\sigma(k)$. It is clear that $\arity{\sigmai} \leq k+1$, so that
        $\arity{P} \leq \max(\arity{\sigma}, k+1)$. Further, 
        the maximal number
	of atoms in the body of a rule is $g_\sigma(k) + 2^{k+1}(k(k+1)/2 +1)$),
        so we obtain the desired bound.

	What is left to check is that $P$ is equivalent to $Q$. It will be
        helpful to reason about $P$ by seeing it as the conjunctive query $Q'$
        obtained by recursively inlining the definition of rules: observe that
        this a conjunctive query, because $P$ is conjunctive, i.e., for each
        intensional atom $P_b$, the rule $\rho_b$ is the only one where
        $P_b$ occurs as head atom. It is clear that $P$ and~$Q'$ are
        equivalent, so we must prove that $Q$ and $Q'$ are equivalent.

        For the forward direction, it is obvious that $Q' \implies Q$, because
        $Q'$ contains every atom of~$Q$ by construction of the $\calA_b$. For
        the backward direction, noting that the only atoms of $Q'$ that are not
        in $Q$ are those added in the sets $\calA^\guard_b$, we observe that
        there is a homomorphism from $Q'$ to $Q$ defined by mapping each atom
        $A(\mathbf{w})$ occurring in some $\calA^\guard_b$ to the atom
        $A(\mathbf{z})$ of~$Q$ used to create it;
        this mapping is the identity on the two variables $x$ and $y$ used to
        create $A(\mathbf{w})$, and maps each fresh variables $w_i$ to
        $z_i$: the fact that these variables are fresh ensures that this
        homomorphism is well-defined. This shows $Q$ and $Q'$, hence $P$, to be
        equivalent, which concludes the proof.
\end{proof}
\end{toappendix}

This implies that ICG-Datalog can express any CQ up to
increasing the body size parameter, unlike, e.g., $\mu$CGF.
Conversely, we can show that bounded-simplicial-width CQs
\emph{characterize}
the queries expressible in ICG-Datalog when disallowing negation, recursion
and disjunction.  Specifically, 
  a Datalog program is \emph{positive} if it contains no negated
  atoms. It is \emph{nonrecursive} if there is no cycle in the directed graph on
  $\sigmai$ having an edge from $R$ to $S$ whenever a rule contains $R$ in its
  head and $S$ in its body.
  It is \emph{conjunctive}~\cite{benedikt2010impact} if each intensional relation
  $R$ occurs in the head of at most one rule.
  We can then show:

\begin{propositionrep}\label{prp:onlysimplicial}
  For any positive, conjunctive, nonrecursive ICG-Datalog program~$P$
  with body size $k$,
  there is a CQ $Q$ of
  simplicial width $\leq k$ that is equivalent to~$P$.
\end{propositionrep}

\begin{toappendix}
  To prove Proposition~\ref{prp:onlysimplicial}, we will use the notion of the
  \emph{call graph} of a Datalog program. This is the graph $G$ on the relations of
  $\sigmai$ which has an edge from $R$ to $S$ whenever a rule contains relation
  $R$ in its head and $S$ in its body. From the requirement that $P$ is
  nonrecursive,
  we know that this graph $G$ is a DAG.

  We now prove Proposition~\ref{prp:onlysimplicial}:
\end{toappendix}

\begin{toappendix}
  \begin{proof}[Proof of~\ref{prp:onlysimplicial}]
  We first check that every intensional relation reachable from $\text{Goal}$ in
  the call graph~$G$ of~$P$ appears in the head of a rule of $P$ (as $P$ is
    conjunctive, this
  rule is then unique). Otherwise, it is clear that $P$ is not satisfiable (it
  has no derivation tree), so we can simply rewrite $P$ to the query False.
  We also assume without loss of generality that each intensional relation
  except $\goal$ occurs in the body of some rule, as otherwise we can simply
  drop them and all rules where they appear as the head relation.

  In the rest of the proof we will consider the rules of~$P$ in some order, and
    create an equivalent ICG-Datalog program $P'$ 
        with rules
  $r'_0, \ldots, r'_{m'}$. We will ensure that $P'$ is also positive,
    conjunctive, and
  nonrecursive, and that it further satisfies the following additional properties:
  \begin{enumerate}
	  \item Every intensional relation other than $\text{Goal}$ appears in
            the body of exactly one rule of~$P'$, and appears there exactly
            once;
	  \item For every $0 \leq i \leq m$, for every variable $z$ in the body
            of rule $r'_i$ that does not occur in its head,
		  then for every $0 \leq j < i$, $z$ does not occur in $r'_j$.
  \end{enumerate}

    We initialize a queue that contains only the one rule that defines
    $\text{Goal}$ in~$P$, and we do the following until the queue is empty:
  \begin{itemize}
	  \item Pop a rule $r$ from the queue.
		  Let $r'$ be defined from $r$ as follows:
                  for every intensional relation $R$ that occurs in the body of
                  $r'$, letting $R(\mathbf{x^1}), \ldots, R(\mathbf{x^n})$ be
                  its occurrences, rewrite these atoms to $R^1(\mathbf{x^1}),
                  \ldots, R^n(\mathbf{x^n})$, where the $R^i$ are \emph{fresh}
                  intensional relations.
	  \item Add $r'$ to $P'$.
	  \item For each intensional atom $R^i(\mathbf{x})$ of $r'$, letting $R$
            be the relation from which $R^i$ was created, 
            let $r_R$ be the rule of $P$ that has $R$ in its head (by our
            initial considerations, there is one such rule, and as the program
            is
	    conjunctive there is exactly one such rule). Define
            $r'_{R^i}$ from~$r_R$ by replacing its head relation from $R$ to $R^i$, and renaming
            its head and body variables such that the head is exactly
            $R^i(\mathbf{x})$. Further rename all variables that occur in the
            body but not in the head, to replace them by fresh new variables.
            Add $r'_{R^i}$ to the queue.
  \end{itemize}

  We first argue that this process terminates. Indeed, considering the graph
  $G$, whenever we pop from the queue a rule with head relation $R$ (or a fresh
  relation created from a relation~$R$), we add to the queue a finite number of
  rules for head relations created from relations $R'$ such that the edge $(R,
  R')$ is in the graph $G$. The fact that $G$ is acyclic ensures that the
  process terminates (but note that its running time may generally be
  exponential in the input).
  Second, we observe that, by construction, $P$ satisfies the first property,
  because each occurrence of an intensional relation in a body of $P'$ is fresh,
  and satisfies the second property, because each variable which is the body of
  a rule but not in its head is fresh, so it cannot occur in a previous rule

  Last, we verify that $P$ and $P'$ are equivalent, but this is immediate,
  because any derivation tree for $P$ can be rewritten to a derivation tree for $P'$ (by
  renaming relations and variables), and vice-versa.

  We define $Q$ to be the conjunction of all extensional atoms occurring
  in~$P'$. To show that it is equivalent to~$P'$, the fact that $Q$ implies $P'$
  is immediate as the leaves are sufficient to construct a derivation tree, and the
  fact that $P'$ implies $Q$ is because, 
  letting $G'$ be the call graph of~$P'$, by the first property of~$P'$ we can easily observe that it is a
  tree, so the structure of derivation trees of $G'$ also corresponds to $P$, and by
  the second property of~$P'$ we know that two variables are equal in two extensional
  atoms iff they have to be equal in any derivation tree. Hence, $P'$ and $Q$ are
  indeed equivalent.

 We now justify that $Q$ has simplicial width at most $k$.
 We do so by building from $P'$ a simplicial decomposition $T$ of $Q$ of width $\leq k$. 
 The structure of $T$ is the same as $G'$ (which is actually a tree). For each
 bag $b$ of $T$ corresponding to a node of $G'$ standing for a rule $r$ of $P'$, we
 set the domain of $b$ to be the variables occurring in~$r$. It is clear that
 $T$ is a tree decomposition of~$Q$, because each atom of~$Q$ is covered by a
 bag of~$T$ (namely, the one for the rule whose body contained that atom) and the
 occurrences of each variable form a connected subtree (whose root is the node
 of $G'$ standing for the rule where it was introduced, using the second
 condition of~$P'$). Further, $T$ is a simplicial decomposition because 
  $P'$ is intensional-clique-guarded; further, from the second condition, the
  variables shared between one bag and its child are precisely the head
  variables of the child rule. The width is $\leq k$ because the
  body size of an ICG-Datalog program is an upper bound on the maximal number of
  variables in a rule body.
\end{proof}
\end{toappendix}

However, our ICG-Datalog fragment is still exponentially more \emph{concise}
than such CQs:

\begin{propositionrep}\label{prp:simplicialconcise}
  There is a signature $\sigma$ and a family $(P_n)_{n\in\NN}$ of ICG-Datalog
  programs with body size at most $9$ 
  which are positive, conjunctive, and nonrecursive, such that
  $|P_n|=O(n)$
  and any CQ $Q_n$ equivalent to $P_n$ has size $\Omega(2^n)$.
\end{propositionrep}

\begin{toappendix}
  To prove Proposition~\ref{prp:simplicialconcise}, we will introduce the following notion:

  \begin{definition}
    \label{def:match}
  A \emph{match} of a CQ $Q$
  in an instance $I$
  is a subset $M$ of facts of~$I$ which is an image of a homomorphism from
  the canonical instance of $Q$ to~$I$, i.e., $M$ witnesses that $I \models Q$,
  in particular $M \models Q$ as a subset of~$I$.
  \end{definition}

  Our proof will rely on the following elementary observation:
  \begin{lemma}
    \label{lem:match}
  If a CQ $Q$ has a match $M$
  in an instance $I$, then necessarily $\card{Q} \geq \card{M}$.
  \end{lemma}
  
  \begin{proof}
    As $M$ is
  the image of $Q$ by a homomorphism, it cannot have more atoms than~$M$
    has facts.
  \end{proof}

  We are now ready to prove Proposition~\ref{prp:simplicialconcise}:

  \begin{proof}[Proof of Proposition~\ref{prp:simplicialconcise}]
  Fix $\sigma$ to contain a binary relation $R$ and a ternary relation $G$.
  Consider the rule $\rho_0: R_0(x, y) \leftarrow R(x, y)$ and define
  the following rules, for all $i > 0$:
  \[
    \rho_i: R_i(x, y) \leftarrow G(x, z, y), R_{i-1}(x, z), R_{i-1}(z, y)
  \]
  For each $i > 0$, we let $P_i$ consist of the rules $\rho_j$ for $1 \leq j
  \leq i$, as well as $\rho_0$ and the rule $\text{Goal}() \leftarrow R_i(x, y)$.
  It is clear that each $P_i$ is positive, conjunctive, and
  non-recursive; further, the predicate $G$ ensures that it is an ICG-Datalog
  program. The arity is~$3$ and the maximum number of atoms is the body is~$3$,
  so the body size is indeed~$9$.

  We first prove by an immediate induction that, for each $i \geq 0$, considering
  the rules of $P_i$ and the intensional predicate $R_i$, whenever an instance $I$
  satisfies $R_i(a, b)$ for two elements $a, b \in \dom(I)$ then there is an
  $R$-path of length $2^{i}$ from~$a$ to~$b$.
  Now, fixing $i \geq 0$, this
  clearly implies there is
  an instance $I_i$
  of size (number of
  facts) $\geq 2^{i}$, namely, an $R$-path of this length
  with the right set of additional $G$-facts,
  such that $I_i \models P_i$ but any strict subset of~$I_i$ does not satisfy
  $P_i$.

  Now, let us consider a CQ $Q_i$ which is equivalent to $P_i$, and let us show
  the desired size bound. By equivalence, we know that $I_i \models Q_i$, hence
  $Q_i$ has a match $M_i$ in $I_i$, but any strict subset of~$I_i$ does not
  satisfy $Q_i$, which implies that, necessarily, $M_i = I_i$ (indeed,
  otherwise $M_i$ would survive as a match in some strict subset of~$I_i$).
  Now, by Lemma~\ref{lem:match}, we deduce that $\card{Q_i} \geq \card{M_i}$,
  and as $\card{M_i} = \card{I_i} \geq 2^{i}$, we obtain the desired size
  bound, which concludes the proof.
\end{proof}
\end{toappendix}

\subparagraph*{Guarded negation fragments.}
Having explained the connections between ICG-Datalog and CQs, we now study its
connections to the more expressive languages of guarded logics, specifically,
the \emph{guarded negation fragment} (GNF), a fragment of first-order logic
\cite{barany2015guarded}. Indeed, when putting GNF formulae in 
\emph{GN-normal form} \cite{barany2015guarded} or even \emph{weak GN-normal form}
\cite{benedikt2014effective}, we can translate them to
ICG-Datalog, and we can use the \emph{CQ-rank}
parameter \cite{benedikt2014effective} (that measures the maximal number of
atoms in conjunctions) to control the body size parameter.

\begin{propositionrep}\label{prp:gnf-to-icg2}
  There is a function $f_\sigma$ (depending only on~$\sigma$) such that, for any
  weak GN-normal form GNF query $Q$ of CQ-rank $r$, we can compute in time $O(\card{Q})$ an
  equivalent nonrecursive ICG-Datalog program $P$ of body size $f_\sigma(r)$.
\end{propositionrep}

\begin{proof}
  We recall from \cite{benedikt2014effective}, Appendix~B.1, that a weak
  GN-normal form formulae is a $\phi$-formula in the inductive definition below:
  
  \begin{itemize}
    \item A disjunction of existentially quantified conjunctions of
      $\psi$-formulae is a $\phi$-formula;
    \item An atom is a $\psi$-formula;
    \item The conjunction of a $\phi$-formula and of a guard is a
      $\psi$-formula;
    \item The conjunction of the negation of a $\phi$-formula and of a guard is
      a $\psi$-formula.
  \end{itemize}

  We define $f_\sigma : n \mapsto \arity{\sigma} \times n$.

  We consider an input Boolean GN-normal form formula $Q$ of CQ-rank~$r$, and call $T$ its
  abstract syntax tree. We rewrite $T$ in linear time to inline in
  $\phi$-formulae the definition of their $\psi$-formulae, so all nodes of~$T$
  consist of $\phi$-formulae, in which all subformulae are guarded (but they can
  be used positively or negatively).
 
  We now process $T$ bottom-up. We introduce one intensional Datalog predicate $R_n$ per node $n$ in~$T$:
  its arity is the number of variables that are free at~$n$.
  We then introduce one rule $\rho_{n,\delta}$ for each disjunct $\delta$ of the disjunction that defines
  $n$ in~$T$: the head of $\rho_{n, \delta}$ is an $R_n$-atom whose free variables are the variables that are
  free in~$n$, and the
  body of $\rho_{n, \delta}$ is the conjunction that defines~$\delta$, with each subformula replaced by the
  intensional relation that codes it. Of course, we use the predicate $R_r$ for
  the root $r$ of~$T$ as our goal predicate; note that it must be $0$-ary, as
  $Q$ is Boolean so there are no free variables at the root of~$T$. This process defines
  our ICG-Datalog program~$P$: it is clear that this process runs in linear
  time.
 
  We first observe that body size for an intensional predicate $R_n$
  is less than the CQ-rank of the corresponding subformula: recall that the
  \emph{CQ-rank} is the overall
  number of conjuncts occurring in the disjunction of existentially quantified
  conjunctions that defines this subformula. Hence, as the arity of~$\sigma$ is
  bounded, clearly~$P$ has body size $\leq
  f_\sigma(r)$. We next observe that intentional predicates in the bodies of
  rules of~$P$ are
  always guarded, thanks to the guardedness requirement on~$Q$.
  Further, it is obvious that $P$ is nonrecursive, as it is
  computed from the abstract syntax tree~$T$. Last, it is clear that $P$ is equivalent to the
  original formula $Q$, as we can obtain $Q$ back simply by inlining the
  definition of the intensional predicates.
\end{proof}

In fact, the efficient compilation of bounded-CQ-rank normal-form GNF programs
(using the fact that subformulae are ``answer-guarded'', like our guardedness
requirements)
has been used recently (e.g., in \cite{benedikt2016step}),
to give efficient procedures for GNF \emph{satisfiability}, compiling them to
automata (to a treewidth which is not fixed, unlike in our context, but
depends on the formula). ICG-Datalog further
allows clique guards (similar to CGNFO \cite{barany2015guarded}), can reuse
subformulae (similar to the idea of DAG-representations in
\cite{benedikt2014effective}), and supports recursion (similar to
GNFP \cite{barany2015guarded}, or GN-Datalog \cite{barany2012queries} but
whose combined complexity is intractable --- P$^\text{NP}$-complete). 
ICG-Datalog also resembles $\mu$CGF \cite{berwanger2001games}, but remember that
it is not a guarded \emph{negation} logic, so, e.g., $\mu$CGF cannot express all CQs.

Hence, the design of ICG-Datalog, and its compilation to automata, has
similarities with guarded logics. However, to our knowledge, the idea of applying
it to query evaluation is new, and ICG-Datalog is designed to support all
relevant features to capture interesting query languages (e.g.,
clique guards are necessary to capture bounded-simplicial-width queries).

\subparagraph*{Recursive languages.}
The use of fixpoints in ICG-Datalog, in particular, allows us
to capture the combined tractability of interesting recursive languages. First,
observe that our
guardedness requirement becomes trivial when all intensional predicates are
monadic (arity-one), so our main result implies that 
\emph{monadic Datalog} of bounded body size is tractable in combined
complexity on treelike instances. This is reminiscent of the results of
\cite{gottlob2010monadic}:

\begin{propositionrep}\label{prp:monadicdl}
  The combined complexity of monadic Datalog query evaluation on
  bounded-treewidth instances is FPT when parameterized by instance treewidth
  and body size (as in Definition~\ref{def:ICG}) of the monadic Datalog program.
\end{propositionrep}

\begin{proof}
  This is simply by observing that any monadic Datalog program is an ICG-Datalog
  program with the same body size, so we can simply apply
  Theorem~\ref{thm:main}.
\end{proof}

Second, ICG-Datalog can capture \emph{two-way regular path 
queries} (2RPQs) \cite{calvanese2000containment,barcelo2013querying}, a 
well-known query language in the context of  
graph databases and knowledge bases:

\begin{definition}
  We assume that the signature $\sigma$ contains only binary relations.
  A \emph{regular path query} (RPQ) $Q_L$ is defined by a regular language $L$ on the
  alphabet $\Sigma$ of the relation symbols of~$\sigma$. Its semantics is that
  $Q_L$ has two free variables $x$ and $y$, and $Q_L(a, b)$ holds on an instance
  $I$ for $a, b \in \dom(I)$ precisely when there is a directed path $\pi$ 
  of relations of~$\sigma$ 
  from $a$ to
  $b$ such that the label of~$\pi$ is in $L$. A \emph{two-way regular
  path query} (2RPQ) is an RPQ on the alphabet $\Sigma^\pm \defeq \Sigma \sqcup
  \{R^- \mid R \in \Sigma\}$, which holds whenever 
  there is a path from~$a$ to~$b$ with label in~$L$, with $R^-$ meaning that we traverse an
  $R$-fact in the reverse direction. 
  A \emph{Boolean 2RPQ} is a 2RPQ which is existentially quantified on its two free variables.
\end{definition}

\begin{proposition}[\cite{mendelzon1989finding,barcelo2013querying}]
  \label{prp:2rpq-to-icg}
  2RPQ query evaluation (on arbitrary instances) has linear time combined complexity.
  \end{proposition}

ICG-Datalog allows us to capture this result for Boolean 2RPQs on treelike
instances. In fact, the above result extends to SAC2RPQs, which are trees of 2RPQs with no multi-edges or
loops. We can prove the following result, for Boolean 2RPQs and SAC2RPQs, which 
further implies
compilability to automata (and efficient compilation of provenance
representations). We do not know whether this extends to the more general classes studied 
in~\cite{barcelo2014does}.

\begin{propositionrep}\label{prp:rpqcompile}
  Given a Boolean SAC2RPQ $Q$, we can compute in time
  $O(\card{Q})$ an equivalent ICG-Datalog program $P$ of body size~$4$.
\end{propositionrep}

\begin{proof}
  We first show the result for 2RPQs, and then explain how to extend it to
  SAC2RPQs.

  We have not specified how RPQs are provided as input. We assume that they are
  provided as a regular expression, from which we can use Thompson's
  construction~\cite{aho1986compilers}
  to compute in linear time an equivalent NFA $A$ (with $\epsilon$-transitions) on
  the alphabet $\Sigma^\pm$. Note that the result of Thompson's
  construction has
  exactly one final state, so we may assume that each automaton has exactly one
  final state.

  We now define the intensional signature of the ICG-Datalog program to consist
  of one unary predicate $P_q$ for each state $q$ of the automaton, in addition
  to $\text{Goal}()$. We add the rule $\text{Goal}() \leftarrow P_{q_\f}(x)$ for the
  final state $q_\f$, and for each extensional relation $R(x, y)$, 
  we add the rules $P_{q_0}(x)
  \leftarrow R(x, y)$ and $P_{q_0}(y)
  \leftarrow R(x, y)$, where $q_0$ is the initial state. We then add
  rules corresponding to automaton transitions:
  \begin{itemize}
    \item for each transition from $q$ to $q'$ labeled with a letter
      $R$, we add the rule $P_{q'}(y) \leftarrow P_q(x), R(x, y)$;
    \item for each transition from $q$ to $q'$ labeled with a negative letter
      $R^-$, we add the rule $P_{q'}(y) \leftarrow P_q(x), R(y, x)$;
    \item for each $\epsilon$-transition from $q$ to $q'$ we add the rule
      $P_{q'}(x) \leftarrow P_q(x)$
  \end{itemize}

  This transformation is clearly in linear time,
  and the result clearly satisfies the desired body size
  bound. Further, as the result is a monadic Datalog program, it is clearly an
  ICG-Datalog program. Now, it is clear that, in any instance~$I$ where $Q$
  holds, from two witnessing elements $a$ and $b$ and a path $\pi: a = c_0, c_1, \ldots,
  c_n = b$ from $a$ to $b$ satisfying $Q$, we can build a derivation tree of the
  Datalog program by deriving $P_{q_0}(a), P_{q_1}(c_1), \ldots, P_{q_n}(c_n)$,
  where $q_0$ is the initial state and $q_n$ is final, to match the accepting
  path in the automaton $A$ that witnesses that $\pi$ is a match of~$Q$.
  Conversely, any derivation tree of the Datalog program $P$ that witnesses that an
  instance satisfies $P$ can clearly be used to extract a path of relations
  which corresponds to an accepting run in the automaton.

  \medskip

  We now extend this argument to SAC2RPQs. Recall from
  \cite{barcelo2013querying} that a C2RPQ is a conjunction of 2RPQs, i.e.,
  writing a 2RPQ as $Q(x, y)$ with its two free variables, a C2RPQ is a CQ
  built on RPQs. An AC2RPQ is a C2RPQ where the undirected graph on variables
  defined by co-occurrence between variables is acyclic, and a SAC2RPQ further
  imposes that there are no loops (i.e., atoms of the C2RPQ of the form $Q(x,
  x)$) and no multiedges (i.e., for each variable pair, there is at most one
  atom where it occurs).

  We will also make a preliminary observation on ICG-Datalog programs: any rule
  of the form (*) $A(x) \leftarrow A_1(x), \ldots, A_n(x)$, where $A$ and each
  $A_i$ is a unary atom, can be rewritten in linear time to rules with bounded
  body size, by creating unary intensional predicates $A_i'$ for $1 \leq i \leq
  n$, writing the rule
  $A_n'(x) \leftarrow A_n(x)$, writing the rule $A_i'(x) \leftarrow A_{i+1}'(x), A_i(x)$ for each
  $1 \leq i < n$, and writing the rule $A(x) \leftarrow A_1'(x)$. Hence, we will write rules of
  the form (*) in the transformation, with unbounded body size, being understood
  that we can finish the process by rewriting out each rule of this form to
  rules of bounded body size.

  Given a SAC2RPQ $Q$, we compute in linear time the undirected graph $G$ on
  variables, and its connected components. Clearly we can rewrite each connected
  component separately, by defining one $\text{Goal}_i()$ 0-ary predicate for
  each connected component $i$, and adding the rule $\text{Goal}() \leftarrow
  \text{Goal}_1(), \ldots, \text{Goal}_n()$: this is a rule of form (*), which we can
  rewrite. Hence, it suffices to consider each connected component separately.

  Hence, assuming that the graph $G$ is connected, we root it at an arbitrary
  vertex to obtain a tree $T$. For each node $n$ of~$T$ (corresponding to a variable of the
  SAC2RPQ), we define a unary intensional predicate $P'_n$ which will intuitively
  hold on elements where there is a match of the sub-SAC2RPQ defined by the
  subtree of~$T$ rooted at~$n$, and one unary intensional predicate $P''_{n,n'}$ for
  all non-root~$n$ and children $n'$ of~$n$ in~$T$ which will hold whenever
  there is a match of the sub-SAC2RPQ rooted at $n$ which removes all children
  of~$n$ except~$n'$. Of course we add the rule $\text{Goal}()
  \leftarrow P'_{n_\r}(x)$, where $n_\r$ is the root of~$T$.
  
  Now, we rewrite the SAC2RPQ to monadic Datalog
  by rewriting each edge of $T$ independently, as in the argument for 2RPQs
  above. Specifically, we assume that the edge when read from bottom to top
  corresponds to a 2RPQ; otherwise, if the edge is oriented in the wrong
  direction, we can clearly compute an automaton for the
  reverse language in linear time from the Thompson automaton, by reversing the
  direction of transitions in the automaton, and swapping the initial state and
  the final state. We modify the
  previous construction by replacing the rule for the initial state $P_{q_0}$ by
  $P_{q_0}(x) \leftarrow P'_{n'}(x)$ where $n'$ is the lower node of the 
  edge that we are rewriting, and the rule for the goal predicate in the head is replaced by a rule
  $P''_{n,n'}(x) \leftarrow P_{q_\f}(x)$, where $n$ is the upper node of the
  edge, and $q_\f$ is the final state of the
  automaton for the edge: this is the rule that defines the $P''_{n,n'}$.
  
  Now, we define each $P'_n$ as follows:

  \begin{itemize}
    \item If $n$ is a leaf node of~$T$, we define $P'_n$ by the same rules that
      we used to define $P_{q_0}$ in the previous construction, so that $P'_n$ holds
      of all elements in the active domain of an input instance.
    \item If $n$ is an internal node of~$T$, we define $P'_n(x) \leftarrow
      P''_{n,n_1}(x), \ldots, P''_{n,n_m}(x)$, where $n_1, \ldots, n_m$ are the
      children of~$n$ in~$T$: this is a rule of form (*).
  \end{itemize}

  Now, given an instance $I$ satisfying the SAC2RPQ, from a match of the SAC2RPQ
  as a rooted tree of paths, it is easy to see by bottom-up induction on the
  tree that we derive $P_v$ with the desired semantics, using the correctness of
  the rewriting of each edge. Conversely, a derivation tree for the rewriting can be
  used to obtain a rooted tree of paths with the correct structure where each
  path satisfies the RPQ corresponding to this edge.
\end{proof}

\section{Compilation to Automata}
\label{sec:compilation}
\begin{toappendix}
  \subsection{Details on Tree Encodings}
  \label{apx:tree-encodings}
 We first explain how we encode and decode structures of bounded treewidth to trees whose alphabet size depends only on the treewidth bound and on the signature.
 Having fixed a signature $\sigma$ and a treewidth $k \in \NN$, 
we define a domain $\calD_k = \{a_1, \ldots, a_{2k+2}\}$ and
a finite alphabet $\Gamma^k_\sigma$ whose elements are pairs $(d, s)$, with $d$
being a subset of up to $k+1$ elements of $\calD_k$, and $s$ being either the
empty set or an instance consisting of a single $\sigma$-fact over some subset
of~$d$: in the latter case, we will abuse notation and identify~$s$ with the one fact that it contains. A
\emph{$(\sigma,k)$-tree encoding} is simply a rooted, binary, ordered,
full $\Gamma^k_\sigma$-tree $\la E,\lambda\ra$; the fact that $\la
  E,\lambda\ra$
  is ordered is merely for technical convenience when running bNTAs, but it is otherwise inessential.
  
  Intuitively, a tree
encoding $\la E,\lambda\ra$ 
can be decoded
  (up to isomorphism)
  to an
instance $\decode{\la E,\lambda\ra}$  with the elements of $\calD$ being decoded
  to the domain elements:
each occurrence of an element $a_i \in \calD$ in an \emph{$a_i$-connected
subtree of~$E$}, i.e., a maximal connected subtree where $a_i$ appears in the
first component of each node, is decoded to a fresh element. In other words, reusing
the same $a_i$ in adjacent nodes in $\la E,\lambda\ra$ mean that they stand for the same
element, and using $a_i$ elsewhere in the tree creates a new element. It is
easy to see that $\decode{\la E,\lambda\ra}$ has treewidth $\leq k$, as a tree decomposition for it can
be constructed from $\la E,\lambda\ra$. Conversely, any instance $I$ of treewidth $\leq k$ has a
\emph{$(\sigma,k)$-encoding}, i.e., a $\Gamma^k_\sigma$-tree $\la
E,\lambda\ra$ such that $\decode{\la E,\lambda\ra}$ is $I$ up
to isomorphism: we can construct it from a tree decomposition, replicating each
bag of the decomposition to code each fact in its own node of the tree encoding.
What matters is that this process is FPT-linear for $k$,
so that we will use the following claim:

\begin{lemma}[\cite{flum2002query} (see \cite{amarilli2016leveraging} for
  our type of encodings)]
  \label{lem:getencoding}
  The problem, given an instance~$I$ of treewidth $\leq k$,
  of computing a tree encoding of $I$, is FPT-linear for~$k$.
\end{lemma}

\end{toappendix}

In this section, we study how we can compile ICG-Datalog queries on treelike
instances to tree automata, to be able to evaluate them efficiently.
As we showed with Proposition~\ref{prp:bntalower}, we need more expressive
automata than bNTAs.
Hence, we use instead the formalism of 
\emph{alternating two-way automata} \cite{tata}, i.e., automata that can
navigate in trees in any direction, and can express transitions using Boolean
formulae on states.
Specifically, we introduce for our purposes a variant of these automata, which
are \emph{stratified} (i.e., allow a form of
stratified negation), and \emph{isotropic} (i.e., no direction is privileged,
in particular order is ignored).

\begin{toappendix}
\subsection{Evaluation}\label{sec:automata-definitions}
\end{toappendix}
As in Section~\ref{sec:treelike}, 
we will define tree automata that run on \emph{$\Gamma$-trees} for some alphabet
$\Gamma$: a $\Gamma$-tree $\la T,
\lambda \ra$ is a
finite rooted ordered
tree with a labeling function $\lambda$ from the nodes of~$T$ to
$\Gamma$.
The \emph{neighborhood} $\neigh(n)$ of a node $n \in T$ is the set which
contains $n$, all children of~$n$, and the parent of~$n$ if it exists.
\subparagraph*{Stratified isotropic alternating two-way automata.}
To define the transitions of our alternating automata, we write 
$\mathscr{B}(X)$ the set of propositional formulae (not necessarily monotone)
over a set~$X$ of variables: we will assume w.l.o.g.\ that \emph{negations are only applied to variables}, 
which we can always enforce using de Morgan's laws.
A \emph{literal} is a propositional variable $x\in X$ (\emph{positive} literal) or the
negation of a propositional variable~$\lnot x$ (\emph{negative} literal).

A \emph{satisfying assignment} of $\phi
\in \mathscr{B}(X)$ consists of two \emph{disjoint} sets $P, N \subseteq X$
(for ``positive'' and ``negative'') such that $\phi$ is a tautology when
substituting the variables of $P$ with~$1$ and those of $N$ with $0$,
i.e., when we have $\nu(\phi)=1$ for every valuation $\nu$ of $X$ such that $\nu(x) = 1$ for all $x \in P$ and
$\nu(x) = 0$ for all $x \in N$.
Note that we allow satisfying assignments with $P \sqcup N
\subsetneq X$, which will be useful for our technical results.
We now define our automata:

\begin{definition}
  \label{def:satwa}
A \emph{stratified isotropic alternating two-way automata} on $\Gamma$-trees
  ($\Gamma$-SATWA) is a tuple $A=(\calQ,q_{\I}, \Delta, \strat)$ with $\calQ$ a finite set of \emph{states}, $q_{\I}$ the \emph{initial state}, 
$\Delta$ the \emph{transition function} from $\calQ \times \Gamma$ to $\mathscr{B}(\calQ)$, and 
  $\strat$ a \emph{stratification function}, i.e., a surjective function from $\calQ$
  to~$\{0,\ldots,m\}$ for some $m \in \NN$,
  such that for any $q, q' \in \calQ$ and $f \in \Gamma$,
  if $\Delta(q,f)$ contains~$q'$ as a
  positive literal (resp., negative literal), then 
  $\strat(q') \leq \strat(q)$ (resp. $\strat(q') < \strat(q)$).

  We define by induction on $0 \leq i \leq m$ an \emph{$i$-run} of~$A$
  on a $\Gamma$-tree $\la T , \lambda \ra$ as a finite tree $\la T_\r,
  \lambda_\r \ra$, with labels of the
  form $(q,w)$ or $\lnot (q,w)$ for $w\in T$ and $q \in \calQ$ with $\strat(q)
  \leq i$, by the
  following recursive definition for all $w \in T$:
  \begin{compactitem}
  \item For $q \in \calQ$ such that $\strat(q) < i$,
    the singleton tree $\la T_\r, \lambda_\r \ra$
    with one node labeled by $(q, w)$ (resp., by~$\neg (q,
    w)$) is an $i$-run if
      there is a $\strat(q)$-run of~$A$ on $\la T , \lambda \ra$ whose root is
      labeled by~$(q, w)$ (resp., if there is no such run);
  \item For $q \in \calQ$ such that $\strat(q) = i$,
    if $\Delta(q,\lambda(w))$ has a satisfying assignment $(P,N)$, if we have a $\strat(q')$-run
          $T_{q^-}$ for each $q^- \in N$ with root labeled by $\neg (q^-, w)$, and a
          $\strat(q^+)$-run $T_{q^+}$ for each $q^+ \in P$ with root labeled by $(q^+, w_{q^+})$ for
          some $w_{q^+}$ in~$\neigh(w)$, then the tree
          $\la T_\r, \lambda_\r \ra$
          whose root is labeled $(q,
          w)$ and has as children all the $T_{q^-}$ and $T_{q^+}$ is an
          $i$-run.
  \end{compactitem}
  A \emph{run} of $A$ starting in a state $q \in \calQ$ at a node $w \in T$ is a
  $m$-run whose root is labeled $(q, w)$.
  We say that $A$ \emph{accepts} $\la T, \lambda \ra$ (written $\la T, \lambda \ra \models A$) if there
  exists a run of $A$ on $\la T, \lambda \ra$ starting in
  the initial state $q_{\I}$ at the root of~$T$.
\end{definition}

Observe that the internal nodes of a run starting in some state $q$ are labeled
by states~$q'$ in the same stratum as~$q$. The leaves of the run may be labeled
by states of a strictly lower stratum or negations thereof, or by states of the
same stratum whose transition function is tautological, i.e., by some $(q', w)$
such that $\Delta(q', \lambda(w))$ has $\emptyset, \emptyset$ as a satisfying
assignment. Intuitively, if we disallow negation in transitions, our automata
amount to
the alternating two-way automata used by \cite{cachat2002two}, with the
simplification that they do not need parity acceptance conditions (because we
only work with finite trees), and that they are \emph{isotropic}:
the run for each positive child state of an internal node may
start indifferently on \emph{any} neighbor of~$w$ in the tree
(its parent, a child, or $w$ itself), no matter the direction. (Note, however,
that the run for negated child states must start on~$w$ itself.)

We will soon explain how the compilation of ICG-Datalog is performed, but we
first note that evaluation of $\Gamma$-SATWAs is in linear time:

\begin{propositionrep}\label{prp:satwaeval}
  For any alphabet $\Gamma$, given a $\Gamma$-tree $T$ and a $\Gamma$-SATWA $A$,
  we can determine whether $T \models A$ in time $O(\card{T} \cdot \card{A})$.
\end{propositionrep}

\begin{proof}
  We use Theorem~\ref{thm:satwaprov} to compute a provenance cycluit~$C$ of the
  SATWA (modified to be a $\overline{\Gamma}$-SATWA by simply ignoring the
  second component of the alphabet) in time $O(\card{T}\cdot\card{A})$.
  Then we conclude by evaluating the
  resulting provenance cycluit (for an arbitrary valuation of that
  circuit) in time $O(\card{C})$ using
  Proposition~\ref{prp:cycluitlinear}.

  Note that, intuitively, the fixpoint evaluation of the cycluit can be understood as a
  least fixpoint computation to determine which pairs of states and tree nodes
  (of which there are $O(\card{T} \cdot \card{A})$) are reachable.
\end{proof}

In fact, this result follows from the definition of provenance cycluits
for SATWAs in the next section, and the claim that these cycluits can be
evaluated in linear time.

\subparagraph*{Compilation.}
\begin{toappendix}
  \subsection{Compilation}
  \label{apx:compilation}
\end{toappendix}

We now give our main compilation result: we can efficiently compile any ICG-Datalog program of
bounded body size into a SATWA that \emph{tests} it (in the same sense as for
bNTAs). 
This is our main technical claim, which is proven in
Appendix~\ref{apx:compilation}.

\begin{theoremrep}\label{thm:maintheorem}
  Given an ICG-Datalog program~$P$ of body size
  $\kp$ and $\ki \in \NN$,
  we can build in FPT-linear time in~$|P|$
  (parameterized by~$\kp, \ki$)
  a SATWA $A_P$ testing $P$ for treewidth~$\ki$.
\end{theoremrep}

\begin{proofsketch}
The idea is to have, for every relational symbol $R$, states of the form 
$q_{R(\mathbf{x})}^{\nu}$, where $\nu$ is a partial valuation of $\mathbf{x}$. This will be the starting state of a run if
        it is possible to navigate the tree encoding from some starting
        node and build in this way a total valuation $\nu'$ that extends $\nu$ and
	such that $R(\nu'(\mathbf{x}))$ holds. 
	When $R$ is intensional, once $\nu'$ is total on~$\mathbf{x}$, we go into a state of the form
        $q_r^{\nu',\calA}$ where $r$ is a rule with head relation $R$ and
        $\calA$ is the set of atoms in the body of $r$ (whose size is bounded by
        the body size). This means that
        we choose a rule to prove $R(\nu'(\mathbf{x}))$. 
	The automaton can then navigate the tree encoding, build $\nu'$ and coherently partition $\calA$ so as to inductively prove the atoms of the body.
	The clique-guardedness condition ensures that, when there is a match of
        $R(\mathbf{x})$, the elements to which $\mathbf{x}$ is mapped can be found together in a bag.
	The fact that the automaton is isotropic relieves us from the syntactic
        burden of dealing with directions in the tree, as one usually has to do with alternating two-way automata.
\end{proofsketch}

\begin{toappendix}
	First, we introduce some useful notations to deal with
        valuations of variables as constants of the encoding alphabet.
        Recall that $\calD_{\ki}$ is the domain
        for treewidth $\ki$, used to define the alphabet of tree
        encodings of width~$\ki$.

	\begin{definition}
          \label{def:partialval}
          Given a tuple $\mathbf{x}$ of variables, a \emph{partial valuation} of~$\mathbf{x}$ is a function $\nu$ from~$\mathbf{x}$ to $\calD_{\ki} \sqcup \{?\}$.
		The set of \emph{undefined} variables of $\nu$ is
                $U(\nu)=\{x_j \mid \nu(x_j)=\mathord{?}\}$: we say that the variables of~$U(\nu)$ are \emph{not defined} by~$\nu$, and the other variables are \emph{defined} by~$\nu$.

		A \emph{total valuation} of $\mathbf{x}$ is a partial valuation $\nu$ of $\mathbf{x}$ such that $U(\nu) = \emptyset$.
                We say that a valuation $\nu$ \emph{extends} another
                valuation~$\nu'$ if the domain of~$\nu'$ is a superset of that
                of~$\nu$, all variables defined by $\nu$ are defined by~$\nu'$
                and are mapped to the same value.
                For $\mathbf{y} \subseteq \mathbf{y}$, we say that $\nu$ is \emph{total on~$\mathbf{y}$} if its restriction to~$\mathbf{y}$ is a total valuation.

                For any two partial valuations $\nu$ of $\mathbf{x}$ and $\nu'$
                of $\mathbf{y}'$
                if we have $\nu(x) = \nu'(x)$ for all $x$ in $(\mathbf{x} \cap
                \mathbf{y}')
                \setminus (U(\nu) \cup U(\nu'))$, we write $\nu \cup \nu'$ for
                the valuation on $\mathbf{x} \cup \mathbf{y}'$ that maps every $x$ to $\nu(x)$ or
                $\nu'(x)$ if one is defined, and to ``?'' otherwise.

		When $\nu$ is a partial valuation of $\mathbf{x}$ with $\mathbf{x} \subseteq \mathbf{x^\prime}$ 
		and we define a partial valuation $\nu^\prime$ of
                $\mathbf{x^\prime}$ with $\nu^\prime \colonequals \nu$, we mean that $\nu^\prime$
		is defined like~$\nu$ on $\mathbf{x}$ and is undefined on $\mathbf{x^\prime} \setminus \mathbf{x}$.
	\end{definition}
	\begin{definition}
                \label{def:hom}
		Let $\mathbf{x}$ and $\mathbf{y}$ be two tuples of
                variables of same arity (note that some variables
                of~$\mathbf{x}$ may be repeated, and likewise for~$\mathbf{y}$).
                Let $\nu : \mathbf{x} \to \calD_{\ki}$
		be a total valuation of~$\mathbf{x}$. We define $\mathrm{Hom}_{\mathbf{y},\mathbf{x}}(\nu)$ to be 
                the (unique) homomorphism between the tuple
                $\mathbf{y}$ and the tuple $\nu(\mathbf{x})$, if
                such a homomorphism exists; otherwise,
                $\mathrm{Hom}_{\mathbf{y},\mathbf{x}}(\nu)$ is~$\mnull$.
	\end{definition}
	
  The rest of this section proves Theorem~\ref{thm:maintheorem}
    in two steps. First, we build a SATWA
    $A'_P$ and we prove that $A'_P$~tests~$P$ for treewidth $\ki$; however,
    the construction of $A'_P$ that we present is not FPT-linear. Second, we
    explain how to modify the construction to construct an equivalent
    SATWA~$A_P$ while respecting the FPT-linear time bound.

  \subparagraph*{Construction of $\bm{A_P^\prime}$.} We construct the SATWA $A'_P$
    by describing its states and transitions.
	First, for every extensional atom $S(\mathbf{x})$ appearing in (the body) of a rule of $P$ and partial valuation $\nu$ of 
	$\mathbf{x}$, we introduce a state $q_{S(\mathbf{x})}^{\nu}$.
	This will be the starting state of a run if
        it is possible to navigate the tree encoding from some starting
        node and build this way a total valuation $\nu'$ that extends $\nu$ and
        such that $S(\nu'(\mathbf{x}))$ holds in the tree encoding, in a node
        whose domain elements that are in the image of~$\nu'$ will decode to the
        same element as they do in the node where the automaton can reach state $q_{S(\mathbf{x})}^{\nu}$.
        In doing so, one has to be careful not to leave the occurrence
        subtree of the values of the valuation (the ``allowed subtree'').
        Indeed, in a tree encoding, an element $a \in \calD_{\ki}$
        appearing in two bags that are separated by another bag not
        containing $a$ is used to encode two distinct elements of the
        original instance, rather than the same element. We now formally define the transitions
        needed to implement this.

	Let $(d,s) \in \Gamma_{\sigmae}^{\ki}$ be a symbol; we have the
        following transitions:

	\begin{itemize}
          \item If there is a $j$ such that $\nu(x_j)\neq\mathord{?}$ and $\nu(x_j) \notin d$, then $\Delta(q_{S(\mathbf{x})}^{\nu},(d,s)) \colonequals \false$.
			This is to prevent the automaton from leaving the
                        allowed subtree.
		\item Else if $\nu$ is not total, then
		$\Delta(q_{S(\mathbf{x})}^{\nu},(d,s)) \colonequals
		q_{S(\mathbf{x})}^{\nu} \lor \bigvee\limits_{a \in d, x_j \in U(\nu)} q_{S(\mathbf{x})}^{\nu \cup \{x_j\mapsto a\}}$.
		That is, either we continue navigating in the same state, or we guess a value for some undefined variable.
              \item Else if $\nu$ is total but $s \neq S(\nu(\mathbf{x}))$,
                  then
		$\Delta(q_{S(\mathbf{x})}^{\nu},(d,s)) \colonequals q_{S(\mathbf{x})}^{\nu}$: if the fact $s$ of the node is not a match, then we continue searching. 
              \item Else, the only remaining possibility is that $\nu$ is total and that $s = S(\nu(\mathbf{x}))$, in which case we set  $\Delta(q_{S(\mathbf{x})}^{\nu},(d,s)) \colonequals \true$, i.e., we have found a node containing the desired fact.
	\end{itemize}
	
	For every rule $r$ of $P$, subset $\calA$ of the literals in the
        body of $r$, and partial valuation $\nu$ of the variables in
        $\calA$ that is total for the variables in the head of $r$,
	we introduce a state $q_r^{\nu,\calA}$. 
	This state is intended to prove the literals in $\calA$ with the partial valuation $\nu$.
	We will describe the transitions for those states later.
	
	For every intensional predicate $R(\mathbf{x})$ appearing in a rule of $P$ and total valuation $\nu$ of $\mathbf{x}$,
	we have a state $q_{R(\mathbf{x})}^{\nu}$. This state is intended to prove $R(\mathbf{x})$ with the total valuation $\nu$.
	Let $(d,s) \in \Gamma_{\sigmae}^{\ki}$ be a symbol; we have
        the following transitions:

	\begin{itemize}
		\item If there is a $j$ such that $\nu(x_j) \notin d$,
                  then $\Delta(q_{R(\mathbf{x})}^{\nu},(d,s)) \colonequals \false$.
                  This is again in order to prevent the automaton from leaving the
                  allowed subtree.
		\item Else, $\Delta(q_{R(\mathbf{x})}^{\nu},(d,s))$ is defined a
                  disjunction of all the $q_r^{\nu^\prime,\calA}$ for
                  each rule~$r$ such that
		the head of~$r$ is $R(\mathbf{y})$, $\nu^\prime \colonequals
                \mathrm{Hom}_{\mathbf{y},\mathbf{x}}(\nu)$ is not~$\mnull$ and
		$\calA$ is the set of all literals in the body of $r$. Notice that because $\nu$ was total on~$\mathbf{x}$, $\nu'$ is also total on~$\mathbf{y}$.
		This transition simply means that we need to chose an
                appropriate rule to prove $R(\mathbf{x})$.
		We point out here that these transitions are the ones that make the construction quadratic instead of linear in $\card{P}$, but this will
		be handled later.
	\end{itemize}

	It is now time to describe transitions for the states
        $q_r^{\nu,\calA}$. Let $(d,s) \in \Gamma_{\sigmae}^{\ki}$, then:

	\begin{itemize}
		\item If there is a variable $z$ in $\calA$ such that $z$ is defined by $\nu$ and $\nu(z) \notin d$, then $\Delta(q_r^{\nu,\calA},(d,s)) \colonequals \false$.
		\item Else, if $\calA$ contains at least two literals, then $\Delta(q_r^{\nu,\calA},(d,s))$ is defined as a disjunction of $q_r^{\nu,\calA}$ and of $\Bigg[$~a disjunction over
		all the non-empty sets $\calA_1,\calA_2$ that partition $\calA$ of $\bigg[$a disjunction over all the total valuations $\nu^\prime$
                of $U(\nu)\cap \vars(\calA_1) \cap \vars(\calA_2)$ with values in $d$ of
		$\big[q_r^{\nu \cup \nu^\prime,\calA_1} \land q_r^{\nu \cup \nu^\prime,\calA_2}\big]\bigg]\Bigg]$.
		This transition means that we allow to split in two partitions
                the literals that need to be proven, and for each partition we launch one run
		that will have to prove it. In doing so, we
                have to take care that the two runs will build valuations
                that are consistent.
		This is why we fix the value of the variables that they have in common with a total valuation $\nu'$.
		\item Else, if $\calA = \{S(\mathbf{y})\}$ where $S$ is an
                  extensional relation, then $\Delta(q_r^{\nu,\calA},(d,s)) \colonequals q_{S(\mathbf{y})}^{\nu}$.
		\item Else, if $\calA = \{R^\prime(\mathbf{y})\}$ or $\{\lnot
                  R^\prime(\mathbf{y})\}$ where $R^\prime$ is an intensional
                  relation, and if $|\mathbf{y}| = 1$, and if $\nu(y)$ is undefined (where we write $y$ the one element of~$\mathbf{y}$), then 
			$\Delta(q_r^{\nu,\calA},(d,s)) \colonequals q_r^{\nu,\calA} \lor \bigvee_{a \in d} q_r^{\nu \cup \{y\mapsto a\},\calA}$.
		\item Else, if $\calA = \{R^\prime(\mathbf{y})\}$ where
                  $R^\prime$ is an intensional relation, then we will only define the transitions in the case where $\nu$ is total on~$\mathbf{y}$, in which case we set
		$\Delta(q_r^{\nu,\calA},(d,s)) \colonequals q_{R^\prime(\mathbf{y})}^{\nu}$.
                It is sufficient to define the transitions in this case, because $q_r^{\nu,\{R^\prime(\mathbf{y})\}}$ can only be reached
                if $\nu$ is total on~$\mathbf{y}$. Indeed, if $|\mathbf{y}| = 1$, then $\nu$ must be
                total on~$\mathbf{y}$ because we would have applied the previous bullet
                point otherwise. If $|\mathbf{y}| > 1$, the only way we could have reached the state $q_r^{\nu,\{R^\prime(\mathbf{y})\}}$
                is by a sequence of transitions involving $q_r^{\nu_0,\calA_0},\ldots,q_r^{\nu_m,\calA_m}$, where $\calA_0$ are all the literals in the body of $r$,
		$\calA_m$ is $\{R^\prime(\mathbf{y})\}$ and $\nu_m$ is $\nu$.
		We can then see that, during the partitioning process,
                $R^\prime(\mathbf{y})$ must have been separated from all the extensional atoms that formed its guard, hence
		all its variables have been assigned a valuation.
		\item Else, if $\calA = \{\lnot R^\prime(\mathbf{y})\}$ with $R^\prime$ intensional, 
		then $\Delta(q_r^{\nu,\calA},(d,s)) \colonequals \lnot q_{R^\prime(\mathbf{y})}^{\nu}$.
                Again, we can show that it suffices to consider the case where $\nu$ is total on~$\mathbf{y}$, for the same reasons as in the previous bullet point.
	\end{itemize}

Finally, the initial state of $A_P^\prime$ is
$q_{\text{Goal}}^{\emptyset}$.

\medskip

We describe the stratification function $\strat^\prime$ of $A^\prime_P$. Let $\strat$ be that of $P$.
For any state $q$ of the form $q_{T(\mathbf{x})}^\nu$ or
$q_r^{\nu,\calA}$ with~$r$ having as head relation $T$, then $\strat^\prime(q)$ is 
$0$ if $T$ is extensional and $\strat(T)$ (which is $\geq 1$) if $T$ is intensional. Notice that then only states corresponding to extensional relations are in the first stratum.
It is then clear from the transitions that $\strat^\prime$ is a valid stratification function for $A^\prime_P$.

\medskip

    As previously mentioned, the construction of $A'_P$ is not
    FPT-linear, but we will explain at the end of the proof how to construct in FPT-linear time a SATWA
    $A_P$ equivalent to~$A'_P$.

\subparagraph*{$\bm{A^\prime_P}$ tests $\bm{P}$ on tree encodings of width
$\bm{\leq \ki}$.}
To show this claim, let $\la T, \lambda_E \ra$ be a $(\sigmae,\ki)$-tree encoding. 
Let $I$ be the instance obtained by decoding $\la T, \lambda_E \ra$;
we know that $I$ has treewidth $\leq \ki$ and that we can define from $\la T, \lambda_E \ra$ a tree 
decomposition $\la T, \dom \ra$ of~$I$ whose underlying tree is also~$T$.
For each node $n \in T$, let $\dec_n : \calD_{\ki} \to \dom(n)$ be the function that decodes the 
elements in node~$n$ of the encoding to the elements of~$I$ that are in the corresponding bag of the tree decomposition,
and let $\enc_n : \dom(n) \to \calD_{\ki}$ be the inverse function that encodes back the elements,
so that we have $\dec_n \circ \enc_n = \enc_n \circ \dec_n =
\mathrm{Id}$.
We will denote elements of $\calD_{\ki}$ by $a$ and elements in the domain of $I$ by $c$.

We recall some properties of tree decompositions and tree encodings:
\begin{property}
\label{propTreeEnc1}
Let $n_1,n_2$ be nodes of $T$ and $a \in \calD_{\ki}$ be an (encoded)
  element that appears in the $\lambda_E$-image of $n_1$ and $n_2$. Then the element $a$ appears in every node in the path from $n_1$ to $n_2$ if and only if $\dec_{n_1}(a) = \dec_{n_2}(a)$.
\end{property}

\begin{property}
\label{propTreeEnc2}
	Let $n_1,n_2$ be nodes of $T$ and $c$ be an element of $I$ that appears
        in $\dom(n_1) \cap \dom(n_2)$.
	Then for every node $n'$ on the path from $n_1$ to $n_2$, $c$ is also in
        $\dom(n')$, and moreover $\enc_{n'}(c)=\enc_{n_1}(c)$.
\end{property}

We start with the following lemma about extensional facts:
\begin{lemma}
\label{lem:extensional}
  For every extensional relation $S$, node $n \in T$, variables $\mathbf{y}$, and partial valuation $\nu$ of $\mathbf{y}$, there exists a run $\rho$ of $A^\prime_P$ starting at node $n$
in state $q_{S(\mathbf{y})}^{\nu}$ if and only if there exists
a fact $S(\mathbf{c})$ in~$I$ 
such that we have $\dec_n(\nu(y_j))=c_j$
  for every $y_j$ defined by $\nu$.
  We call this a match $\mathbf{c}$ of~$S(\mathbf{y})$ in~$I$ that is \emph{compatible with $\nu$ at node $n$}.
		
\end{lemma}
\begin{proof}
  We prove each direction in turn.

  \subparagraph*{\qquad Forward direction.}
	Suppose there exists  a run $\rho$ of $A^\prime_P$ starting at node $n$
in state $q_{S(\mathbf{y})}^{\nu}$. First, notice that by design of the
  transitions starting in a state of that form, 
  states appearing in the labeling of the run can only be of the form
  $q_{S(\mathbf{y})}^{\nu\prime}$ for an extension $\nu'$ of~$\nu$.
We will show by induction on the run that for every node $\pi$ of the run labeled by $(q_{S(\mathbf{y})}^{\nu\prime},m)$, 	
there exists $\mathbf{c'}$ such that $S(\mathbf{c'}) \in I$ and
  $\mathbf{c'}$ is compatible with $\nu'$ at node $m$. This will conclude
  the proof of the forward part of the lemma, by taking $m=n$.

The base case is when $\pi$ is a leaf of $\rho$.
	The node $\pi$ is then labeled by $(q_{S(\mathbf{y})}^{\nu^\prime}, m)$ such that $\Delta(q_{S(\mathbf{y})}^{\nu^\prime},\lambda_E(m)) = \true$.
	Let $(d,s) = \lambda_E(m)$. 
	By construction of the automaton we have that $\nu^\prime$ is total and $s = S(\nu^\prime(\mathbf{y}))$.
	We take $\mathbf{c'}$ to be $\dec_m(\nu^\prime(\mathbf{y}))$, which satisfies the compatibility condition by definition
	and is such that $S(\mathbf{c'}) =
        S(\dec_m(\nu^\prime(\mathbf{y}))) = \dec_m(s) \in I$.

	When $\pi$ is an internal node of $\rho$, we write
        $(q_{S(\mathbf{y})}^{\nu^\prime}, m)$ its label.
        By definition of the
        transitions of the automaton, we have
        $\Delta(q_{S(\mathbf{y})}^{\nu^\prime},(d,s)) =
	q_{S(\mathbf{y})}^{\nu^\prime} \lor \bigvee\limits_{a \in d, y_j \in U(\nu^\prime)} q_{S(\mathbf{y})}^{\nu^\prime \cup \{y_j\mapsto a\}}$.
        Hence, the node $\pi$ has exactly one child $\pi'$, the first component
        of its label is some $m' \in \neigh(m)$, and we have two
        cases depending on the first component of its label:

	\begin{itemize}
          \item $\pi'$ may be labeled by $(q_{S(\mathbf{y})}^{\nu^\prime}, m^\prime)$.
                Then by induction on the run there exists $\mathbf{c''}$ such that $S(\mathbf{c''}) \in I$ and
                $\mathbf{c''}$ is compatible with $\nu'$ at node $m'$.
                We take $\mathbf{c'}$ to be $\mathbf{c''}$, so
                that we only need to check the compatibility
                condition, i.e., that for every $y_j$ defined
                by $\nu'$, $\dec_m(\nu'(y_j)) = c_j = \dec_{m'}(\nu'(y_j))$.
                This is true by Property~\ref{propTreeEnc1}. Indeed, 
                for every $y_j$ defined by $\nu'$, we must have $\nu'(y_j) \in
                m'$, otherwise $\pi'$ would have a label that cannot occur in a
                run.
        \item $\pi^\prime$ is labeled by $(q_{S(\mathbf{y})}^{\nu^\prime \cup \{y_j\mapsto a\}}, m^\prime)$
          for some $a \in d$ and for some $y_j \in U(\nu^\prime)$. 
                Then by induction on the run there exists $\mathbf{c''}$ such that $S(\mathbf{c''}) \in I$ and
                $\mathbf{c''}$ is compatible with $\nu^\prime \cup \{y_j\mapsto a\}$ at node $m'$.
                We take $\mathbf{c'}$ to be $\mathbf{c''}$, which again satisfies the compatibility condition thanks to Property~\ref{propTreeEnc1}.	
\end{itemize}

  \subparagraph*{\qquad Backward direction.}
Now, suppose that there exists $\mathbf{c}$ such that
$S(\mathbf{c}) \in I$ and $\mathbf{c}$ is compatible with $\nu$ at node~$n$.
The fact $S(\mathbf{c})$ is encoded somewhere
   $\la T, \lambda_E \ra$, so there exists a node $m$ such that, letting $(d,s)$ be $\lambda_E(m)$, we have $\dec_m(s)=S(\mathbf{c})$.
Let $n = m_1, m_2,\ldots,m_p = m$ be the nodes on the path from $n$ to $m$, and $(d_i,s_i)$ be $\lambda_E(m_i)$ for $1 \leq i \leq p$.
By compatibility, for every $y_j$ defined by $\nu$ we have $\dec_n(\nu(y_j)) = c_j$. But $\dec_n(\nu(y_j)) \in \dom(n)$ and $c_j \in \dom(m)$ so by Property~\ref{propTreeEnc2}, for every $1 \leq i \leq p$ we have
$c_j \in \dom(m_i)$ and $\enc_{m_i}(c_j) =\enc_{n}(c_j) = \enc_{n}(\dec_n(\nu(y_j))) = \nu(y_j)$, so that $\nu(y_j) \in d_i$.
We can then construct a run $\rho$ starting at node $n$ in state
  $q_{S(\mathbf{y})}^{\nu}$ as follows. The root $\pi_1$ is labeled by $(q_{S(\mathbf{y})}^{\nu},n)$, and
for every $2 \leq i \leq p$, $\pi_i$ is the unique child of $\pi_{i-1}$ and is labeled by $(q_{S(\mathbf{y})}^{\nu}, m_i)$. 
This part is valid because we just proved that for every $i$, there is no $j$ such that $y_j$ is defined by $\nu$ and $\nu(y_j) \notin d_j$.
Now from $\pi_m$, we continue the run by staying at node $m$ and building up the
valuation, until we reach a total valuation $\nu_{\f}$ such that $\nu_\f(\mathbf{y}) = \enc_m(\mathbf{c})$.
Then we have $s = S(\nu_\f(\mathbf{y}))$ and the transition is $\true$, which
completes the definition of the run.
\end{proof}

The preceding lemma concerns the base case of extensional relations. We
now prove a similar \emph{equivalence lemma} for intensional relations.
This lemma allows us to conclude the correctness proof, by applying it to the
$\mathrm{Goal}()$ predicate and to the root of the tree-encoding.

\begin{lemma}
  \label{lem:correctness}
For every relation $R$, node $n \in T$ and total valuation $\nu$ of $\mathbf{x}$, there exists a run $\rho$ of $A^\prime_P$ starting at node $n$
in state $q_{R(\mathbf{x})}^{\nu}$ if and only if 
  $R(\dec_n(\nu(\mathbf{x}))) \in P(I)$. 
\end{lemma}

\begin{proof}
We will prove this equivalence by induction on the stratum $\strat(R)$ of
  the relation $R$. 
The base case ($\strat(R)=0$, so $R$ is an extensional relation) was shown in Lemma~\ref{lem:extensional}.
For the inductive case, where $R$ is an intensional relation, we prove each direction separately.

\subparagraph*{\qquad Forward direction.}
First, suppose that there exists a run $\rho$ of $A^\prime_P$ starting at node $n$
in state $q_{R(\mathbf{x})}^{\nu}$. We show by induction on the run (from bottom to top) that for every node $\pi$ of the run the following implications hold:

\begin{enumerate}[(i)]
\item \label{case:positive-atom} 
  If $\pi$ is labeled with $(q_{R^\prime(\mathbf{y})}^{\nu^\prime}, m)$, then there exists $\mathbf{c}$ such that
		$R^\prime(\mathbf{c}) \in P(I)$ and $\mathbf{c}$ is
                compatible with $\nu'$ at node~$m$.
              \item \label{case:negative-atom}
                If $\pi$ is labeled with $\lnot (q_{R^\prime(\mathbf{y})}^{\nu^\prime}, m)$ with $\nu^\prime$ total, then
		$R^\prime(\dec_m(\nu^\prime(\mathbf{y}))) \notin P(I)$.
              \item \label{case:rule} If $\pi$ is labeled with $(q_r^{\nu^\prime,\calA},m)$, then
          there exists a mapping $\mu : \vars(\calA) \to \mathrm{Dom}(I)$
          that is compatible with $\nu'$ at node~$m$ and such that:
		\begin{itemize}
			\item For every positive literal $S(\mathbf{z})$ in $\calA$, then $S(\mu(\mathbf{z})) \in P(I)$.
			\item For every negative literal $\lnot S(\mathbf{z})$ in $\calA$, then $S(\mu(\mathbf{z})) \notin P(I)$.
		\end{itemize}
\end{enumerate}

The base case is when $\pi$ is a leaf. Notice that in this case, and by
  construction of $A^\prime_P$, the node $\pi$ cannot be labeled by states corresponding
  to rules of~$P$: indeed, there
are no transition for these states leading to a tautology, and all
  transitions to such a state are from a state in the same stratum, so $\pi$
  could not be a leaf. Thus,
  we have three subcases:
\begin{itemize}
	\item $\pi$ may be labeled by
          $(q_{R^\prime(\mathbf{y})}^{\nu^\prime}, m)$, where $R^\prime$ is
          extensional. We must show~(\ref{case:positive-atom}), but this follows
          from Lemma~\ref{lem:extensional}.
	\item $\pi$ may be labeled by $(q_{R^\prime(\mathbf{y})}^{\nu^\prime},
          m)$, where $R^\prime$ is intensional
	and verifies $\strat(R^\prime) < i$, and where $\nu^\prime$ is total. Again we
        need to show~(\ref{case:positive-atom}).
	By definition of the run $\rho$, this implies that there exists a
        run of~$A^\prime_P$ starting at $m$ in state
        $q_{R^\prime(\mathbf{y})}^{\nu'}$.
	But $\nu^\prime$~is total, so by induction on the strata we have
        (using the forward direction of the equivalence lemma) that
        $R^\prime(\dec_m(\nu^\prime(\mathbf{y}))) \in P(I)$.
	We take $\mathbf{c}$ to be $\dec_m(\nu^\prime(\mathbf{y}))$, which satisfies the required conditions.
	\item $\pi$ may be labeled by $\lnot
          (q_{R^\prime(\mathbf{y})}^{\nu^\prime}, m)$, where $R^\prime$ is intensional 
		and verifies $\strat(R^\prime) < i$, and where $\nu^\prime$ is total. We
                need to show~(\ref{case:negative-atom}).
		By definition of the run $\rho$ there exists no run of
                $A^\prime_P$ starting at $m$ in state
                $q_{R^\prime(\mathbf{y})}^{\nu'}$. 
		But $\nu^\prime$ is total, so by induction on the strata
                we have (using the backward direction of the
                equivalence lemma) that
                $R^\prime(\dec_m(\nu^\prime(\mathbf{y}))) \notin P(I)$, which is
                what we needed to show.
\end{itemize}

For the induction case, where $\pi$ is an internal node and letting
  $(d,s)$ be $\lambda_E(m)$ in what follows, we have five subcases:
\begin{itemize}
\item $\pi$ may be labeled by $(q_{R^\prime(\mathbf{y})}^{\nu^\prime},
  m)$ with $R^\prime$ extensional. We must show (\ref{case:positive-atom}), but
    this follows from Lemma~\ref{lem:extensional}.
\item $\pi$ may be labeled by $(q_{R^\prime(\mathbf{y})}^{\nu^\prime}, m)$ with $R^\prime$ intensional and $\nu^\prime$ total. 
  We need to prove~(\ref{case:positive-atom}).
	In that case, given the definition of $\Delta(q_{R^\prime(\mathbf{y})}^{\nu^\prime},(d,s))$ and by induction (on the run),
        there exists a child $\pi^\prime$ of~$\pi$ labeled by
        $(q_r^{\nu'',\calA},m')$, where $m^\prime \in \neigh(m)$, where $r$ is a
        rule with head $R^\prime(\mathbf{z})$, where
	$\nu'' = \mathrm{Hom}_{\mathbf{z},\mathbf{y}}(\nu')$ is a partial
        valuation which is not~$\mnull$, and where
	$\calA$ is the set of literals of $r$. 
        Then, by induction on the run, there exists
        a mapping $\mu : \vars(\calA) \to \mathrm{Dom}(I)$ that
        verifies~(\ref{case:rule}).
	Thus by definition of the semantics of $P$ we have that
        $R'(\mu(\mathbf{z})) \in P(I)$, and we take $\mathbf{c}$
	to be $\mu(\mathbf{z})$. What is left to check is that the compatibility condition holds.
	We need to prove that $\dec_m(\nu'(\mathbf{y})) = \mathbf{c}$,
        i.e., that $\dec_m(\nu'(\mathbf{y})) = \mu(\mathbf{z})$.
	We know, by definition of $\mu$, that $\dec_{m'}(\nu''(\mathbf{z})) = \mu(\mathbf{z})$. 
	So our goal is to prove $\dec_m(\nu'(\mathbf{y})) =
        \dec_{m'}(\nu''(\mathbf{z}))$, i.e., by definition of $\nu''$
	we want $\dec_m(\nu'(\mathbf{y})) = \dec_{m'}(\mathrm{Hom}_{\mathbf{z},\mathbf{y}}(\nu')(\mathbf{z}))$.
        By definition of
        $\mathrm{Hom}_{\mathbf{z},\mathbf{y}}(\nu')$,
        we know that $\nu'(\mathbf{y}) =
        \mathrm{Hom}_{\mathbf{z},\mathbf{y}}(\nu')(\mathbf{z})$, and this
        implies the desired equality by applying
        Property~\ref{propTreeEnc1} to~$m$ and~$m'$.

\item $\pi$ may be labeled by $(q_r^{\nu',\calA},m)$, where $\calA =
  \{R''(\mathbf{y})\}$ or $\{\lnot R''(\mathbf{y})\}$, where $|\mathbf{y}| = 1$,
    where the head of~$r$ uses relation $R'$.,
    and where $y \in U(\nu')$ (writing $y$ the one element of~$\mathbf{y}$).
	We need to prove~(\ref{case:rule}).
	By construction we have $\Delta(q_r^{\nu',\calA},(d,s)) = q_r^{\nu',\calA} \lor \bigvee_{a \in d} q_r^{\nu' \cup \{y\mapsto a\},\calA}$.
	So by definition of a run there is $m' \in \neigh(m)$ and a child $\pi'$ of $\pi$ such that $\pi'$ is labeled by $( q_r^{\nu',\calA},m')$ or by
	$( q_r^{\nu' \cup \{y\mapsto a\},\calA},m')$ for some $a\in d$. 
	In both cases it is easily seen that we can define an appropriate $\nu$ from the valuation~$\nu'$ that we obtain
        by induction on the run (more details are given in the next bullet point).
\item $\pi$ may be labeled by $(q_r^{\nu',\calA},m)$, with $\calA =
	\{R''(\mathbf{y})\}$ and $\nu'$ total on~$\mathbf{y}$, the head of~$r$ using relation $R'$. We need to
    prove~(\ref{case:rule}).
	By construction we have $\Delta(q_r^{\nu',\calA},(d,s)) = q_{R''(\mathbf{y})}^{\nu^\prime}$,
	so that by definition of the run there is $m' \in \neigh(m)$ and a child $\pi'$ of $\pi$ such that $\pi'$
	is labeled by $(q_{R''(\mathbf{y})}^{\nu^\prime},m')$.
	Thus by induction on the run there exists $\mathbf{c}$ such that
        $R''(\mathbf{c}) \in P(I)$ and $\mathbf{c}$ compatible with~$\nu'$ at
        node~$m'$. By Property~\ref{propTreeEnc1}, $\mathbf{c}$ is also
        compatible with~$\nu'$ at node~$m$.
	We define $\mu$ by $\mu(\mathbf{y}) \defeq \mathbf{c}$, which effectively defines it because in this case $\vars(r)=\mathbf{y}$, and this choice satisfies the required properties.
\item $\pi$ may be labeled by $(q_r^{\nu',\calA},m)$, with $\calA =
	\{\lnot R''(\mathbf{y})\}$ and $\nu'$ total on~$\mathbf{y}$ and the head of~$r$ has relation $R'$. We
    again need to prove (\ref{case:rule}).
	By construction we have $\Delta(q_r^{\nu',\calA},(d,s)) = \lnot q_{R''(\mathbf{y})}^{\nu^\prime}$ and then by definition of the automaton
	there exists a child $\pi'$ of $\pi$ labeled by $\lnot (q_{R''(\mathbf{y})}^{\nu^\prime},m)$ with $\strat(R'') < i$ and
	there exists no run starting at node $m$ in state $q_{R''(\mathbf{y})}^{\nu^\prime}$.
	So by induction on the strata (using the backward direction of
        the equivalence lemma) we have $R''(\dec_m(\nu'(\mathbf{y})))
        \notin P(I)$.
	We define $\mu$ by $\mu(\mathbf{y}) = \dec_m(\nu'(\mathbf{y}))$, which effectively defines it because $\vars(r)=\mathbf{y}$, and the compatibility conditions are satisfied.

\item $\pi$ may be labeled by $(q_r^{\nu',\calA},m)$, with $\card{\calA}
  \geq 2$. We need to prove (\ref{case:rule}).
	Given the definition of $\Delta(q_r^{\nu',\calA},(d,s))$ and by
        definition of the run,
	one of the following holds:
	\begin{itemize}
		\item There exists $m' \in \neigh(m)$ and a child $\pi'$ of $\pi$ such that $\pi'$ is labeled by $(q_r^{\nu',\calA},m')$.
			By induction there exists $\mu' : \vars(\calA)
                        \to \mathrm{Dom}(I)$ satisfying (\ref{case:rule})
                        for node $m'$. We can take $\mu$ to be $\mu'$, which satisfies the required properties.
		\item 
	There exist $m_1, m_2 \in \neigh(m)^2$ and $\pi_1, \pi_2$ children of $\pi$ and non-empty sets $\calA_1,\calA_2$ 
	that partition $\calA$ and a total valuation $\nu''$ of
            $\vars(\calA_1) \cap \vars(\calA_2)$ with values in $d$
	such that $\pi_1$ is labeled by $(q_r^{\nu' \cup \nu'',\calA_1},m_1)$ 
	and $\pi_2$ is labeled by $(q_r^{\nu' \cup \nu'',\calA_2},m_2)$.
	By induction there exists $\mu_1 : \vars(\calA_1) \to
            \mathrm{Dom}(I)$ and similarly $\mu_2$ that satisfy
            (\ref{case:rule}).
	Thanks to the compatibility conditions for $\mu_1$ and $\mu_2$ and to Property~\ref{propTreeEnc1} applied to~$m_1$ and~$m_2$ via~$m$, we can define $\mu : \vars(\calA) \to \mathrm{Dom}(I)$
	with $\mu = \mu_1 \cup \mu_2$. One can check that $\mu$ satisfies the required properties.
	\end{itemize}
\end{itemize}

Hence, the forward direction of our equivalence lemma is proven.

\subparagraph*{\qquad Backward direction.}
We now prove the backward direction of the induction case of
our main equivalence lemma (Lemma~\ref{lem:correctness}).
From the induction hypothesis on strata,
we know that, for every relation $R$ with $\strat(R) \leq i-1$, for every node $n \in T$ and total valuation $\nu$ of $\mathbf{x}$, 
there exists a run $\rho$ of $A^\prime_P$ starting at node $n$
in state $q_{R(\mathbf{x})}^{\nu}$ if and only if we have $R(\dec_n(\nu(\mathbf{x}))) \in P(I)$. 
Let~$R$ be a relation with $\strat(R)=i$, $n \in T$ be a node and $\nu$ be a total valuation of $\mathbf{x}$ such that $R(\dec_n(\nu(\mathbf{x}))) \in P(I)$.
We need to show that there exists a run $\rho$ of $A^\prime_P$ starting at node $n$
in state $q_{R(\mathbf{x})}^{\nu}$.
We will prove this by induction on the smallest $j \in \mathbb{N}$ such
that $R(\dec_n(\nu(\mathbf{x}))) \in \Xi^j_{P}(P_{i-1}(I))$, where
$\Xi^j_{P}$ is the $j$-th
application of the immediate consequence operator for the program~$P$ (see
\cite{abiteboul1995foundations}) and $P_{i-1}$ is the restriction of
$P$ with only the rules up to strata~$i-1$.
The base case, when $j=0$, is in fact vacuous since 
$R(\dec_n(\nu(\mathbf{x}))) \in \Xi^0_{P}(P_{i-1}(I))=P_{i-1}(I)$ implies
that $\strat(R)\leq i-1$, whereas we assumed $\strat(R)=i$.
For the inductive case ($j \geq 1$), we have $R(\dec_n(\nu(\mathbf{x})))
\in \Xi_P^j(P_{i-1}(I))$ so by definition of the semantics of $P$, there is a rule
$r$ of the form $R(\mathbf{z}) \leftarrow L_1(\mathbf{y}_1) \ldots
L_t(\mathbf{y}_t)$ of $P$ and a mapping $\mu :
\mathbf{y}_1\cup\dots\cup\mathbf{y}_t\to \mathrm{Dom}(I)$ such that
$\mu(\mathbf{z}) = \dec_n(\nu(\mathbf{x}))$ and, for every literal $L_l$ in the body of $r$:
\begin{itemize}
  \item If $L_l(\mathbf{y}_l) = R_l(\mathbf{y}_l)$ is a positive
    literal, then $R_l(\mu(\mathbf{y}_l)) \in
          \Xi_P^{j-1}(P_{i-1}(I))$
        \item If $L_l(\mathbf{y}_l) = \lnot R_l(\mathbf{y}_l)$ is a
          negative literal, then $R_l(\mu(\mathbf{y}_l)) \notin
          P_{i-1}(I)$
\end{itemize}
To achieve our goal of building a run starting at node $n$ in state  $q_{R(\mathbf{x})}^{\nu}$, we will construct a run starting at node $n$ in state $q_r^{\nu',\{L_1,\ldots,L_t\}}$, 
with $\nu' = \mathrm{Hom}_{\mathbf{z},\mathbf{x}}(\nu)$.
The first step is to take care of the literals of the rule and to prove that:

\begin{enumerate}[(i)]
\item If $L_l(\mathbf{y}_l) = R_l(\mathbf{y}_l)$ is a positive literal, then there exists a node $m_l$ and a valuation $\nu_l$ such that there exists a run $\rho_l$
		starting at node $m_l$ in state
                $q_{R_l(\mathbf{y}_l)}^{\nu_l}$ and such that
                $\dec_{m_l}(\nu_l(\mathbf{y}_l)) = \mu(\mathbf{y}_l)$.
              \item If $L_l(\mathbf{y}_l) = \lnot R_l(\mathbf{y}_l)$ is a negative literal, then for every node $m_l$ such that $\dec_{m_l}(\nu_l(\mathbf{y}_l)) = \mu(\mathbf{y}_l)$,
                there exists no run starting at node $m_l$ in state $q_{R_l(\mathbf{y}_l)}^{\nu_l}$. 
\end{enumerate}

We straightforwardly get (ii) by using the induction on the strata of our equivalence
lemma.
We now prove (i).
Suppose first that $R_l$ is an extensional relation.
We define $m_l$ to be the node in which $R_l(\mu(\mathbf{y}_l))$ appears (in the tree decomposition), and we define $\nu_l$ to be $\enc_{m_l}(\mu(\mathbf{y}_l))$.
We then have $\dec_{m_l}(\nu_l(\mathbf{y}_l)) =
\dec_{m_l}(\enc_{m_l}(\mu(\mathbf{y}_l))) = \mu(\mathbf{y}_l)$, so by
Lemma~\ref{lem:extensional} there exists a run $\rho_l$ starting at
node $m_l$ in state $q_{R_l(\mathbf{y}_l)}^{\nu_l}$.

Suppose now that $R_l$ is intensional. By (syntactical) definition of our fragment, $R_l(\mathbf{y}_l)$ is clique-guarded by some extensional relations in the body of $r$, say 
$R_{l_1}(\mathbf{y}_{l_1}),\ldots,R_{l_c}(\mathbf{y}_{l_c})$.
Moreover, there exist nodes $m_{l_1},\ldots,m_{l_c}$ such that for every $1 \leq p \leq c$, $R_{l_p}(\mu(\mathbf{y}_{l_p}))$ is in $m_{l_p}$. 
By a well-known property of tree decompositions, this implies that there exists a node in which all the elements of $\mu(\mathbf{y}_l)$ appear
(see Lemma~1 of~\cite{gavril1974intersection}, Lemma~2 of~\cite{BodlaenderK10}).
We define $m_l$ to be this node.
We define $\nu_l : \mathbf{y}_l \to \calD_{\ki}$ to be $\bigsqcup\limits_{1\leq p \leq c} \nu_{l_p}$,
where the $\nu_{l_p}$ are the valuations obtained when proving (i) for extensional relations in the case above.
This definition makes sense. Indeed, let $v$ be a variable in $\mathbf{y}_{l_p} \cap \mathbf{y}_{l_{p'}}$. 
Because we defined $\nu_{l_p}$ by $\enc_{m_l}(\mu)$, we have that $\mu(v)$ appears in the bag (of the tree decomposition) of $m_{l_p}$, and similarly in that of $m_{l_{p'}}$.
Then by Property~\ref{propTreeEnc2}, we have that $\enc_{m_{l_p}}(\mu(v)) = \enc_{m_{l_{p'}}}(\mu(v))$, and thus $\nu_{l_p}(v) = \nu_{l_{p'}}(v)$. 
We now show that $\dec_{m_l}(\nu_l(\mathbf{y}_l)) = \mu(\mathbf{y}_l)$, which will imply by induction hypothesis (on the number~$j$ of applications of the immediate consequence operator) that there exists a run $\rho_l$ starting at node $m_l$
in state $q^{\nu_l}_{R_l(\mathbf{y}_l)}$. Pick $v \in \mathbf{y}_l$. 
It is in some $\mathbf{y}_{l_p}$ for some $1 \leq p \leq c$, so by definition of $\nu_l$ and of $\nu_{l_p}$ we only need
to prove $\dec_{m_l}(\enc_{m_{l_p}}(\mu(v))) = \mu(v)$.
But we have $\mu(v)$ is in the bag of $m_{l_p}$ (by definition of $\nu_{l_p}$), and in that of $m_l$ (by definition of $m_l$), so that again by Property~\ref{propTreeEnc2} we get 
$\enc_{m_{l_p}}(\mu(v)) = \enc_{m_l}(\mu(v))$, which gives us what we wanted because we can decode.
Hence, (i) and (ii) are proven.

The second step is, from the runs $\rho_l$ that we just constructed, to construct a run starting at node $n$ in state $q_r^{\nu',\{L_1,\ldots,L_t\}}$.
We describe in a high-level manner how we build the run.
Starting at node~$n$, we partition the literals to prove (i.e., the atoms of the
body of the rule that we are applying), in the following way:

\begin{itemize}
  \item We create one class in the partition for each positive literal $R_l$
    (which can be intentional or extensional)
    such that $m_l$ is $n$, which we prove directly at the current node.
    Specifically, we handle these literals one by one, by splitting the
    remaining literals in two using the transition formula corresponding to the rule and by staying at node $n$ and building the valuations according to 
$\dec_n(\mu)$.
    \item We create one class in the partition for each negative literal $\lnot
      R_l(\mathbf{y}_l)$ such that all its variables $\mathbf{y}_l$ are defined
      by the valuation: we use (ii) to know that there will be no run for there literals.
    \item For the remaining literals, considering all neighbors of~$n$ in the
      tree encoding, we split the literals into one class per neighbor $n'$,
      where each literal $L_l$ is mapped to the neighbor that allows us to reach
      its node~$m_l$. We ignore the empty classes. If there is only one class,
      i.e., we must go in the same direction to prove all facts, we simply go to
      the right neighbor~$n'$, remaining in the same state. If there are
      multiple classes, we partition the facts and prove each class on the
      correct neighbor.
      
      One must then argue that, when we do so, we can indeed
      choose the image by~$\nu'$ of all elements that were shared between
      literals in two different classes and were not yet defined in~$\nu'$. The
      reason why this is possible is because we are working on a tree encoding:
      if two facts of the body share a variable $x$, and the two facts will be
      proved in two different directions, then the variable $x$ must be mapped
      to the same element in the two direction, which implies that it must occur
      in the node $m$ where we split. Hence, we can indeed choose the image
      of~$x$ at the moment when we split.
\end{itemize}
\end{proof}

\subparagraph*{FPT-linear time construction.}
Finally, we justify that we can construct in FPT-linear time the automaton $A_P$ which recognizes the same language as $A_P^\prime$.
The size of $\Gamma_{\sigmae}^{\ki}$ only depends on $\ki$ and on the
extensional signature, which are fixed. As the number of states is linear
in $\card{P}$, the number of transitions is linear in~$\card{P}$. 
Most of the transitions are of constant size, and in fact one can check that the only problematic transitions
are those for states of the form $q^\nu_{R(\textbf{x})}$ 
with $R$ intensional, specifically the second bullet point. Indeed, we have
defined a transition from $q^\nu_{R(\textbf{x})}$, for each valuation $\nu$ of a
rule body, to the $q_r^{\nu^\prime,\calA}$ for linearly many rules, so in
general there are quadratically many transitions.

However, it is easy to fix this problem: instead of having one state $q^\nu_{R(\textbf{x})}$
for every occurrence of an intensional predicate $R(\textbf{x})$ in a rule body
of~$P$ and total valuation $\nu$ of this rule body,
we can instead have a constant number of states $q_{R(\textbf{a})}$ for $\textbf{a} \in \calD_{\ki}^{\arity{R}}$. 
In other words, when we have decided to prove a single intensional atom in the
body of a rule, instead of remembering the entire valuation of the rule body (as
we remember $\nu$ in $q^\nu_{R(\textbf{x})}$), we can simply forget all other
variable values, and just remember the tuple which is the image of~$\textbf{x}$
by~$\nu$, as in~$q_{R(\textbf{a})}$. Remember that the number of such states is
only a function of $\kp$ and $\ki$, because bounding $\kp$ implies that we bound
the arity of~$P$, and thus the arity of intensional predicates.

We now redefine the transitions for those states :

	\begin{itemize}
		\item If there is a $j$ such that $a_j \notin d$,
                  then $\Delta(q_{R(\mathbf{a})},(d,s)) = \false$.
		\item Else, $\Delta(q_{R(\mathbf{a})},(d,s))$ is a
                  disjunction of all the $q_r^{\nu^\prime,\calA}$ for
                  each rule~$r$ such that
		  the head of~$r$ is $R(\mathbf{y})$, $\nu^\prime(\mathbf{y}) = \mathbf{a}$ and
		$\calA$ is the set of all literals in the body of $r$.
	\end{itemize}

        The key point is that a given $q_r^{\nu',\calA}$ will only appear
        in rules for states of the form $q_{R(\mathbf{a})}$ where $R$ is
        the predicate of the head of $r$, and there is a constant number
        of such states.
        
We also redefine the transitions that used these states:
\begin{itemize}
\item Else, if $\calA = \{R^\prime(\mathbf{y})\}$ with $R^\prime$ intensional, 
		then $\Delta(q_r^{\nu,\calA},(d,s)) = q_{R^\prime(\nu(\mathbf{y}))}$.
\item Else, if $\calA = \{\lnot R^\prime(\mathbf{y})\}$ with $R^\prime$ intensional, 
		then $\Delta(q_r^{\nu,\calA},(d,s)) = \lnot q_{R^\prime(\nu(\mathbf{y}))}$. 
\end{itemize}

$A_P$ recognizes the same language as $A_P'$. Indeed, consider a run of $A'_P$, and replace every state $q^\nu_{R(\textbf{x})}$ with $R$ intensional by the state 
$q_{R(\nu(\textbf{x}))}$: we obtain a run of $A_P$. 
Conversely, being given a run of $A_P$, observe that every state $q_{R(\textbf{a})}$ comes from a state  $q_r^{\nu,\{R(\mathbf{y})\}}$ with 
$\nu(\mathbf{y}) = \mathbf{a}$.
We can then replace $q_{R(\textbf{a})}$ by the state $q^\nu_{R(\textbf{x})}$ to obtain a run of $A'_P$.
\end{toappendix}

\section{Provenance Cycluits}
\label{sec:provenance}
In the previous section, we have seen how ICG-Datalog programs could be compiled
efficiently to tree automata (SATWAs) that test them on treelike instances. To
show that SATWAs can be evaluated in linear time
(stated earlier as Proposition~\ref{prp:satwaeval}), we will
introduce an operational semantics for SATWAs based on the notion of
\emph{cyclic circuits}, or \emph{cycluits} for short.

We will also use
these cycluits as a new powerful tool to 
compute (Boolean) \emph{provenance
information}, i.e., a representation of how the query result depends on the
input data:

\begin{definition}
  \label{def:provenance}
  A (Boolean) \emph{valuation} of a set $S$ is a function $\nu: S \to \{0, 1\}$.
  A \emph{Boolean function} $\phi$ on variables~$S$ is a mapping that associates
  to each valuation $\nu$ of~$S$ a Boolean value in $\{0, 1\}$ called the
  \emph{evaluation} of~$\phi$ according to~$\nu$; for consistency with further
  notation, we write it $\nu(\phi)$.
  The \emph{provenance} of a query $Q$ on an
  instance $I$ is the Boolean function $\phi$, whose variables are the facts
  of~$I$, which is defined as follows:
  for any valuation $\nu$ of the facts of $I$, we have $\nu(\phi)
  = 1$ iff the subinstance $\{F \in I \mid \nu(F) = 1\}$ satisfies~$Q$.
\end{definition}

We can represent Boolean provenance as
Boolean formulae
\cite{ImielinskiL84,green2007provenance}, or (more recently) as
Boolean circuits \cite{deutch2014circuits,amarilli2015provenance}.
In this section, we first introduce \emph{monotone cycluits} 
(monotone Boolean circuits with cycles), for which we define a semantics
(in terms of the
Boolean function that they express);
we also show that cycluits can be
evaluated in linear time, given a valuation. 
Second, we extend them to \emph{stratified cycluits}, allowing a form of stratified
negation.
We conclude the section by showing how to construct the \emph{provenance}
of a SATWA as a cycluit, in FPT-linear time. Together with
Theorem~\ref{thm:maintheorem}, this claim implies our main provenance result:
\begin{theorem}
  \label{thm:mainprov}
  Given an ICG-Datalog program~$P$ of body size
  $\kp$ and a relational instance~$I$ of treewidth $\ki$,
  we can construct in FPT-linear time in~$|I|\cdot|P|$
  (parameterized by~$\kp$ and~$\ki$)
  a representation of the provenance
  of $P$ on $I$ as a stratified cycluit. Further,
  for fixed $\ki$, this cycluit has treewidth~$O(\card{P})$.
\end{theorem}

Of course, this result implies the analogous claims for query languages that are
captured by ICG-Datalog parameterized by the body size, as we studied in
Section~\ref{sec:ICG}. When combined with the fact that cycluits can
be tractably evaluated, it yields our main result,
Theorem~\ref{thm:main}. The rest of this section formally introduces
cycluits and proves Theorem~\ref{thm:mainprov}.

\subparagraph*{Cycluits.}
\begin{toappendix}
  \subsection{Cycluits}
  \label{apx:cycluits}
\end{toappendix}
We define \emph{cycluits} as Boolean circuits without the acyclicity requirement,
as in~\cite{riedel2012cyclic}. To avoid the problem of feedback loops, however,
we first study \emph{monotone cycluits}, and then cycluits with stratified
negation.

\begin{definition}
  \label{def:cycluits}
    A \emph{monotone Boolean cycluit} is a directed graph $C = (G,W,g_0,\mu)$ 
    where $G$ is the set of \emph{gates},
    $W\subseteq G^2$ is the set of directed edges called \emph{wires} (and written $g
    \rightarrow g'$),
    $g_0 \in G$ is the \emph{output gate}, and
    $\mu$ is the \emph{type} function mapping each gate
    $g \in G$ to one of $\inp$ (input gate, with no incoming wire in $W$), $\land$
  (AND gate) or $\lor$ (OR gate).
\end{definition}

We now define the semantics of monotone cycluits. 
A (Boolean) \emph{valuation} of~$C$ is a function $\nu : C_{\inp} \to \{0,1\}$ 
indicating the value of the input gates.
As for standard monotone circuits, a valuation yields 
an \emph{evaluation} $\nu' : C \to \{0,1\}$, that we will define shortly,
indicating the value of each gate under the valuation~$\nu$: we abuse notation
and
write
$\nu(C) \in \{0, 1\}$ for the \emph{evaluation result}, i.e.,
$\nu'(g_0)$ where 
$g_0$ is the
output gate of~$C$. The Boolean function \emph{captured} by
a cycluit $C$ is thus 
the Boolean
function $\phi$ on~$C_\inp$ defined by $\nu(\phi) \colonequals \nu(C)$ for each
valuation
$\nu$ of~$C_\inp$. 
We define the evaluation $\nu'$ from $\nu$ by a least fixed-point computation (see Algorithm~\ref{alg:semantics-monotone} in
Appendix~\ref{apx:cycluits}): we
set all input gates to their value by $\nu$, and other gates to $0$.
We then iterate until the evaluation no longer changes, by evaluating OR-gates
to~$1$ whenever some input evaluates to~$1$, and AND-gates to~$1$
whenever all their inputs evaluate to~$1$.
The Knaster--Tarski
theorem~\cite{tarski1955lattice} gives an equivalent characterization:

\begin{toappendix}
	The semantics of monotone cycluits is formally defined by
        Algorithm~\ref{alg:semantics-monotone}.
\begin{algorithm}
\DontPrintSemicolon
\KwIn{Monotone cycluit $C=(G,W,g_0,\mu)$, Boolean valuation $\nu: C_{\inp} \to
  \{0,1\}$}
        \KwOut{$\{g \in C \mid \nu(g)=1\}$}
	$S_0 \defeq \{ g \in C_\inp \mid \nu(g)=1 \}$\;
	$i \defeq 0$\;
	\Do{$S_i \neq S_{i-1}$}{
		$i$++\;
                $S_{i} \defeq S_{i-1} \cup\Big\{g \in C \mid (\mu(g) =
        {\lor}), \exists g^\prime \in S_{i-1}, g^\prime \rightarrow
      g\in W\Big\} \cup{}$\\$\quad \Big\{g \in C \mid (\mu(g) =
        {\land}), \{g^\prime \mid g^\prime \rightarrow g\in W\} \subseteq
    S_{i-1}\Big\}$}
	\Return $S_i$\;
	\caption{Semantics of monotone cycluits}
	\label{alg:semantics-monotone}
\end{algorithm}

\end{toappendix}

\begin{propositionrep}\label{prp:altersem}
  For any monotone cycluit $C$ and Boolean valuation $\nu$ of $C$, 
  the set $S \defeq \{g \in C \mid \nu'(g) = 1\}$ is \emph{the} minimal set of
  gates (under inclusion) such that:
  \begin{compactenum}[(i)]
    \item $S$ contains the true input gates, i.e., it contains $\{g \in C_\inp \mid
  \nu(g) = 1\}$;
\item for any $g$ such that $\mu(g) = \lor$, if some input gate of $g$ is
  in $S$, then $g$ is in~$S$;
\item for any $g$ such that $\mu(g) = \land$, if all input gates of $g$ are
  in $S$, then $g$ is in~$S$.
\end{compactenum}
\end{propositionrep}
\begin{proof}
  The operator used in Algorithm~\ref{alg:semantics-monotone} is clearly
  monotone, so by the Knaster--Tarski theorem, the outcome of the
  computation is the intersection of all set of gates satisfying the
  conditions in Proposition~\ref{prp:altersem}.
\end{proof}

We show that this definition is computable in linear time
(Algorithm~\ref{alg:linear-time-cycluits} in Appendix~\ref{apx:cycluits}):

\begin{toappendix}
	Algorithm~\ref{alg:semantics-monotone} is a naive fixpoint algorithm running in quadratic time, but we show that the same output can be computed in linear time with Algorithm~\ref{alg:linear-time-cycluits}.
\end{toappendix}
\begin{propositionrep}\label{prp:moncycluitlinear}
  Given any monotone cycluit $C$ and Boolean valuation $\nu\!$ of~$C$, we can compute the
  evaluation $\nu'\!$ of~$C$ in linear time.
\end{propositionrep}

\begin{proof}
	We use Algorithm~\ref{alg:linear-time-cycluits}.
  We first prove the claim about the running time. The preprocessing to compute
  $M$ is linear-time in~$C$ (we enumerate at most once every wire), and the rest of the algorithm is clearly
  in linear time as it is a variant of a DFS traversal of the graph, with the added
  refinement that we only visit nodes that evaluate to~$1$ (i.e., OR-gates
  with some input that evaluates to~$1$, and AND-gates where all inputs
  evaluate to~$1$).

  We now prove correctness. We use the characterization of
  Proposition~\ref{prp:altersem}. We first check that $S$ satisfies the
  properties:
  
  \begin{enumerate}[(i)]
    \item $S$ contains the true input gates by construction.
    \item Whenever an OR-gate $g'$ has an input gate $g$ in~$S$, then, when
      we added $g$ to~$S$, we have necessarily followed the wire $g \rightarrow
      g'$ and added $g'$ to~$Q$, and later added it to~$S$. 
    \item Whenever an AND-gate $g'$ has
  all its input gates $g'$ in~$S$, there are two cases. The first case in when
      $g$ has no input gates at all, in which case $S$ contains it by
      construction. The second case is where such input gates exist: in this
      case, observe that $M[g']$ was initially equal to the degree of~$g'$, and
      that we decrement it for each input gate $g$ of~$g'$ that we add to~$S$.
      Hence, considering the last input gate $g$ of~$g'$ that we add to~$S$, it
      must be the case that $M[g']$ reaches zero when we decrement it,
      and then we add $g$ to~$Q$, and later to~$S$.
  \end{enumerate}
  
  Second, we check that $S$ is minimal. Assume by contradiction that it
  is not the case, and consider the first gate $g$ which is
  added to $S$ while not being in the minimal Boolean valuation $S'$.  It cannot
  be the case that $g$ was added when initializing $S$, as we initialize $S$ to
  contain true input gates and AND-gates with no inputs, which must be true
  also in~$S'$ by the characterization of Proposition~\ref{prp:altersem}. Hence,
  we added $g$ to~$S$ in a later step of the algorithm. However,
  we notice that we must added $g$ to $S$ because of the value of its input
  gates. By minimality of~$g$, these input gates have the same value in~$S$ and
  in~$S'$. This yields a contradiction, because the gates that we add to~$S$ are
  added following the characterization of Proposition~\ref{prp:altersem}.
  This concludes the proof.
\end{proof}

\begin{toappendix}
\begin{algorithm}
\DontPrintSemicolon
\KwIn{Monotone cycluit $C=(G,W,g_0,\mu)$,  Boolean valuation $\nu: C_{\inp} \to \{0,1\}$}
        \KwOut{$\{g \in C \mid \nu(g)=1\}$}
        ~\tcc{Precompute the in-degree of $\land$ gates}
        \For{$g \in C$ s.t.\ $\mu(g) = \land$}{
          $M[g] \defeq \card{\{g' \in C \mid g' \rightarrow g\}}$\;
        }
        $Q\defeq\{g \in C_\inp \mid \nu(g)=1\} \cup \{g \in C \mid (\mu(C) =
        \land) \land M[g] = 0\}$ \tcc*{as a stack} $S\defeq \emptyset$ \tcc*{as a bit array}
	\While{$Q \neq \emptyset$}{
          pop $g$ from $Q$\;
                        \If{$g \notin S$}{
                        add $g$ to~$S$\;
                        \For{$g^\prime \in C \mid g \to g^\prime$}{
                                \If{$\mu(g^\prime)=\lor$}{
                                        push $g^\prime$ into $Q$\;
                                }
                                \If{$\mu(g^\prime)=\land$}{
                                        $M[g^\prime]\defeq M[g^\prime]-1$\;
                                        \If{$M[g^\prime]=0$}{
                                        push $g^\prime$ into $Q$\;
                                        }
                                }
                        }
                      }
                        
        }
	\Return $S$
        \caption{Linear-time evaluation of monotone cycluits}
	\label{alg:linear-time-cycluits}
\end{algorithm}
\end{toappendix}

\subparagraph*{Stratified cycluits.}
\begin{toappendix}
  \subsection{Stratified cycluits}
\end{toappendix}
We now move from monotone cycluits to general cycluits featuring negation.
However, allowing arbitrary negation would make it difficult to define a
proper
semantics, because of possible cycles of negations. Hence, we focus on
\emph{stratified cycluits}:

\begin{definition}
  A \emph{Boolean cycluit} $C$ is defined like a \emph{monotone cycluit},
  but further allows NOT-gates ($\mu(g) = \neg$), which are required to have a
  single input. It is 
  \emph{stratified} if there exists a
  \emph{stratification function}~$\strat$ mapping its gates surjectively
  to $\{0,\ldots,m\}$ for some $m \in \NN$ such that $\strat(g) = 0$ iff $g
  \in C_\inp$, and
  $\strat(g)
  \leq \strat(g')$ for each wire $g \rightarrow g'$, the inequality being strict
  if $\mu(g') = \neg$.
\end{definition}

Equivalently, $C$ contains no cycle of gates involving a $\neg$-gate. If $C$ is
stratified, we can
compute a stratification function in linear time by a topological sort, and use it to define
the evaluation of~$C$ (which will clearly be independent of the choice of stratification function):

\begin{toappendix}
  We show the claim that a Boolean cycluit is stratified iff it contains no cycle of gates involving a $\neg$-gate, and that
  a stratification function can be computed in linear time.
\begin{propositionrep}\label{prp:stratifun}
  Any Boolean cycluit $C$ is stratified iff it it contains no cycle of gates involving a $\neg$-gate.
Moreover, a stratification function can be computed in linear time from~$C$.
\end{propositionrep}

\begin{proof}
  To see why a stratified Boolean cycluit $C$ cannot contain a cycle of gates involving a $\neg$-gate,
  assume by contradiction that it has such a cycle $g_1 \rightarrow g_2
  \rightarrow \cdots \rightarrow g_n \rightarrow g_1$. As $C$ is stratified, there exists a stratification function $\strat$.
  From the properties of a stratification function, we know that $\strat(g_1) \leq \strat(g_2) \leq \cdots \leq
  \strat(g_1)$, so that we must have $\strat(g_1) = \cdots = \strat(g_n)$.
  However, letting $g_i$ be such that $\mu(g_i) = \neg$, we know that
  $\strat(g_{i-1}) < \strat(g_i)$ (or, if $i = 1$, $\strat(g_n) < \strat(g_1)$),
  so we have a contradiction.

  We now prove the converse direction of the claim, i.e., that any 
  Boolean cycluit which does not contain a cycle of gates involving a $\neg$-gate must have a stratification function, 
  and show how to compute such a function in linear time.
  Compute in linear time the strongly connected components (SCCs) of~$C$, and a
  topological sort of the SCCs. As the input gates of~$C$ do not themselves have
  inputs, each of them must have their own SCC, and each such SCC must be a
  leaf, so we can modify the topological sort by merging these SCCs
  corresponding to input gates, and putting them first in the topological sort.
  We define the function
  $\strat$ to map each gate
  of~$C$ to the index number of its SCC in the topological sort, which ensures
  in particular that the input gates of~$C$ are exactly the gates assigned
  to~$0$ by~$\strat$.
  This can be performed in linear time. Let us show that the result $\strat$ is
  a stratification function:
  \begin{itemize}
    \item For any edge $g \rightarrow g'$, we have $\strat(g) \leq \strat(g')$.
      Indeed, either $g$ and $g'$ are in the same strongly connected component and we
      have $\strat(g) = \strat(g')$, or they are not and in this case the edge
      $g \rightarrow g'$ witnesses that the SCC of $g$ precedes that of~$g'$,
      whence, by definition of a topological sort, it follows that $\strat(g) <
      \strat(g')$.
    \item For any edge $g \rightarrow g'$ where $\mu(g') = \lnot$, we have $\strat(g) <
      \strat(g')$. Indeed, by adapting the reasoning of the previous bullet
      point, it suffices to show that $g$ and $g'$ cannot be in the same
      SCC. Indeed, assuming by contradiction that they are, by
      definition of a SCC, there must be a path from $g'$ to $g$, and combining
      this with the edge $g \rightarrow g'$ yields a cycle involving a
      $\lnot$-gate,
      contradicting our assumption on~$C$.\qedhere
  \end{itemize}
\end{proof}
\end{toappendix}

\begin{definition}
  \label{def:strateval}
  Let $C$ be a stratified cycluit with stratification function $\strat: C
  \rightarrow \{0,\ldots,m\}$,
  and let $\nu$ be a Boolean valuation of~$C$. We inductively define the
  \emph{$i$-th stratum evaluation} $\nu_i$, for $i$ in the range of~$\strat$,
  by setting $\nu_0 \colonequals \nu$, and letting $\nu_{i}$ extend the $\nu_j$ ($j < i$) as
  follows:
  \begin{compactenum}
  \item For $g$ such that $\strat(g) = i$ with $\mu(g) = \neg$, set $\nu_i(g)
    \defeq \neg \nu_{\strat(g')}(g')$ for its one input $g'$.
  \item Evaluate all other $g$ with $\strat(g) = i$ as for monotone cycluits,
    considering the $\neg$-gates of point \textsf{\bfseries
    \textcolor{darkgray}{1.}}\ and all gates of lower
    strata as input gates fixed to their value in~$\nu_{i-1}$.
  \end{compactenum}
  Letting $g_0$ be the output gate of~$C$,
  the Boolean function $\phi$ \emph{captured} by~$C$ is then defined
  as $\nu(\phi) \colonequals \nu_m(g_0)$
  for each valuation $\nu$ of $C_\inp$.
\end{definition}

\begin{propositionrep}\label{prp:cycluitlinear}
  We can compute $\nu(C)$ in linear time in the stratified cycluit $C$ and in $\nu$.
\end{propositionrep}

\begin{proof}
	Compute in linear time a stratification function $\strat$ of $C$ using
        Proposition~\ref{prp:stratifun}, and compute the evaluation following Definition~\ref{def:strateval}.
	This can be performed in linear time.
        
        To see why this evaluation is independent from the choice of
        stratification, observe that any stratification function must clearly assign the
        same value to all gates in an SCC. Hence, the
        choosing a stratification function amounts to choosing
        the stratum that we assign to each SCC. Further,
        when an SCC $S$ precedes another SCC $S'$, the stratum of $S$ must be no
        higher than the stratum of~$S'$. So in fact the only freedom that we
        have is to choose a topological sort of the SCCs, and optionally to
        assign the same stratum to consecutive SCCs in the topological sort:
        this amounts to ``merging'' some SCCs, and is only possible when there
        are no $\lnot$-gates between them. Now, in the evaluation, it is clear
        that the order in which we evaluate the SCCs makes no difference, nor
        does it matter if some SCCs are evaluated simultaneously. Hence, the
        evaluation of a stratified cycluit is well-defined.
\end{proof}

\subparagraph*{Building provenance cycluits.}
\begin{toappendix}
  \subsection{Building provenance cycluits}
\end{toappendix}
Having defined cycluits as our provenance representation, we compute 
the provenance of a query on an instance as the 
\emph{provenance} of its SATWA on a tree encoding.
To do so, we must give a general definition of the provenance of SATWAs.
Consider a $\Gamma$-tree $\calT \defeq \la T, \lambda \ra$ for some
alphabet $\Gamma$, as in 
Section~\ref{sec:compilation}.
We define
a (Boolean) \emph{valuation}~$\nu$ of $\calT$ 
as a mapping from the nodes of~$T$ to $\{0, 1\}$.
Writing $\overline{\Gamma} \colonequals \Gamma \times \{0, 1\}$, 
each valuation~$\nu$ then defines a $\overline{\Gamma}$-tree
$\nu(\calT) \defeq \la T, (\lambda \times \nu) \ra$, obtained by annotating each
node of~$\calT$ by its $\nu$-image.
As in~\cite{amarilli2015provenance},
we define the provenance of a $\overline{\Gamma}$-SATWA $A$ on~$\calT$, which
intuitively captures all possible results of evaluating $A$ on possible
valuations of~$\calT$:

\begin{definition}
  \label{def:satwaprov}
  The \emph{provenance} of a $\overline{\Gamma}$-SATWA $A$ on a $\Gamma$-tree $\calT$
  is the Boolean function $\phi$ defined on the nodes of~$T$ such
  that, for any valuation $\nu$ of~$\calT$, 
  $\nu(\phi) = 1$ iff $A$
  accepts $\nu(\calT)$.
\end{definition}

We then show that we can efficiently build provenance representations of
SATWAs on trees as stratified cycluits:

\begin{theoremrep}\label{thm:satwaprov}
  For any fixed alphabet $\Gamma$, given a $\overline{\Gamma}$-SATWA $A$ and a
  $\Gamma$-tree $\calT$, we can build 
  a stratified cycluit capturing the provenance of~$A$
  on~$\calT$
  in time $O(\card{A} \cdot \card{\calT})$.
  Moreover, this stratified cycluit has treewidth $O(\card{A})$.
\end{theoremrep}

\begin{proofsketch}
The construction generalizes Proposition~3.1
of~\cite{amarilli2015provenance} from bNTAs and circuits to SATWAs and cycluits. 
For each node $w$ of $T$ and state $q$, we create a gate $g_w^q$ which, following the transitions of $A$,
 is connected to those created for the neighbors of $w$; $g_w^q$ is such
  that
for every Boolean valuation $\nu : T \to \{0,1\}$ of the inputs of $C$, there exists a run $\rho$ of $A$ on $\nu(T)$ starting at $w$ in state $q$ if and only if $\nu(g_w^q)=1$.
The reason why we need cycluits rather than circuits is because SATWAs may loop
back on previously visited nodes.
\end{proofsketch}

Note that the proof can be easily modified to make it work for standard
alternating two-way automata rather than our isotropic automata.
This result allows us to prove
Theorem~\ref{thm:mainprov}, by applying it to the SATWA obtained from the
ICG-Datalog program (Theorem~\ref{thm:maintheorem}), slightly modified
so as to extend it to the alphabet~$\overline{\Gamma}$. Recalling that nodes of
the tree encodings each encode 
at most one fact of the instance, we use the second coordinate of
$\overline{\Gamma}$ to indicate whether the fact is actually present or should
be discarded.
This allows us to range over possible subinstances, and thus to compute the provenance. This
concludes the proof of our main result (Theorem~\ref{thm:main} in
Section~\ref{sec:ICG}): we can evaluate an ICG-Datalog program on a treelike
instance in FPT-linear time by computing its provenance by
Theorem~\ref{thm:mainprov} and evaluating the provenance in linear time
(Proposition~\ref{prp:cycluitlinear}).

\begin{toappendix}
  To prove Theorem~\ref{thm:satwaprov}, we construct a cycluit $C^A_T$ as follows, generalizing the
construction
of~\cite{amarilli2015provenance}. 
For each node $w$ of $T$, we create an input node $g_w^{\i}$, a $\lnot$-gate
$g_w^{\lnot \i}$ defined as~$\NOT(g_w^{\i})$, and an OR-gate $g_w^q$ for each state $q \in
Q$. Now for each $g_w^q$, for $b \in \{0, 1\}$, we consider the propositional
formula $\Delta(q, (\lambda(w),b))$, and we express it as a circuit that
captures this formula:
we let
$g_w^{q,b}$ be the output gate of that circuit, we
replace each variable~$q'$ occurring positively by an
OR-gate $\bigvee_{w^\prime \in \neigh(w)} g_{w^\prime}^{q'}$, 
and we replace each variable $q'$ occurring negatively by the gate~$g_w^{q'}$.
We then define $g^q_w$ as
  $\OR(\AND(g_w^{\i},g_w^{q,0}), \AND(g_w^{\lnot
  \i},g_w^{q,1}))$. Finally, we let the output gate of $C$ be $ g_r^{q_{\I}}$, where $r$ is the root of $T$.

It is clear that this process runs in linear time in $\card{A} \cdot \card{T}$.
Moreover, for every node~$w$ of~$T$, we create $O(\card{A})$ gates, and those gates
  can only be connected between them or between gates created for the neighbors
  of~$w$. Hence, the treewidth of $C^A_T$ is $O(\card{A})$: we can compute
  a tree decomposition of~$C^A_T$ where the underlying tree is~$T$ and where
  each bag $b$ corresponding to a node $w$ of~$T$ contains the gates defined for
  the node~$w$ and for the neighbors of~$w$.
The proof of Theorem~\ref{thm:satwaprov} then follows from the following claim:

\begin{lemmarep}\label{lem:goodprov}
  The cycluit $C^A_T$ is a stratified cycluit capturing the provenance of $A$
  on~$\la T,\lambda\ra$.
\end{lemmarep}

\begin{proof}
  We first show that $C \colonequals C^A_T$ is a stratified cycluit.
  Let $\strat$ be the stratification function
  of the $\overline{\Gamma}$-SATWA $A$ and let $\{0, \ldots, m\}$ be its range.
  We use $\strat$ to define $\strat'$ as the following function from the gates
  of~$C$ to $\{0, \ldots, m+1\}$:
\begin{itemize}
	\item For any input gate $g_w^\i$, we set
          $\strat^\prime(g_w^\i)\colonequals 0$ and
          $\strat^\prime(g_w^{\lnot \i})\colonequals 1$.
        \item For an OR gate $g \colonequals \bigvee_{w^\prime \in \neigh(w)} g_{w^\prime}^{q'}$, we set
          $\strat'(g) \colonequals \strat(q')$.
	\item For any state $g_w^q$, we set $\strat^\prime(g_w^n) \colonequals \strat(q) +
          1$, and do the same for the intermediate AND-gates used in its
          definition, as well as the gates in the two circuits that capture the
          transitions $\Delta(q, (\lambda(w), b))$ for $b \in \{0, 1\}$, except
          for the input gates of that circuit (i.e., gates of the form $\bigvee_{w^\prime \in \neigh(w)}
          g_{w^\prime}^{q'}$, which are covered by the previous point, or
          $g^{q'}_w$, which are covered by another application of that point).
\end{itemize}

Let us show that $\strat'$ is indeed a stratification function for~$C$. We first
  observe that it is the case that the gates in stratum zero are precisely the
  input gates. We then check the condition for the various possible wires:

\begin{itemize}
    \item $g^\i_w \rightarrow g^{\lnot\i}_w$: by construction, we have
      $\strat(g^\i_w) < \strat'(g^{\lnot\i}_w)$.
    \item $g \rightarrow g'$ where $g'$ is a gate of the form $g_w^q$ and $g$ is an
      intermediate AND-gate in the definition of a $g_w^q$: by construction we
      have $\strat'(g) = \strat'(g')$, so in particular $\strat'(g) \leq
      \strat'(g')$.
    \item $g \rightarrow g'$ where $g'$ is an intermediate AND-gate in the
      definition of a gate of the form $g_w^q$, and $g$ is $g^\i_w$ or
      $g^{\lnot\i}_w$: by construction we have $\strat'(g) \in \{0, 1\}$ and
      $\strat'(g') \geq 1$, so $\strat'(g) \leq \strat'(g')$.
    \item $g \rightarrow g'$ where $g$ is a gate in a circuit capturing
      the propositional formula of some transition of $\Delta(q, \cdot)$ without
      being an input gate or a NOT-gate of this circuit, and $g'$ is also such a
      gate, or is an intermediate AND-gate in the definition of~$g_w^q$: then
      $g'$ cannot be a NOT-gate (remembering
      that the propositional formulae of transitions only have negations on
      literals), and
      by construction we have $\strat'(g) = \strat'(g')$.
    \item $g \rightarrow g'$ where $g$ is of the
	 form $\bigvee_{w^\prime \in \neigh(w)} g_{w^\prime}^{q}$, and $g'$ is
         a gate in a circuit describing $\Delta(q', \cdot)$ or an intermediate
         gate in the definition of $g^{q'}_w$.
         Then we have $\strat'(g) = \strat(q)$ and $\strat'(g') = \strat(q')$,
         and as $q$ occurs as a positive literal in a transition of~$q'$, by
         definition of $\strat$ being a transition function, we have $\strat(q)
         \leq \strat(q')$. Now we have $\strat'(g) = \strat(q)$ and $\strat'(g')
         = \strat'(q')$ by definition of~$\strat'$, so we deduce that
         $\strat'(g) \leq \strat'(g')$.
    \item $g \rightarrow g'$ where $g'$ is of the
	 form $\bigvee_{w^\prime \in \neigh(w)} g_{w^\prime}^{q'}$, and $g$ is
         one of the $g_{w'}^{q'}$. Then by definition of~$\strat'$ we have
         $\strat'(g) = \strat(q')$ and $\strat'(g') = \strat(q')$, so in
         particular $\strat'(g) \leq \strat'(g')$.
    \item $g \rightarrow g'$ where $g$ is a NOT-gate in a circuit
      capturing a propositional formula $\Delta(q',
      (\lambda(w), b))$, and $g$ is then necessarily a gate of the form
      $g^{q}_w$: then clearly $q'$ was
      negated in $\phi$ so we had $\strat(q) < \strat(q')$, and as by
      construction we have $\strat'(g) = \strat(q)$ and $\strat'(g') =
      \strat(q')$, 
      we deduce that $\strat'(g) < \strat'(g')$.
\end{itemize}

We now show that $C$ indeed captures the provenance of~$A$ on~$\la T,\lambda\ra$.
Let $\nu : T \to \{0,1\}$ be a Boolean valuation of the inputs of $C$, that we
extend to an evaluation $\nu' : C \to \{0,1\}$ of~$C$.
We claim the following \textbf{equivalence}: for all $q$ and $w$, there exists a run $\rho$ of $A$ on $\nu(T)$ \emph{starting at $w$} in state $q$ if and only if $\nu(g_w^q)=1$.

We prove this claim by induction on the stratum $\strat(q)$ of~$q$. Up to adding
  an empty first stratum, we can make sure that the base case is vacuous.
For the induction step, we prove each implication separately.

  \subparagraph*{\qquad Forward direction.}
First, suppose that there exists a run $\rho$ starting at $w$ in state $q$, and
let us show that $\nu'(g^q_w) = 1$. We show by induction on the run (from
bottom to top) that for each node $y$ of the run labeled by a \emph{positive} state
$(q', w')$ we have $\nu'(g^{q'}_{w'}) = 1$, and for every node $y$ of the run
labeled by a \emph{negative} state $\lnot (q', w')$ we have $\nu'(g^{q'}_{w'})
= 0$. The base case concerns the leaves, where there are three possible
subcases:

\begin{itemize}
  \item We may have $\lambda_r(y)=(q^\prime, w')$ with $\strat(q^\prime) = i$, so
  that $\Delta(q^\prime,(\lambda(w'),\nu(w')))$ is tautological. In
    this case, $g_{w'}^{q^\prime}$ is defined as
    $\OR(\AND(g_{w'}^{\i},g_{w'}^{q',1}),\AND(g_{w'}^{\lnot
    \i},g_{w'}^{q',0}))$. Hence, we know that $\nu(g_{w'}^{q^\prime,\nu(w)}) =
    1$ because the circuit is also tautological, and depending on whether
    $\nu(w)$ is~$0$ or~$1$ 
    we know that $\nu(g_{w'}^{\lnot\i})=1$ or $\nu(g_{w'}^{\i})=1$, so this
    proves the claim.
		
  \item We may have $\lambda_r(y)=(q^\prime, {w'})$ with $\strat(q^\prime) = j$ for $j < i$. 
		By definition of the run $\rho$, this implies that there exists
                a run starting at ${w'}$ in state $q^\prime$. 
		But then, by the induction on the strata (using the forward
                direction of the \textbf{equivalence}), we must have $\nu(g_{w'}^{q^\prime})=1$.

  \item We may have $\lambda_r(y)=\lnot (q^\prime, {w'})$ with $\strat(q^\prime)
                = j$ for $j < i$. 
                Then by definition there exists no run starting at ${w'}$ in state $q^\prime$.
                Hence again by induction on the strata (using the backward
                direction of the \textbf{equivalence}), we have that $\nu(g_{w'}^{q^\prime})=0$.
\end{itemize}

For the induction case on the run, where $y$ is an internal node,
by definition of a run there is a subset $S = \{q_{P_1},\cdots,q_{P_n}\}$ of
positive literals and a subset $N = \{ \lnot q_{N_1}, \cdots, \lnot q_{N_m}\}$
of negative literals that satisfy $\phi_{\nu(w')} \defeq \Delta(q^\prime,(\lambda(w'),\nu(w')))$ such
that:

\begin{itemize}
  \item For all $q_{P_k} \in P$, there exists a child $y_k$ of~$y$ with $\lambda_r(y_k)=(q_{P_k}, {w'}_k)$
    where ${w'}_k \in \neigh({w'})$;
  \item For all $\lnot q_{N_k}
\in N$ there is a child $y_{{w'}_k}$
of $y$ with
$\lambda_r(y_{w'})=\lnot (q_{N_k}, {w'})$.
\end{itemize}
    
    Then, by induction on the run, we
know that for all $q_{P_k}$ we have $\nu(g_{{w'}_k}^{q_{P_k}})=1$ and for all
$\lnot q_{N_k}$ we have
$\nu(g_{{w'}}^{q_{N_k}})=0$. By construction of $C$, we have
$\nu(g_{w'}^{q'})=1$. 
There are two cases:
either $\nu(w') = 1$ or $\nu(w') = 0$. In the
first case, remember that the first input of the OR-gate $g_{w'}^{q'}$ is an AND-gate of
$g^{\i}_{w'}$ and the output gate $g_{w'}^{q',1}$ of a circuit coding $\phi_{1}$ on inputs including
the $g^{q_{P_k}}_{w'_k}$ and $g^{q_{N_k}}_{w''}$. We have $\nu(g^{\i}_{w'}) = 1$ because $\nu(w') = 1$,
and the second gate evaluates to~$1$ by construction of the circuit, as
witnessed by the Boolean valuation of the $g^{q_{P_k}}_{w'_k}$ and $g^{q_{N_k}}_{w''}$.
In the second case we follow the same reasoning but with the second input
to~$g^{q'}_{w'}$ instead, which is an AND-gate on
$g^{\lnot \i}_{w'}$ and $\phi_{0}$.

By induction on the run, the claim is proven, and applying it to the root of the
run concludes the proof of the first direction of the \textbf{equivalence} (for the
induction step of the induction on strata).

  \subparagraph*{\qquad Backward direction.}
We now prove the converse implication for the induction step of the induction on
strata, i.e., letting $i$ be the current stratum, for every node $w$ and state
$q$ with $\strat(q) = i$, if $\nu(g^q_w) = 1$ then there exists a run $\rho$
of~$A$ starting at~$w$. From the definition of the stratification
function~$\strat'$ of the cycluit from~$\strat$, we have $\strat'(g^q_w) =
\strat(q)$, so as $\nu(g^q_w) = 1$ we know that $\nu_i(g^q_w) = 1$, where
$\nu_i$ is the $i$-th
stratum evaluation $\nu_i$ of~$C$ (remember Definition~\ref{def:strateval}).
By induction hypothesis on the strata, we know from the \textbf{equivalence} that, for any $j < i$, for any
gate $g^{q''}_{w''}$ of $C$ with $\strat(g^{q''}_{w''}) = j$, we have $\nu_j(g^{q''}_{w''}) = 1$ iff 
there exists a run $\rho$ of $A$ on $\nu(T)$ \emph{starting at $w''$}
in state $q''$.

We show the claim by an induction on the iteration in the application of
Algorithm~\ref{alg:semantics-monotone} for $\nu_i$ where the gate $g^q_w$ was
set to~$1$. The base case concerns gates that were initially true before
applying the algorithm: by 

Recall that the definition of $\nu_i$ according to
Definition~\ref{def:strateval} proceeds in three steps. Initially, we
fix the value in $\nu_i$ of gates of
lower strata, so we can then conclude by induction hypothesis on the strata. We
then set the value of all NOT-gates in~$\nu_i$, but these cannot be of the form
$g^{q'}_{w'}$ so there is nothing to show. Last, we evaluate all other gates
with Algorithm~\ref{alg:semantics-monotone}. We then show our claim by an
induction on the iteration in the application of
Algorithm~\ref{alg:semantics-monotone} for $\nu_i$ where the gate $g^q_w$ was
set to~$1$. The base case, where $g^q_w$ was initially true, was covered in the
beginning of this paragraph.

For the induction step on the application of
Algorithm~\ref{alg:semantics-monotone},
when a gate $\nu_i(g^{q'}_{w'})$ is set to true, as $\nu_i(g^{q'}_{w'})$ is an
OR-gate by construction, from the workings of Algorithm~\ref{alg:semantics-monotone}, there are two possibilities:
either its input AND-gate that includes $g^{\i}_{w'}$ was true,
or its input AND-gate that includes $g^{\neg\i}_{w'}$ was true. We prove the
first case, the second being analogous. From the fact that $g^{\i}_{w'}$ is
true, we know that $\nu(w') = 1$. Consider the other input gate to that AND
gate, which is the output gate of a circuit $C'$ reflecting 
$\phi \defeq \Delta(q^\prime,(\lambda(w'),\nu(w')))$, with the input gates
adequately substituted. We consider the value by $\nu_i$ of the gates that
are used as input gates of $C'$ in the construction of~$C$ (i.e., OR-gates, in
the case of variables that occur positively, or directly $g^{q''}_{w'}$-gates,
in the case of variables that occur negatively). By construction of
$C'$, the corresponding Boolean valuation $\nu'$ is a witness to the satisfaction of $\phi$. By induction hypothesis on
the strata (for the negated inputs to $C'$; and for the non-negated inputs to
$C'$ which are in a lower stratum) and on the step at which the gate was
set to true by Algorithm~\ref{alg:semantics-monotone} (for the inputs in the same stratum, which must be positive), the
valuation of these inputs reflects the existence of the corresponding runs.
Hence, we can assemble these (i.e., a leaf node in the first two cases, a run in
the third case) to obtain a run starting at $w'$ for state $q'$
using the Boolean valuation $\nu'$ of the variables of~$\phi$; this valuation satisfies
$\phi$ as we have argued.

This concludes the two inductions of the proof of the \textbf{equivalence} for the
induction step of the induction on strata, which concludes the proof.
\end{proof}

\subparagraph*{Putting it together.}
We now conclude the proof of Theorem~\ref{thm:mainprov} by explaining how this
provenance construction for $\overline{\Gamma}$-SATWAs can be used to compute
the provenance of an ICG-Datalog query on a treelike instance. This is again
similar to~\cite{amarilli2015provenance}.

Recall the definition of tree encodings from
Appendix~\ref{apx:tree-encodings}, and the
definition of the alphabet~$\Gamma^k_\sigma$. To represent the dependency of
automaton runs on the presence of individual facts, we will be working with
$\overline{\Gamma^k_\sigma}$-trees, where the Boolean annotation on a node~$n$ indicates
whether the fact coded by~$n$ (if any) is present or absent. The semantics
is that we map back the result to $\Gamma^k_\sigma$ as follows:

\begin{definition}
  We define the mapping $\epsilon$ from $\overline{\Gamma^k_\sigma}$ to
  $\Gamma^k_\sigma$ by:
  \begin{itemize}
    \item $\epsilon((d, s), 1)$ is just $(d, s)$, indicating that the fact
      of~$s$ (if any) is kept;
    \item $\epsilon((d, s), 0)$ is $(d, \emptyset)$, indicating that the fact of
      $s$ (if any) is removed.
  \end{itemize}

  We abuse notation and also see $\epsilon$ as a mapping from
  $\overline{\Gamma^k_\sigma}$-trees to $\Gamma^k_\sigma$-trees by applying it
  to each node of the tree.
\end{definition}

As our construction of provenance applies to automata on
$\overline{\Gamma^k_\sigma}$, we show the following easy \emph{lifting lemma}
(generalizing Lemma~3.3.4 of~\cite{amarilli2016leveraging}):

\begin{lemmarep}\label{lem:lifting}
  For any $\Gamma^k_\sigma$-SATWA $A$, we can compute in linear time a
  $\overline{\Gamma^k_\sigma}$-SATWA $A'$ such that, for any
  $\overline{\Gamma^k_\sigma}$-tree $E$, we have that $A'$ accepts $E$ iff $A$
  accepts $\epsilon(E)$.
\end{lemmarep}

\begin{proof}
  The proof is exactly analogous to that of  Lemma~3.3.4
  of~\cite{amarilli2016leveraging}.
\end{proof}

We are now ready to conclude the proof of our main provenance result
(Theorem~\ref{thm:mainprov}):

\begin{proof}[Proof of Theorem~\ref{thm:mainprov}]
  Given the program $P$ and instance $I$, use
Theorem~\ref{thm:maintheorem} to compute in
FPT-linear time in $\card{P}$ a $\Gamma^k_\sigma$-SATWA $A$ that tests it on tree
encodings of width $\leq k_{\text{I}}$, for $k_{\text{I}}$ the treewidth bound.
Compute also in FPT-linear time a tree encoding $E$ of the instance $I$, i.e., a
$\Gamma^k_\sigma$-tree $E$, using Lemma~\ref{lem:getencoding}.
Lift the  $\Gamma^k_\sigma$-SATWA $A$ in linear time using
Lemma~\ref{lem:lifting} to a $\overline{\Gamma^k_\sigma}$-SATWA~$A'$, and use Theorem~\ref{thm:satwaprov} on $A'$ and $E$ to compute in
FPT-linear time a stratified cycluit $C'$ that captures the provenance of~$A'$
on~$E$: the inputs of~$C'$ correspond to the nodes of~$E$. 
  The treewidth of~$C'$ is in $O(A')$, so it is FPT-linear in~$\card{P}$.
  Let $C$ be obtained
from $C'$ in linear time by changing the inputs of $C'$ as follows: those which
correspond to nodes $n$ of~$E$ containing a fact (i.e., with label $(d, s)$ for
$\card{s} = 1$) are renamed to be an input gate that stands for the fact of~$I$
coded in this node; the nodes $n$ of~$E$ containing no fact are replaced by a
0-gate, i.e., an OR-gate with no inputs. Clearly, $C$ is still a stratified
Boolean cycluit that satisfies the required treewidth bound,
  and $C_\inp$ is exactly the set of facts of~$I$.

All that remains to show is that $C$ captures the provenance of $P$ on~$I$ in
the sense of Definition~\ref{def:provenance}. To see why, consider an arbitrary
Boolean valuation $\nu$ mapping the facts of~$I$ to $\{0, 1\}$, and call $\nu(I) \defeq
\{F \in I \mid \nu(F) = 1\}$. We must show that $\nu(I)$ satisfies $P$ iff
$\nu(C) = 1$. By construction of~$C$, it is obvious that $\nu(C) = 1$ iff
$\nu'(C') = 1$, where $\nu'$ is the Boolean valuation of $C'_\inp$ defined by $\nu'(n) =
\nu(F)$ when $n$ codes some fact $F$ in $E$, and $\nu'(n) = 0$ otherwise. By
definition of the provenance of $A'$ on~$E$, we have $\nu'(C') = 1$ iff $A'$
accepts $\nu'(E)$, that is, by definition of lifting, iff $A$ accepts
$\epsilon(\nu'(E))$.
Now all that remains to observe is that $\epsilon(\nu'(E))$ is precisely a tree
encoding of the instance $\nu(I)$: this is by definition of $\nu'$ from $\nu$,
and by definition of our tree encoding scheme (see ``subinstance-compatibility''
in \cite{amarilli2016leveraging}). Hence, by definition of $A$ testing $P$, the tree
$\epsilon(\nu'(E))$ is accepted by $A$ iff $\nu(I)$ satisfies $P$. This finishes
the chain of equivalences, and concludes the proof of
Theorem~\ref{thm:mainprov}.
\end{proof}
\end{toappendix}

\section{From Cycluits to Circuits and Probability Bounds}
\label{sec:rmcycles}
We have proven our main result on ICG-Datalog, Theorem~\ref{thm:main}, in the
previous section, introducing stratified cycluits in the process as a way to
capture the provenance of ICG-Datalog.
In this section, we study how these stratified cycluits can be transformed
into \emph{equivalent} acyclic Boolean circuits, and we then show how we can use
this to derive bounds for the \emph{probabilistic query evaluation} problem (PQE).

\begin{toappendix}
\subsection{From cycluits to circuits}
\end{toappendix}
\subparagraph*{From cycluits to circuits.}
We call two cycluits or
circuits $C_1$ and $C_2$ \emph{equivalent} if they have the same set of
inputs $C_\inp$ and, for each valuation $\nu$ of $C_\inp$, we have
$\nu(C_1) = \nu(C_2)$.
A first result from existing work is that we can remove cycles in cycluits and
convert them to circuits, with a
quadratic blowup, by creating linearly many copies to materialize the fixpoint
computation. This allows us to remain FPT in combined complexity, but not FPT-linear:

\begin{proposition}[(\cite{riedel2012cyclic}, Theorem~2)]
  \label{prp:rmcyclequad}
  For any stratified cycluit $C$, we can compute in time $O(\card{C}^2)$ a Boolean circuit
  $C'$ which is equivalent to~$C$.
\end{proposition}

In addition to being quadratic rather than linear, another disadvantage
of this approach is that bounds on the treewidth
of the cycluit (which we will need later for probability computation) are generally not preserved on the output.
Hence, we prove a second cycle removal result, that proceeds in FPT-linear time when
parameterized by the treewidth of the input cycluit. When we use this result, we no
longer preserve FPT combined complexity of the overall computation,
because the stratified
cycluits produced by Theorem~\ref{thm:mainprov} generally have treewidth
$\Omega(|P|)$. On the other hand, we obtain an FPT-linear data complexity
bound, and
a bounded-treewidth circuit as a result. 

\begin{theoremrep}\label{thm:rmcycle}
  There is an $\alpha\in\mathbb{N}$ s.t.,
  for any stratified cycluit $C$ of treewidth $k$, we can compute in time
  $O(2^{k^\alpha} \card{C})$ a circuit $C'$ which is equivalent to~$C$ and has
  treewidth $O(2^{k^\alpha})$.
\end{theoremrep}

\begin{toappendix}
  The proof of Theorem~\ref{thm:rmcycle} is quite technical and proceeds in
several successive stages, corresponding to the following sections. In all
sections but the last one, we focus on \emph{monotone cycluits}. Throughout the
whole proof, we assume without loss of generality that there are no
\emph{isolated gates}, i.e., gates that have no input gate and being the input to
no gate (because such gates are not necessarily reflected in tree decompositions). We call a \emph{0-gate} an OR-gate with no inputs, and a \emph{1-gate}
an AND-gate with no inputs.

\subsubsection{Rewriting to arity-two cycluits}

For technical convenience, it will be easier to work with \emph{arity-two
cycluits}:

\begin{definition}
  The \emph{fan-in} of a gate $g$ in a monotone cycluit $C$ is the number of gates $g'$
  such that $g' \wire g$ is a wire of~$C$ (note that, as we are working with
  cycluits, this may include $g' = g$, i.e., $g$ may be its own input). For $c \in \NN$, a monotone cycluit $C$
  is \emph{arity-$c$} if
  each AND- and OR-gate has fan-in at most $c$, and if no gate is its own input
  (i.e., there is no wire of the form $g \wire g$).
\end{definition}

It is obvious that monotone cycluits, just like circuits, can be rewritten in linear time
to an \emph{arity-two} monotone cycluit, by taking advantage of the associativity of the OR and
AND Boolean operations: we just rewrite each AND- and OR-gate with fan-in $>2$
to a tree of gates of fan-in $\leq 2$ that computes the required Boolean
operation on the original set of inputs.

What is more technical to show, however, is that this rewriting operation can be
performed on a \emph{treelike} cycluit, while ensuring that the treewidth of the
result is still only a function of the treewidth of the original cycluit.
Hence, the result that we wish to prove is the following:

\begin{lemma}
  \label{lem:getarity2}
  For any monotone cycluit $C$ and tree decomposition $T$ of width~$k$ of~$C$,
  we can compute in time $O(\card{C} + \card{T})$ a monotone arity-two cycluit
  $C'$ which is equivalent to~$C$, and a tree decomposition $T'$ of width
  $k^2$ of~$C'$.
\end{lemma}

\begin{proof}
  First, to ensure the condition that no gate of the circuit is its own input,
  observe that any wire of the form $g \rightarrow g$ can be dropped if $g$ is an OR-gate
  with other inputs; if $g$ is an OR-gate whose only input is itself, then it
  will never evaluate to~$1$ so we can replace it by a $0$-gate, and if $g$
  is an AND-gate with a wire $g \rightarrow g$ then we can replace $g$ by a $0$-gate as
  it will never evaluate to~$1$. This can be done in linear time and does not
  increase the width of the tree decomposition.

  We assume without loss of generality that the tree decomposition~$T$ has arity
  at most two (i.e., each bag has at most two children), as we can otherwise
  ensure it by rewriting~$T$ in linear time (by replacing each bag with more
  than two children with a hierarchy of bags with the same contents and with two
  children). We assign each wire $g \rightarrow g'$ of the circuit to some
  bag $b$ of~$T$ such that $g, g' \in \dom(b)$: this is doable in linear time by
  Lemma~3.1 of~\cite{flum2002query}.
  We now perform a bottom-up
  traversal of the tree decomposition~$T$.
  In this process, we will annotate bags $b$ of the tree decomposition with sets
  of pairs of gates $g \mapsto g'$, indicating that $g$ is to be renamed to $g'$
  in the subtree of~$T$ rooted at~$b$.
  During the bottom-up traversal, we memorize, for each gate $g$ in the
  current bag $b$, its fan-in for the wires that we have seen so far: it is
  initially $0$, and is computed at bag $b$  as the sum of the fan-in of~$g$ at
  all the child bags (if any) plus its fan-in for the wires assigned to bag~$b$.
  Whenever the fan-in of a gate $g$ at a bag $b$ exceeds two, we do the following:
  we create one gate $g_1$ and label the first child bag of~$b$ (if it exists)
  with $g \mapsto g_1$,
  create one gate $g_2$ and label the second child bag of~$b$ (if it exists)
  with $g \mapsto g_2$, and create a gate $g'$ for the wires of the current bag. We let the
  input gates of~$g'$ be the input gates present at the current bag, and rewrite
  this to arity-two, which introduces $\leq k-1$ new gates in total.
  Now, we set the input gates of $g$ to be $g''$ and $g_1$, where $g''$ is a
  fresh gate for the same operation as~$g$, and set the input gates of~$g''$ to
  be $g_2$ and $g'$. We have introduced $\leq k$ new gates in the current bag
  per gate initially in~$b$, so the new size of~$b$ is at most~$k^2$.

  Once this pass is done, we perform a top-down pass to perform the
  renamings indicated on the bags, which is doable in linear time because the number of renamings to
  perform at any stage is constant (it is bounded by the width of the bags). 

  It is clear that this process produces an arity-two cycluit equivalent to~$C$,
  and the rewritten tree decomposition $T'$ has width $k^2$.
\end{proof}

\subsubsection{Regrouped tree decompositions}

Now that we can rewrite monotone cycluits to arity-two monotone cycluits in
linear time with only a quadratic blowup in the treewidth, we show that we can
impose an additional condition on tree decompositions of arity-two circuits,
which we call being \emph{regrouped}:

\begin{definition}
  A tree decomposition $T$ of a cycluit $C = (G, W, g_0, \mu)$ is \emph{regrouped} if it satisfies
  the following property: for any gate $g \in G$, letting $\ins(g) \defeq \{g' \in G
  \mid g' \wire g \in W\}$, there is a bag $b_g$ of~$T$ such that $\{g\} \cup
  \ins(g) \subseteq \dom(b)$.
\end{definition}

In other words, a tree decomposition is regrouped if each gate has a witness bag
which contains that gate and \emph{all} of its inputs. This can be contrasted
with the usual notion of tree decomposition, where each of the input wires can
be witnessed by a different bag. We show that, for arity-two cycluits, we can
assume that tree decompositions are regrouped in this sense:

\begin{lemma}
  \label{lem:getregrouped}
  For any monotone arity-two cycluit $C$ and tree decomposition $T$ of~$C$ of
  width~$k$, we
  can compute in time $O(\card{C} + \card{T})$ a regrouped tree decomposition
  $T'$ of~$C$ of width $\leq 3k$.
\end{lemma}

\begin{proof}
  We define $T'$ as having same skeleton as~$T$ and define each $b' \in T'$,
  letting $b$ be the corresponding bag in~$T$, to have $\dom(b') = \dom(b) \cup
  \bigcup \{\ins(g) \mid g \in b\}$. In other words, we add to each bag in~$T'$ the
  input gates of the gates that occur there. It is clear that $T'$ then
  satisfies the regrouped condition: take as witnessing bag $b_g$ of a gate
  $g$ any bag $b'$ of $T'$ corresponding to a bag $b \in T$ such that $g \in
  \dom(b)$. It is clear that $T'$ has width at most $3k$, as each bag $b'$
  of $T'$ contains at most the contents of the corresponding bag $b$ of~$T$ plus
  two additional input gates for each gate occurring in~$b$, thanks
  to the fact that $C$ is arity-two.

  As the regrouped condition subsumes the condition of tree decompositions that
  requires an occurrence of every wire, the only thing left to show is that the
  occurrences of any gate~$g$ in~$T'$ still forms a connected subtree. As $T$ is
  a tree decomposition of~$C$, let $T_g$ be the connected subtree containing the
  occurrences of~$g$ in~$T$, and, for every $g'$ such that $g \in \ins(g')$, let $T_{g'}$ be the
  connected subtree of the occurrences of~$g'$ in~$T$. For any $g'$ such that $g
  \in \ins(g')$, as
  $T$ is a tree decomposition, it must witness the wire $g \wire g'$, so the
  subtrees $T_g$ and $T_{g'}$ intersect in~$T$. Now, the subtree $T'_g$ of the
  occurrences of~$g$ in~$T'$ is by construction the analogue in~$T'$ of the
  union of the connected subtrees $T_g$ and $T_{g'}$ for $g'$ such that $g \in
  \ins(g')$. We
  conclude because the union of connected subtrees that intersect is also a
  connected subtree.
\end{proof}

\subsubsection{Normal-form tree decompositions}

Having presented the notion of regrouped tree decompositions, we now introduce
the notion of \emph{normal-form} tree decompositions, which satisfy the
regrouped conditions and satisfy other additional conditions. Unlike what we
have done so far, for technical reasons that will become apparent in the next
section,
we will define these tree decompositions as \emph{undirected}
connected trees rather than \emph{rooted} trees: this makes no difference, as we
used to consider tree decompositions as rooted trees simply for convenience. We
call $T$ an \emph{unrooted} tree decomposition if it is a tree decomposition
whose underlying tree is unrooted, and, for $b \in T$, the \emph{rooting} of~$T$
at bag~$b$ is a tree decomposition, in the usual rooted sense, obtained by
picking $b$ as the root in this tree; this always yields a valid tree
decomposition according to our earlier definitions.

\begin{definition}
  \label{def:normaldec}
  An unrooted tree decomposition $T$ of a monotone cycluit $C$ is in \emph{normal form} if
  it satisfies the following additional conditions:
  \begin{enumerate}
    \item Each bag $b$ of~$T$ has (undirected) degree either~$1$ or~$3$; we call
      $b$ a \emph{leaf bag} if it has degree~$1$, and an \emph{internal bag} if
      it has degree~$3$.
  \item There is an injective mapping $\phi$ from each gate $g$ of~$C$ to a leaf bag
    $b_g$ with $\dom(b_g) = \{g\} \cup \ins(g)$, where $\ins(g)$ are the input gates
  of~$g$. The
      \emph{$\phi$-occurrences} of~$g$ are the leaf bags $\{\phi(g)\} \sqcup
      \{\phi(g') \mid g' \in \ins(g)\}$.
  \item The tree decomposition is \emph{scrubbed}, meaning that:
    \begin{itemize}
      \item if a gate $g$ occurs in a leaf bag $b$, then $b$ is a
        $\phi$-occurrence of~$g$;
      \item if a gate $g$
    occurs in a non-leaf bag $b$ then there are two $\phi$-occurrences $b_1 \neq
      b_2$ of~$g$ and two neighbor bags $b_1' \neq b_2'$ 
      of~$b$ (note that we may have $b_i = b_i'$),
        such that, for $i \in \{1, 2\}$, the unique simple path from $b$ to~$b_i$ goes through
      $b_i'$.
    \end{itemize}
  \end{enumerate}
\end{definition}

Observe that the first condition ensures that, for any leaf bag $b$ 
of~$T$, the rooting of~$T$ at bag~$b$ is a tree decomposition where the root
bag~$b$ has exactly one child, and all other bags have exactly two children
except the leaf bags which have zero. The second condition implies that $T$, or
indeed any rooting of~$T$, is regrouped, with the witnesses for regrouping being chosen
at the leaf bags in an injective fashion. The third condition
intuitively ensures that each gate $g$ occurs in bags of~$T$ precisely in its
$\phi$-occurrences and on the 
paths between the $\phi$-occurrences of~$g$. In particular, the first point of
the third condition implies that any leaf bag with non-empty domain must be in
the image of~$\phi$. We also call a rooted tree decomposition \emph{scrubbed}
when it satisfies the analogous condition (with undirected paths).

We now show:

\begin{lemma}
  \label{lem:getnormal}
  For any monotone arity-two cycluit $C$ and regrouped tree decomposition $T$
  of~$C$ of width~$k$, we can compute in time $O(\card{C} + \card{T})$ a
  normal-form tree decomposition $T'$ of~$C$ of width~$k$.
\end{lemma}

\begin{proof}
  We modify $T$ in linear time in the following way to satisfy the second
  condition.
  For each gate $g$, letting $b_g$ be the witness bag for~$g$ guaranteed by
  replace $b_g$ by a new bag
  $b_g'$ with $\dom(b_g') = \dom(b_g)$, whose first child is $b_g$ and whose
  second child is a bag $b_g''$ whose domain is~$(\{g\} \cup
  \ins(g))$ (note that this is a subset of $\dom(b_g)$).
  We can now set
  $\phi(g)$ to be $b_g''$, with $\phi$ being injective.

  We then rewrite~$T$ in linear time in the usual way to ensure that each bag
  has at most two children.

  We next make $T$ scrubbed in three linear time passes. First, we remove all gates $g$ in leaf bags~$b$ when $b$ is not a
  $\phi$-occurrence of~$g$; this removal does not affect the fact that $T$ is a
  tree decomposition, because each wire is still covered at the $\phi$-image of
  its target gate, and the occurrences of each gate still form a connected
  subtree as we are only removing occurrences at leaf bags. Second, we do a
  bottom-up pass to compute, for each occurrence of a gate $g$ at a bag~$b$, a
  value $n_{g,b}$ indicating how
  many different children of~$b$ have a $\phi$-occurrence of~$g$ as a descendent
  (this number is between~$0$ and~$2$): if this number is~$0$ at an internal
  bag, then we remove~$g$. The removal of~$g$ in this case does not violate the
  fact that~$T$ is a tree decomposition: indeed, it does not affect the fact
  that the wires are covered, because this is still the case at the
  $\phi$-images, which are not internal bags; and gate occurrences still form
  connected subtrees, because whenever we remove a gate $g$ from a bag~$b$ then there is
  only one neighbor of~$b$ at most (namely, its parent) where $g$ also occurs.
  Third, we do a top-down pass where,
  whenever we have $n_{g,b} = 1$ for a gate~$g$ at a bag~$b$ when we encounter
  $b$ top-down (i.e.,
  $n_{g,b} = 1$ and $n_{g,b'} = 1$ for all ancestors of~$b'$ where $g$ occurs),
  then we remove $g$ from~$b$.
  The result is still a tree decomposition, because again wires are still covered at
  the $\phi$-images, and whenever we remove a gate~$g$ from a bag~$b$ then there is only one
  neighbor of~$b$ at most (namely, one of its children) where~$g$ also occurs.

  We now let $T'$ be the result of forgetting the orientation of~$T$. It is
  clear that the first point of the third condition is satisfied, and the second
  point is satisfied for any occurrence of a gate~$g$ at an internal bag~$b$: 
  indeed, $b$ either had two different children having a
  $\phi$-occurrence of~$g$ as a descendant, or $b$ had one such descendant and had an
  ancestor $b'$ whose parent $b''$ had a $\phi$-occurrence as a descendent of
  its other child. Hence, $T'$ is scrubbed.

  We last ensure that $T'$ satisfies the first condition. Note that, as each bag
  of~$T$ had at most two children, the undirected degree of bags of~$T$ is at
  most~$3$, so we can simply enforce the condition by adding bags with empty
  domain. This concludes the proof.
\end{proof}

\subsubsection{Rewriting monotone cycluits}

Armed with normal-form tree decompositions and the lemmas required to obtain
them, we are now ready to show our rewriting result for \emph{monotone}
cycluits:

\begin{theorem}
  \label{thm:rmcyclemono}
  There is an $\alpha\in\mathbb{N^*}$ s.t.,
  for any monotone cycluit $C$ of treewidth $k$, we can compute in time
  $O(2^{k^\alpha} \cdot \card{C})$ a circuit $C'$ which is equivalent to~$C$ and has
  treewidth $\leq 2^{k^\alpha}$.
\end{theorem}

To prove the theorem, we will rely on a notion of \emph{partial assignment} of a
monotone cycluit, which we can use to enforce that some gates are evaluated
to~$1$, even when they are not input gates:

\begin{definition}
  \label{def:partialass}
  For any monotone cycluit $C$ and disjoint subsets $S^+ \sqcup S^-$ of gates that contain no input gate,
  letting $S \defeq (S^+, S^-)$,
  we define the \emph{partial assignment}
  $\rho_S(C)$ to be the monotone cycluit obtained from~$C$ by replacing each
  gate of~$S^+$ by a $1$-gate, and each gate of~$S^-$ by a $0$-gate.
\end{definition}

We will also rely on a variant of Algorithm~\ref{alg:semantics-monotone} to
evaluate parts of a cycluit in a tree decomposition based on the child parts in
the decomposition. It will work using the following notion:

\begin{definition}
  The \emph{constant gates} of a cycluit $C = (G, W, g_0, \mu)$ are the gates of $G \backslash C_\inp$ that have no input gates,
  i.e., $C_\cst \defeq \{g \in G \mid \nexists g' \in G, g' \wire g \in W\}
  \backslash C_\inp$: in other words, these are precisely the $0$- and
  $1$-gates.
  The \emph{internal gates} of~$G$ are the gates which are
  neither constant gates nor input gates.

  An \emph{evaluation partition} of a cycluit $C = (G, W, g_0, \mu)$ 
  is a partition $G_1 \sqcup G_2$ of the internal gates of~$C$.
  The \emph{frontier} of $G_1 \sqcup
  G_2$ is $F(G_1 \sqcup G_2) \defeq F_1(G_1 \sqcup G_2) \sqcup F_2(G_1 \sqcup
  G_2)$, where $F_j(G_1 \sqcup G_2) \defeq \{g \in
  G_{3-j} \mid \exists g' \in
  G_{j}, g \wire g' \in W\}$ for $j \in \{1, 2\}$.
\end{definition}

This allows us to define Algorithm~\ref{alg:semantics-monotone-part}, a variant
of Algorithm~\ref{alg:semantics-monotone} which evaluates a cycluit by
evaluating separately the parts of its evaluation partition, memorizing only the
state of the frontier across evaluations. We claim the following:

\begin{algorithm}[t]
  \DontPrintSemicolon
  \KwIn{Monotone cycluit $C=(G,W,g_0,\mu)$,
  Boolean valuation $\nu: C_{\inp} \to \{0,1\}$,
  evaluation partition $G_1 \sqcup G_2$}
  \KwOut{$\{g \in C \mid \nu(g)=1\}$}
  $\CST \defeq \{ g \in C_\inp \mid \nu(g)=1 \} \cup \{ g \in C_\cst \mid
  \mu(g) = \land\}$\;
  $T_0 \defeq \CST$\;
  $n \defeq \card{F(G_1 \sqcup G_2)} + 2$\;
  \For{$i$ from $1$ to $n$}{
    \For{$j$ from $1$ to $2$}{
      $U_{i,j} \defeq T_{i-1} \cap F_j(G_1 \sqcup G_2)$\;
      $V_{i,j} \defeq \textup{SubEval}(C, \CST \cup U_{i,j}, G_j)$\;
    }
    $T_i \defeq V_{i,1} \cup V_{i,2} \cup \CST$
  }
  \Return $T_n$\;
  \SetKwProg{myfunc}{Function}{}{}
  \myfunc{$\textup{SubEval}(C, Z, G')$}{
    \KwIn {Monotone cycluit $C=(G,W,g_0,\mu)$, subset $Z \subseteq G$ of true
    gates, subset $G' \subseteq G$ of iteration gates}
    \KwOut{Subset of gates that became true (observe the
    similarity to Algorithm~\ref{alg:semantics-monotone})}
  $Z_0 \defeq Z$\;
  $l \defeq 0$\;
  \Do{$Z_l \neq Z_{l-1}$}{
    $l \defeq l+1$\;
    $Z_{l} \defeq Z_{l-1}\cup\Big\{g \in G' \mid (\mu(g) =
    {\lor}), \exists g^\prime \in Z_{l-1}, g^\prime \rightarrow
  g\in W\Big\}
$\par~\mbox{${}\qquad\qquad~\cup\Big\{g \in G' \mid (\mu(g) =
    {\land}), \{g^\prime \mid g^\prime \rightarrow g\in W\} \subseteq
Z_l\Big\}$}
  }
  \Return $Z_l$\;
  }
  \caption{Partition-based evaluation of monotone cycluits}
  \label{alg:semantics-monotone-part}
\end{algorithm}

\begin{lemma}
  \label{lem:algocorrect}
  For any cycluit $C$, valuation $\nu$, and evaluation partition $G_1 \sqcup
  G_2$ of~$C$, the result of Algorithm~\ref{alg:semantics-monotone-part} is
  indeed the set $g$ of gates such that $g$ evaluates to~$1$ in~$C$ under~$\nu$.
  In other words, the output is the same as that of
  Algorithm~\ref{alg:semantics-monotone}.
\end{lemma}

\begin{proof}
  Fix the cycluit $C$, valuation $\nu$, and evaluation partition $G_1 \sqcup
  G_2$, let $F \defeq F(G_1 \sqcup G_2)$, and write $n \defeq \card{F} + 2$.

  First, observe that the output on input gates and on constant gates is
  correct, as all such gates that evaluate to~$1$ are added to~$\CST$ and are
  returned, and the ones that do not are never otherwise added to any $T_i$ 
  (remember that the $G_j$ do not
  contain any such gates, so we only add internal gates to the $Z_l$). Hence, it
  suffices to restrict our attention to the internal gates.

  The easy direction is the forward direction: let us show that, for any internal gate $g$ in~$T_n$, this
  gate is also in the output of Algorithm~\ref{alg:semantics-monotone}. We do so
  by induction on the step of the execution at which the gate was added to one
  of the sets~$Z_l$. Initially, there are no internal gates in these sets.
  For the induction, observe that the
  gates added to such sets are either $\lor$-gates which have one input already
  previously
  in such a set (and so, by induction hypothesis, they are also in the output of
  Algorithm~\ref{alg:semantics-monotone}), or $\land$-gates where all inputs are
  already in such a set (and we can use the induction hypothesis again), so
  that indeed, when that same gate is considered in one iteration of
  Algorithm~\ref{alg:semantics-monotone}, it will also be added to the output of
  that algorithm.

  We now prove the backward direction. To do this, first observe that, for any
  $j \in \{1, 2\}$, for any $1 \leq i < n$, we have $U_{i,j} \subseteq
  U_{i+1,j}$. Indeed, we
  have $U_{i,j} \subseteq V_{i,j} \subseteq T_i$, and all gates of
  $U_{i,j}$ are gates of $F_j(G_1 \sqcup
  G_2)$, so $U_{i+1,j}$ contains all gates of~$U_{i,j}$. Now, as $U_{i,j}$
  has size at most $\card{F(G_1 \sqcup G_2)}$, this means that there is $i' \leq
  \card{F(G_1 \sqcup G_2)} + 1$ such that $U_{i',j} = U_{i'+1,j}$. As the
  same argument applies to the other value of $j$, we deduce the existence of
  $1 \leq i' \leq \card{F(G_1 \sqcup G_2)} + 1$ such that this holds for all $j \in
  \{1, 2\}$. This clearly implies, then, that $T_{i'} = T_{i'+1}$, so that, as
  we chose $n$ to be sufficiently large, 
  we know that $T_n = T_{n-1}$.

  We use this claim to conclude the proof of the backward direction.
  Assume by way of contradiction that there is a gate~$g$ which is added to the
  output of Algorithm~\ref{alg:semantics-monotone} and which is not part of the
  output of
  Algorithm~\ref{alg:semantics-monotone-part}, and choose $g$ to be the first
  gate that has this property (formally, as several such gates may be added in
  the same iteration of Algorithm~\ref{alg:semantics-monotone}, choose any gate
  among the first ones). By our initial remark about the correctness on constant
  gates and input gates, $g$ must
  be an internal gate. Let $\Gamma$ be the subset of the inputs of~$g$ which
  were evaluated to~$1$ at an iteration of
  Algorithm~\ref{alg:semantics-monotone} which is strictly before the iteration
  where $g$ was evaluated to~$1$.
  By definition of that algorithm,
  $\Gamma$ is in the output of Algorithm~\ref{alg:semantics-monotone}, and
  hence, by minimality of~$g$, it is in the output of
  Algorithm~\ref{alg:semantics-monotone-part}. Hence, $\Gamma \subseteq T_n$,
  hence $\Gamma \subseteq T_{n-1}$ by the claim that we showed.
  Let us use this to study what happened during the
  last iteration of Algorithm~\ref{alg:semantics-monotone-part} and reach a
  contradiction.

  As $g$ is an internal gate, let $j \in \{1, 2\}$ such that $g \in G_j$, and
  let us show that $g \in V_{n,j}$, by considering the invocation
  $\textup{SubEval}(C, \CST \cup U_{n,j}, G_j)$. 
  It suffices to show that there
  is some $l \in \NN$ such that $\Gamma \subseteq Z_l$. Indeed, if this is the case,
  then, taking a
  minimal such $l$, we can consider the $(l+1)$-th iteration of the loop of
  $\textup{SubEval}$: this iteration takes place, because $Z_l$ has just changed
  by addition of some gate in~$\Gamma$ (or, if
  $l = 0$, it takes place by definition of a \textbf{do... while} loop), and in
  this iteration we know that we have $g \in Z_{l+1}$, because $g$ was set to~$1$ in
  Algorithm~\ref{alg:semantics-monotone} on the basis of~$\Gamma$, and the
  iterations in $\textup{SubEval}$ are defined in the same way on~$G_j$. So let us show
  that $\Gamma \subseteq Z_l$ for some $l \in \NN$. As the $Z_l$ are monotone,
  we can show this gate-per-gate: so let $g' \in \Gamma$ and show that it is
  in~$Z_l$ for some~$l \in \NN$.

  We do a case disjunction. First, if $g' \in \CST$, then 
  we have $g' \in Z_0$, so there is nothing to show.
  Second, assume that $g' \in G_j$. We use the fact that 
  $g' \subseteq T_{n-1}$ and $g' \in G_j$ to deduce that, by definition
  of~$T_{n-1}$,
  we must have had $g' \in \textup{SubEval}(C, \CST \cup U_{n-1,j},
  G_j)$. Now, as $U_{n-1,j} \subseteq U_{n,j}$, the monotonicity of
  $\textup{SubEval}$ ensures that we also have $g' \in \textup{SubEval}(C,
  \CST \cup U_{n,j}, G_j)$. Hence, there is indeed $l\in \NN$ during this
  evaluation
  such that $g' \in Z_l$.
  Third, assume that $g' \in G_{3-j}$.
  As $g'$ is an input gate of~$g$, and $g \in G_j$, this implies that $g' \in 
  F_j(G_1 \sqcup G_2)$. Hence, as
  $\Gamma \subseteq T_{n-1}$, we have $g' \in
  U_{n,j}$, hence, $g' \in Z_0$.
  
  This shows that, indeed, in the last iteration ($i = n$) of
  Algorithm~\ref{alg:semantics-monotone-part}, for $j$ defined such that $g \in
  G_j$, there is some iteration of the corresponding invocation of
  $\textup{SubEval}$ where all gates of $\Gamma$ are added to $Z_i$, so that we must add $g$ to
  $Z_{i+1}$. Hence, $g$ is in the output of
  Algorithm~\ref{alg:semantics-monotone-part}, which contradicts our initial
  assumption and concludes the proof.
\end{proof}

We are now ready to prove Theorem~\ref{thm:rmcyclemono}:

\begin{proof}[Proof of Theorem~\ref{thm:rmcyclemono}]
  We will not define $\alpha$ explicitly, but it will be apparent from the proof that a suitable
  $\alpha$ can be defined.
  Note that, up to adding a fresh gate to serve as the new distinguished output
  gate, we can clearly ensure that the
  distinguished output gate~$g_0$ has exactly one input, and that it is not an input gate to any other gate of~$C$,
  and we can modify~$T$ accordingly.
  We start by computing in time $O(f(k) \cdot \card{C})$, for some suitable
  exponential $f$, a tree
  decomposition $T$ of~$C$ \cite{bodlaender1996linear}.
  We use Lemmas~\ref{lem:getarity2}, \ref{lem:getregrouped}
  and~\ref{lem:getnormal} to ensure that, up to increasing the degree of
  the polynomial in the exponent of~$f$,
  we can assume that the monotone cycluit $C$ is arity-two and that
  the tree decomposition~$T$ is in normal form. We choose to root the
  normal-form tree decomposition~$T$ at the bag $\phi(g_0)$ where $g_0$ is the
  distinguished output gate. The
  resulting~$T$ is such that $\phi$ maps each gate to a leaf bag, except the
  distinguished output gate $g_0$ which is mapped to the root; however, from our initial
  modification of~$C$ and the fact that $T$ is scrubbed, we know that $g_0$ only
  occurs at the root (specifically, the root bag contains $g_0$ and its one
  input gate). Hence, in the sequel, we will never apply~$\phi$ to~$g_0$, and so
  we see $\phi$ as a mapping to the leaf bags of~$T$. To simplify the induction, we
  ensure that the root bag also has two children, by adding to it a child with
  empty domain; this does not break any of the conditions on~$T$. 

  We must now create the circuit $C'$ and its tree decomposition $T'$. In fact,
  we will create a \emph{pre-tree-decomposition} $T'$ of~$C'$, namely, a tree
  with same skeleton as~$T$, where each bag contains at most
  $2^{k^\alpha}$ gates
  of~$C'$, and which satisfies the following conditions: first,
  all wires in~$C'$ are
  either between nodes of a same bag, or from nodes of a bag to nodes of a
  parent bag (i.e., all wires go upward in $T'$); second, each gate only occurs in
  the bag where it is created, except the gates of $C'_{\inp} = C_{\inp}$,
  whose occurrences in~$T'$ will correspond to a subset of the bags of~$T$ where
  they occur. The
  pre-tree-decomposition~$T'$ can clearly be rewritten to a tree decomposition
  in linear time and with a blowup of a factor~$3$ in the width: first, by
  adding the input gates to the bags corresponding to all their actual occurrences
  in~$T$, and second by adding each
  gate of~$T'$ to its parent bag, which ensures that all wires are indeed covered.
  Hence, we will only describe how to construct a pre-tree-decomposition $T'$
  of~$C'$.
  
  We will create $C'$ and $T'$ by a bottom-up
  traversal of~$T$. At each bag $b$ of~$T$, we will denote by $T_b$ the subtree
  of~$T$ rooted at~$b$, we will write $\dom(T_b)$ to refer to all gates
  occurring in~$T_b$, and will write $C_b$ to mean the sub-cycluit of~$C$
  obtained by restricting to the gates of~$\dom(T_b)$ and a subset of the wires
  between them, namely, the wires whose target $g$ 
  is such that $\phi(g) \in T_b$. In other words, when $g \in \dom(T_b)$ but
  $\phi(g)$ is not in~$T_b$, then $g$ occurs in~$C_b$ but it has no inputs, even though it may
  be the case that both inputs also occur in~$\dom(T_b)$. Further, the
  inputs of~$C_b$ are then precisely $C_\inp \cap \dom(T_b)$.
  Throughout the proof, for each bag~$b$, we will partition
  the gates of~$\dom(b)$ in two sets (clearly computable in linear time):
  \begin{itemize}
  \item The \emph{upward gates}, written $\domu(b)$, which are the gates
  $g$ such that $\phi(g)$ is a descendant of~$b$, as well as all the input gates
      $g$ in $\dom(b)$ and all constant gates $g$ in~$\dom(b)$.
  \item The \emph{downward gates}, written $\domd(b)$, which are all
  other gates.
  \end{itemize}
 
  Intuitively, when processing the bag $b$, the \emph{upward gates} are those which can
  be evaluated from other gates of~$T_b$ (and from the valuation of the input
  gates) based on what happens in~$T_b$. (However, this evaluation may depend on
  downward gates of~$\dom(b)$, so that we may be unable to perform it fully.) By
  contrast, the \emph{downward gates} are those for which we have no complete information yet,
  and whose evaluation will take place later

  To rewrite each $C_b$ to a circuit when processing $T$, what is relevant to us is
  the behavior of~$C_b$ for each valuation of its inputs, in terms of how the
  upward gates evaluate \emph{depending on the value of the downward gates}.
  Intuitively, to rewrite $C_b$, we need to describe the ``function'' that gives
  us the value of the upward
  gates (which is propagated to parent bags that we will rewrite later)
  depending on each possible valuation of the downward gates (for which we do not know the value).
  Specifically, for any subset $S^+ \subseteq \domd(b)$,
  letting 
  $S \defeq (S^+, \domd(b) \backslash S^+)$, for any gate $\gamma
  \in \domu(b)$, and for any valuation $\nu$ of the inputs of~$C_b$, we wish
  to know whether $\gamma$ evaluates to~$1$ in~$\nu(\rho_S(C_b))$. Recalling the
  definition of partial assignments (Definition~\ref{def:partialass}), note that
  this is well-defined because $S^+$ contains no input gates. For brevity we write 
  $\Pi_b \defeq 2^{\domd(b)} \times \domu(b)$

  \medskip

  We are now ready to describe the inductive rewriting process. For each bag $b$
  of~$T$, we will create a bag $b'$ of~$T'$ containing gates of~$C'$. Remember
  that the circuit $C'$ must have the same inputs as $C$; in particular we will
  ensure that, for any $b \in T$, the cycluit $C_b$ and the circuit $C'_{b'}$
  have the same inputs; further, by definition of the pre-tree-decomposition
  $T'$, all wires will go upward in $T'$, meaning that, for any $b' \in T'$, the
  circuit $C'_{b'}$ can always be fully evaluated from a valuation of its
  inputs. From this observation, we will inductively guarantee that 
  each bag $b'$ in $T'$, corresponding to bag $b$ in~$T$, will contain one gate $g_b^{S^+, \gamma}$ for each $(S^+, \gamma)
  \in \Pi_b$, with the following semantics: for any valuation $\nu$
  of the inputs of~$C_b$, the gate $g^{S^+, \gamma}$ evaluates to~$1$ in
  $\nu(C'_{b'})$ iff $\gamma$ evaluates to~$1$ in $\nu(\rho_S(C_b))$ where $S \defeq
  (S^+, \domd(b) \backslash S^+)$.

  We first describe the base case of the construction. Remember that, as $T$ is scrubbed, each
  leaf bag of~$T$ is either in the image of~$\phi$ or has empty domain.
  For a leaf bag $b$ of~$T$ with empty domain,
  there is nothing to do. Otherwise, let $g$ be the preimage of~$b$ by~$\phi$,
  i.e., the unique $g$ (because $\phi$ is injective, and it is surjective on leaf
  bags with non-empty domain) such that $b = \phi(g)$. We then
  have $\domu(b) = \{g\} \cup (C_\inp \cap \dom(b))$ and $\domd(b) = \ins(g) \setminus C_\inp$; observe that, as $C$ is
  arity-two, it has no self-loops, so that $\ins(g)$ and $\{g\}$ are disjoint.
  If $g$ is an input gate, then $\ins(g) = \emptyset$ and the
  only gate to create in~$b'$ is $g_b^{\emptyset,g}$ which can simply be taken to
  be $g$ itself and satisfy the required conditions. If $g$ is not an input
  gate, the gates to define are $g_b^{S^+,g}$ for all $S^+ \subseteq \ins(g)
  \backslash C_\inp$ (which has
  size at most 2), and we can simply realize this
  following the truth table of the
  operation $\mu(g)$ of~$g$, i.e., by setting $g_b^{S^+,g}$ to be a sub-circuit with
  input $\dom(b) \cap C_\inp$
  whose value is the value of $g$ in $\rho_S(C_b)$ with $S = (S^+,\ins(g) \setminus
  (C_\inp \cup S^+))$.

  We now describe the inductive case of the construction. We consider an internal bag $b$ which has
  two children $b_1$ and~$b_2$ in~$T$; we write $b'$, $b_1'$, $b_2'$ the
  corresponding bags in $T'$, with $b_1'$ and $b_2'$ being already inductively
  defined.
  In particular, we know that $b_1'$ and $b_2'$ contain gates (which for brevity we
  will write $g_1^{S^+,\gamma}$ and $g_2^{S^+,\gamma}$, rather than
  $g_{b'_1}^{S^+,\gamma}$ and $g_{b'_2}^{S^+,\gamma}$) with the prescribed
  properties; and we will omit for brevity the $b$ subscript when defining gates
  of~$b'$ in~$C'$ standing for the current bag $b$ of~$T$. For $i$ from $0$
  to~$(k+1)+2$ ,
  we will now define gates $g^{S^+,\gamma,i}$ in~$b'$ as follows, for all
  $(S^+,\gamma) \in \Pi_b$:
  \begin{itemize}
    \item Case of constant gates: if $\gamma$ is a $1$-gate
      (resp., a $0$-gate),
      then the gate $g^{S^+,\gamma,i}$ is a $1$-gate (resp., a
      $0$-gate) for all $0
      \leq i \leq k+3$.
    \item Case of input gates: if $\gamma$ is an input gate, then the
      gate $g^{S^+,\gamma,i}$ is the same input gate for all $0 \leq i \leq
      k+3$.
  \item Base case:
	  the gate $g^{S^+,\gamma,0}$ is simply a $0$-gate.
  \item Induction:
    letting $j$ be such that $\phi(\gamma)$ is a descendant of~$b_j$ (so that
      $\gamma \in
    \domu(b_j)$ and $\gamma \in \domd(b_{3-j})$ if $\gamma \in \dom(b_{3-j})$),
      the gate $g^{S^+,\gamma,i}$ for $i \geq 1$ is
      an OR-gate, for all $\Gamma' \subseteq \domd(b_j) \cap \domu(b_{3-j})$
      (which is necessarily $\subseteq \dom(b)$),
      of the AND of $g^{S^+ , \gamma', i-1}$ for all $\gamma' \in \Gamma'$
      and of $g_j^{(S^+ \cap \domd(b_j)) \cup \Gamma',\gamma}$.
  \end{itemize}
  We simply let $g^{S^+,\gamma}$ to be $g^{S^+,\gamma,k+3}$. It is clear that we
  create only exponentially many gates at each bag of~$T'$, so that we can obey
  the prescribed treewidth bound (in addition to the transformations of~$C$ that
  we have performed at the beginning, namely, transformation to arity-2, and
  up to the use of a normal-form tree decomposition which implies that the bag
  size may be multiplied by~$3$).

  \medskip

  We now show the correctness of this construction at~$b$, i.e., we check
  inductively whether the gates that we create satisfy their prescribed
  semantics.
  Fixing $(S^+, \gamma) \in
  \Pi_b$ and a valuation $\nu$ of the inputs of~$C_b$, we consider the monotone
  cycluit $\rho_S(C_b)$ and its evaluation under~$\nu$, where $S \defeq (S^+, \domd(b) \backslash S^+)$
  (note that this does not include any input gates because $\domd(b)$ does not).
  What we must show is that $g^{S^+,\gamma}$
  evaluates to~$1$ under $\nu$ in~$C'$ iff
  $\gamma$ evaluates to~$1$ in $\nu(\rho_S(C_b))$.
  To show this claim, we must connect our construction of the~$g^{S^+,\gamma,i}$
  to the evaluation of the cycluit $\nu(\rho_S(C_b))$; we will do so using the
  notion of an evaluation partition, by defining one from the structure of~$T_b$.
  
  Specifically, consider the evaluation partition $G_1 \sqcup G_2$
  of the internal gates of~$\rho_S(C_b)$ defined by putting in~$G_1$ the internal gates
  whose image by $\phi$ is a descendant of~$b_1$, and in $G_2$ those whose
  image is a descendant of~$b_2$. We show the following claims:
  
  \begin{itemize}
    \item \emph{This partition indeed covers the
      internal gates of~$\rho_S(C_b)$.} Indeed, all internal gates of~$C_b$ are
  either gates of $\dom(T_1) \backslash \dom(b)$ (in which case their
  $\phi$-image must be in~$T_1$), or gates of $\dom(T_2) \backslash \dom(b)$
  (same reasoning), or gates of $\dom(b)$ in which case they must be in
  $\domu(b)$ (as the gates of $\domd(b)$ are constant in $\rho_S(C_b)$) and then
  their $\phi$-image is a descendant of~$b$ by definition. Now, the $\phi$-image cannot be $b$
  itself as $\phi$ maps only to leaves, so it must be either a descendant
      of~$b_1$ or of~$b_2$.
\item \emph{For each $j \in \{1, 2\}$, we have $F_j(G_1 \sqcup G_2) = \domu(b_{3-j})
  \cap \domd(b_j)$.} Indeed, for the forward inclusion, any gate $g$ of $F_j(G_1
    \sqcup G_2)$ is a gate of~$G_{3-j}$, that is, $\phi(g)$ is a descendant
    of~$b_{3-j}$. Further, there is a
    wire $(g, g')$ with $g' \in G_j$, that is, $\phi(g')$ is a descendant
      of~$b_j$, and $g$ occurs in $\phi(g')$. Now, as $g \in \dom(T_j) \cap
      \dom(T_{3-j})$, as $T$ is a tree decomposition, we have $g \in \dom(b_j)$
      and $g \in \dom(b_{3-j})$, and more specifically we have $g \in \domu(b_{3-j})$ and $g \in
      \domd(b_j)$ by what precedes.  Conversely, any gate $g$ of
      $\domu(b_{3-j})$ is such that $\phi(g)$ is a descendant of~$b_{3-j}$, so
      that $g \in G_{3-j}$. Further, as $g \in \dom(b_j)$, as the tree
      decomposition $T$ is scrubbed, there is a descendant bag $b$ of~$b_j$
      where $g$ appears in a $\phi$-image, and as $\phi(g)$ is not a descendant
      of~$b_j$ it must be the case that $g$ appears in the $\phi$-image of
      some~$g'$ (so that $g' \in G_j$) such that $g \wire g'$ is a wire. Hence, $g
      \in G_{3-j}$ and $g$ is an input of a gate in~$G_j$, so indeed $g \in
      F_j(G_1 \sqcup G_2)$.
    \item \emph{We have $\card{F(G_1 \sqcup G_2)} \leq k+1$.} This follows from the
      previous point. Indeed, we have $\card{F(G_1 \sqcup G_2)} = \card{F_1(G_1 \sqcup
      G_2)} + \card{F_2(G_1 \sqcup G_2)} \leq \card{\domu(b_1)} +
      \card{\domd(b_1)} \leq \card{\dom(b)} \leq k+1$ because $T$ is a tree decomposition
      of width~$\leq k$.
  \end{itemize}
  
  Now, consider the evaluation of
  $\rho_S(C_b)$ under $\nu$ using the evaluation partition $G_1 \sqcup G_2$,
  implemented with Algorithm~\ref{alg:semantics-monotone-part}: write
  $T_0(S), T_1(S), \ldots, T_{k+3}(S)$ the sequence of sets defined in the
  execution of the algorithm. (As this sequence may be shorter, if necessary we
  pad it until $T_{k+3}(S)$ by copying the last value.)
  Nested in our induction over~$T$, we now show by induction on $0 \leq i \leq k+1$
  that, in $b'$ as we have
  defined it,
  for any $S^+ \subseteq \domd(b)$
  and $\gamma \in \domu(b)$,
  for any valuation $\nu$ of the
  inputs of~$C_b$, the gate $g^{S^+,\gamma,i}$ evaluates to~$1$ under $\nu$ iff
  $\gamma \in T_i(S) \cap \domu(b)$. For the base case of
  $i=0$, indeed $g^{S^+,\gamma,0}$ evaluates to true iff
  $\gamma \in
  T_0(S) \cap \domu(b)$ by definition of~$T_0(S)$. For the induction step, observe
  that, by definition, $g^{S^+,\gamma,i}$ evaluates to true 
  iff, letting $j$ be such that $\phi(\gamma)$ is a descendant of~$b_j$,
  there is some $\Gamma' \subseteq \domd(b_j) \cap \domu(b_{3-j})$
  (equivalently, by the second bullet point above, $\Gamma' \subseteq F_j(G_1 \sqcup G_2)$) such
  that $g^{S^+,\gamma',i-1}$
  evaluates to~$1$ for all $\gamma' \in \Gamma'$ and
  $g_j^{(S^+ \cap \domd(b_j)) \cup \Gamma', \gamma}$ evaluates to~$1$.

  For the forward direction of this induction claim (for the nested induction),
  assume the existence of such a~$\Gamma'$, and show that $\gamma \in
  T_{i}(S)$ (this suffices as we have $\gamma \in \domu(b)$ by definition).
  By the induction
  hypothesis for the inner induction, as $g^{S^+,\gamma',i-1}$
  evaluates to~$1$ for all $\gamma' \in \Gamma'$, we know that $\Gamma'
  \subseteq T_{i-1}(S)$. By the induction hypothesis for the outer induction, as 
  $g_j^{(S^+ \cap \domd(b_j)) \cup \Gamma', \gamma}$ evaluates to~$1$, we know that 
  $\gamma$ evaluates to~$1$ in
  $\nu_j(\rho_{S_j}(C_j))$, where $S_j \defeq ((S^+ \cap \domd(b_j)) \cup \Gamma', \domd(b_j)
  \backslash (S^+ \cup \Gamma'))$ and $\nu_j$ is the restriction of~$\nu$ to the inputs
  of~$C_j$. Now,
  using Lemma~\ref{lem:algocorrect},
  consider the execution of
  Algorithm~\ref{alg:semantics-monotone-part} (on $\rho_{S}(C_b)$ under $\nu$), and consider the invocation of
  $\textup{SubEval}$ for this $i$ and $j$. We know that
  $\CST$ contains all input gates of~$\rho_{S_j}(C_j)$ set to~$1$ and all
  gates of~$S^+$ (these are constant gates of~$\rho_{S}(C_b)$), so they are
  in~$Z_0$; and as
  $\Gamma' \subseteq F_j(G_1 \sqcup G_2)$ and $\Gamma' \subseteq T_{i-1}(S)$, we
  know that $\Gamma' \subseteq U_{i,j}$, so $\Gamma' \subseteq Z_0$. Now,
  as we know that $\gamma$ evaluates to~$1$ under $\nu_j$ in $\rho_{S_j}(C_j)$,
  consider 
  the evaluation of Algorithm~\ref{alg:semantics-monotone} to witness this,
  and observe that, by what precedes,
  the invocation of $\textup{SubEval}$ that we are considering is such that
  $Z_0$ contains all gates which are initially true in
  $\nu_j(\rho_{S_j}(C_j))$, namely, the input gates that evaluate to~$1$, the
  constant $1$-gates, and the gates of the first component of~$S_j$,
  specifically,
  those of~$\Gamma'$ and a subset of those of~$S^+$. Hence, each time
  Algorithm~\ref{alg:semantics-monotone} sets a gate to~$1$, $\textup{SubEval}$
  also does. Hence, $\gamma$ is also returned by $\textup{SubEval}$, so indeed
  $\gamma \in T_i(S)$, which concludes the forward direction.

  For the backward direction of the nested induction claim, assume that $\gamma
  \in T_i(S) \cap \domu(b)$, and show the existence of a suitable $\Gamma'$. 
  Let $j \in \{1, 2\}$ be such that $\gamma \in \domu(b_j)$.
  We
  choose $\Gamma'$ to be $T_{i-1}(S) \cap (\domd(b_j) \cap \domu(b_{3-j}))$, or,
  in other words, $U_{i,j}$ in Algorithm~\ref{alg:semantics-monotone-part}.
  Now, for any $\gamma' \in \Gamma'$,
  observe that $\gamma' \in \domu(b)$ (as $\gamma' \in \domu(b_{3-j})$, we know
  that $\phi(\gamma')$ is a descendant of $b_{3-j}$, hence of~$b$), so that
  $\gamma' \in T_{i-1}(S) \cap \domu(b)$, and by induction hypothesis for the
  inner induction we know that $g^{S^+, \gamma', i-1}$ evaluates to~$1$. Now we
  will use the induction hypothesis for the outer induction to argue that
  $g_j^{(S^+ \cap \domd(b_j)) \cup \Gamma', \gamma}$ evaluates to~$1$, by
  showing that $\gamma$ evaluates to true under~$\nu_j$ in $\rho_{S_j}(C_j)$,
  with the same notations for~$\nu_j$ and $S_j$ as for the forward direction above.
  To do this, as in the
  forward direction, we consider the evaluation of $\textup{SubEval}$ at
  iteration~$i$ and for the current $j$. In this evaluation, the gates that
  are initially true are the constant $1$-gates (in particular those of~$S^+$),
  the input gates that are true under~$\nu$, and the gates of~$T_{i-1}(S) \cap
  F_j(G_1 \sqcup G_2)$, that is, of~$\Gamma'$. In $\rho_{S_j}(C_j)$, all gates
  of~$\Gamma'$ are also true, and the other gates are also true whenever they
  appear in $C_j$: in other words, all gates that are in~$Z_0$ at the beginning
  of this evaluation of $\textup{SubEval}$ are also initially true in
  $\nu_j(\rho_{S_j}(C_j))$ except if they are gates that do not occur
  in~$\rho_{S_j}(C_j)$ (and have no inputs or outputs in $\rho_S(C_b)$).
  This property ensures that all gates that evaluate to true in
  $\textup{SubEval}$, which are gates of $G_j$, can also evaluate to true in
  Algorithm~\ref{alg:semantics-monotone} thanks to
  Lemma~\ref{lem:algocorrect}.
  Hence, indeed $\gamma$ evaluates to true under~$\nu_j$ in $\rho_{S_j}(C_j)$, so
  we conclude that indeed $g_j^{(S^+ \cap \domd(b_j)) \cup \Gamma', \gamma}$
  evaluates to~$1$. This finishes the proof of the inner induction, and hence
  concludes the correctness proof of the induction step of the outer induction.

  \medskip

  To conclude the proof of the Theorem, it suffices to observe that the root bag
  $b'_\r$ of the rewriting $C'$ contains a gate $g \defeq g_\r^{\emptyset, \{g_0\}}$
  where $g_0$ is the output gate of~$C$, whose semantics (by the outer
  induction claim) is guaranteed to be the following: for any valuation $\nu$ of
  $C$ (hence, of $C'$), the gate $g$ evaluates to true under $\nu$ in~$C'$ iff
  $g_0$ evaluates to true in $\nu(C)$. Hence, we set $g$ as the output gate
  of~$C'$, and we have shown that $C'$ is indeed equivalent to~$C$. This concludes
  the proof.
\end{proof}

\subsubsection{Isotropic rewritings of monotone cycluits}

To prove Theorem~\ref{thm:rmcycle} for non-monotone cycluits, we will of course
perform an induction on the strata. The challenge when doing so is that, when
rewriting a stratum, we may need to access the value of \emph{all} gates of the lower
strata, not just a single output gate. So, in the induction step on strata, it
will not be sufficient to rewrite a stratum to a circuit with a single output
gate: we need to rewrite in a way where all gates of the circuit can be reused
by lower strata. We define this formally:

\begin{definition}
  \label{def:emulates}
  A circuit (or cycluit) $C'$ \emph{emulates} a cycluit $C$ iff $C$ and $C'$
  have the same input gates, all gates of $C$ are gates of~$C'$, and for any
  valuation $\nu$ of the inputs, for any gate~$g$ of $C$, $g$ evaluates to true
  in $C$ under $\nu$ iff it does in $C'$ under $\nu$.
\end{definition}

Similarly to how we modified the circuit at the beginning of the proof of
Theorem~\ref{thm:rmcyclemono}, it will be convenient to rewrite the circuit in
linear time to ensure that each gate of interest has exactly one input and is
not itself the input to another gate. We call such gates the \emph{potential
outputs}:

\begin{definition}
  \label{def:potentialout}
  A \emph{potential output} of a circuit (or cycluit) $C$ is a gate $g$ of~$C$
  which has exactly one input gate and is not itself the input to any gate.
\end{definition}

We can very simply rewrite our input cycluit so that all gates of interest are
potential outputs:

\begin{lemma}
  \label{lem:introgates}
  For any cycluit $C$ of treewidth $k$, we can compute in linear time a cycluit
  $C'$ of treewidth $\leq 2k$ that emulates $C$, such that, for each gate $g$ of
  $C \setminus C_\inp$, the gate $g$ in $C'$ is a potential output.
\end{lemma}

\begin{proof}
  Simply build $C'$ by renaming all gates of $C$ that are not input gates to a
  different name. Now, for each gate $g$ of $C$, add gate $g$ to $C'$ by setting
  it to be an AND-gate whose one input is the gate $g'$ which is the renaming of $g$ in~$C'$. This
  clearly ensures that $g$ evaluates to TRUE iff its renaming does, so that $C'$
  indeed emulates $C$. Further, the process is clearly in linear time and the
  treewidth bound is respected because we are simply duplicating each gate.
\end{proof}

This allows us to assume that we are working with a cycluit where all
``relevant'' gates are not the input gates to other gates.
Thanks to this, 
the process that we really need for an induction on strata is the following variant of
Theorem~\ref{thm:rmcyclemono}:

\begin{theorem}
  \label{thm:rmcyclemono2}
  There is an $\alpha\in\mathbb{N^*}$ s.t.,
  for any monotone cycluit $C$ of treewidth $k$, we can compute in time
  $O(2^{k^\alpha} \cdot \card{C})$ a circuit $C'$ that emulates~$C$ and has
  treewidth $\leq 2^{k^\alpha}$.
\end{theorem}

\begin{proof}.
First note that, thanks to Lemma~\ref{lem:introgates},
we can process the input monotone cycluit~$C$ in linear time, keeping the
parametrization in the treewidth, such that all gates of the original input
cycluit are now potential output gates. Hence, it now suffices to construct a
cycluit $C'$ that \emph{quasi-emulates} the cycluit $C$ that we are working
with, i.e., only emulates the potential output gates of~$C$: indeed,
$C'$ will then emulate all gates of the original cycluit.
As in Theorem~\ref{thm:rmcyclemono}, we will assume that $C$ is arity-two, and that
it has a normal-form tree decomposition $T$. 
Further, as in the proof of Theorem~\ref{thm:rmcyclemono}, we will
construct a pre-tree decomposition $T'$ of the result $C'$, i.e., all wires are either within one bag or between a gate of
a bag and a gate of a neighboring bag: $T'$ will have same skeleton as~$T$ when
forgetting about edge orientation.

We first notice that we can obtain the desired circuit $C'$ in a naive way,
which runs in quadratic time and does not satisfy the treewidth bound. To do so,
we can simply apply Theorem~\ref{thm:rmcyclemono}, for each potential output
$g$ of~$C$, on the tree decomposition $T^b$ obtained by rooting $T$ at $b
\defeq \phi(g)$. We could then build
the rewriting of~$C$ as a circuit $C'$ that quasi-emulates $C$, by taking the union of
the rewritings thus produced: these rewritings are disjoint except for the input
  gates. Further, we can construct a pre-tree
decomposition $T'$ of~$C'$ by taking all the pre-tree
decompositions produced in the applications of Theorem~\ref{thm:rmcyclemono},
  forgetting about their orientation (so they have the same undirected skeleton
  as $T^b$, namely, that of~$T$), taking their bag-per-bag union, and rerooting
  the result at an arbitrary bag to obtain a rooted tree decomposition.

Clearly the result $T'$ of this process
is still a pre-tree decomposition. First, all facts are still contained
within bags or hold between one bag and its neighboring bag (as it did in the
original pre-tree decompositions). Second, the occurrences of gates still form a
connected subtree. Indeed, for input gates, this connected subtree is the same as
  in~$T$, as is readily observed from the
construction of Theorem~\ref{thm:rmcyclemono}. For the other gates, the
  connected subtree is the same as
the original pre-tree decomposition from which they come.
Further,
   Theorem~\ref{thm:rmcyclemono} ensures that
  the resulting circuit $C'$ quasi-emulates the cycluit $C$: note that $C'$ is
  still a cycluit, because it is a union of acyclic circuits which is disjoint
except for the input gates. The two problems are that $C'$ (and hence
the overall computation process) is of quadratic size, and that the treewidth
  of~$C'$ is no longer constant (each bag contains a linear number of groups of
  gates, each group having been produced by the application of Theorem~\ref{thm:rmcyclemono}).

To explain how these problems can be addressed, we recall the specifics of the
construction used to prove Theorem~\ref{thm:rmcyclemono}. For each bag $b$ of
the rooted tree decomposition $T''$ to which it is applied, the construction creates a set of
gates $g_b^{S^+, \gamma}$ for some subsets $S^+$ and for some $\gamma$ (along
with some intermediate gates, e.g., $g_b^{S^+,\gamma,i}$), whose
inputs are gates $g_{b_1}^{\ldots}$ and $g_{b_2}^{\ldots}$ for the two children
of~$b$ in~$T''$, and these gates describe the behavior of the subtree of $T''$
rooted at~$b$. In the process that we explained, for each internal bag
$b$ of the original normal-form tree decomposition~$T$, we have considered $b$
across all the rootings $T^{b'}$ of~$T$ for all $b'$ being the $\phi$-image of a
potential output gate. For each
such choice of root bag $b'$, we have applied Theorem~\ref{thm:rmcyclemono}, so that in the bag $b''$
obtained as a final result of the process, we have one group of gates $g_b^{S^+,
\gamma}$ (with their intermediate gates) for each choice of~$b'$. However, the
crucial observation is that there are only
three possible groups, up to equivalence. Indeed, for each internal
non-root bag $b$, we can classify the possible root bags $b'$ depending on their
orientation relative to~$b$: as $b$ has degree~$3$, there are
only three possibilities. Specifically, for any two choices of root $b'_1$ and $b'_2$ that
have the same orientation relative to~$b$, the gates $g_b^{S^+, \gamma}$ were
created for the same directed subtree of~$T$ rooted at~$b$,
so they are equivalent: the same is in fact true of all the intermediate gates
at~$b$, and indeed all gates for all descendants of~$b$ according to this
orientation.
In other words, the bottom-up construction of
Theorem~\ref{thm:rmcyclemono} at node~$b$ does the exact same thing for $b'_1$
and for $b'_2$.

From this observation, we can rewrite the result of the naive process to have
linear size and to satisfy the treewidth bound, by taking
each bag $b$ of $T$, considering the $O(C)$ groups of gates created in the
analogue $b''$ of $b$ in the pre-tree decomposition $T''$ obtained as a result
of the naive process, and merge these groups depending on the
orientation of the root used for each group relative to~$b$, so that only three
groups of gates remain.
This ensures that the
result has linear size, and that the width bound is respected (i.e., the width
is only multiplied by a factor of 3 relative to the width in the output of
Theorem~\ref{thm:rmcyclemono}). By what precedes, the gates that we merge are
equivalent, and, as they occur in the same bags of~$T''$, the result is still a
pre-tree decomposition. Further, the result of this process is still an acyclic
circuit: indeed, if we perform the merges bottom-up following some rooting
of~$T''$, it is clear that whenever we merge two gates then they have the same
input gates, so no cycles can be introduced. Hence, the result of this process
is a linear-sized circuit obeying the treewidth bound, and it still
quasi-emulates $C$, so it satisfies all of the required conditions.

The last thing to justify is that, instead of performing the quadratic-time
construction and merging the equivalent gates, we can directly construct the
final output with the merged gates in linear time.
We now explain how to do so. For each bag~$b$ 
of~$T$ and each neighbor $b'$ of~$b$ in~$T$,
we call $S_b(b')$, the gates to create at bag~$b$ for the rootings of~$T$ at a
$\phi$-image whose unique simple path to~$b$ goes via~$b'$.
We then start by performing the entire bottom-up construction of
Theorem~\ref{thm:rmcyclemono} for some rooting $T^{b_\r}$ of~$T$, where $b_\r$
is the $\phi$-image of some arbitrary gate: for each bag~$b$, this creates the
gates $S_b(b')$ for $b'$ the parent of~$b$ in~$T^{b_\r}$.
Second, we perform an analogous process
in a top-down pass on the tree $T$ to reach the leaves, i.e., the
$\phi$-images of the other output gates, so as to cover all the other possible
rootings. We explain what the top-down process does on 
a bag $b$ with parent
$b'$ and children $b_1$ and $b_2$ in~$T^{b_\r}$. Recalling that $S_b(b')$ has
already been created, the top-down process must create $S_b(b_1)$ and $S_b(b_2)$.
Recall that the gates of $S_b(b_1)$ will stand for 
the subtree rooted at~$b$
with children $b'$ and $b_2$: their inputs are 
gates of $S_{b_2}(b)$ and
gates of~$S_{b'}(b)$. The first kind of gates have already been created by the
bottom-up process at~$b_2$, and the second kind of gates have been created
previously in the top-down process, because $b'$ was visited before~$b$.
So the top-down process simply creates the gates of~$S_b(b_1)$
as in the inductive step of the construction of
Theorem~\ref{thm:rmcyclemono} from these inputs. Likewise, the inputs to the gates of~$S_b(b_2)$
are gates of~$S_{b_1}(b)$ and gates of~$S_{b'}(b)$, which have already been
created.
We can see that this process creates the same gates in each bag as
those created by the naive process, except that we do not create the multiple
equivalent sets of gates, so its output satisfies the required conditions.
Further, it is easy to see that this process is in linear time, at it consists
of one bottom-up application of Theorem~\ref{thm:rmcyclemono}, and one top-down traversal
of the tree while performing two times at each node the inductive step of the
construction of Theorem~\ref{thm:rmcyclemono}. This concludes the proof.
\end{proof}

\subsubsection{Rewriting stratified cycluits}

We now use Theorem~\ref{thm:rmcyclemono2} to conclude the proof of
Theorem~\ref{thm:rmcycle}:

\begin{proof}[Proof of Theorem~\ref{thm:rmcycle}]
  Given a stratified cycluit $C$ of treewidth $k$,
assuming (up to adding an additional gate) that its output gate is
a potential output (recall Definition~\ref{def:potentialout}),
we transform $C$ in linear time to an arity-two cycluit by the same construction as in
the proof of Lemma~\ref{lem:getarity2}. Indeed, this construction still applies to non-monotone cycluits,
because it does not rely on the semantics of the gates except the associativity
of OR- and AND-gates (NOT-gates have arity~1 so they do not need to be
rewritten); and the construction clearly does not affect
the fact that the cycluit is stratified.
We now compute
a stratification function for $C$ in linear time by
Proposition~\ref{prp:stratifun}. We write
$C_1, \ldots, C_m$ the strata of~$C$, and write $C_{\leq i} \defeq \bigcup_{p
\leq i} C_p$ for brevity.
Analogously to the proof of Lemma~\ref{lem:introgates}, we now 
rewrite $C$ in linear time by re-wiring it in the following way:
whenever a wire connects a gate $g$ of stratum $i$ to a gate $g'$ of a
stratum $j > i$, we introduce an intermediate AND-gate on this wire assigned to
the stratum~$i$, so that $g$ is a potential output gate of $C_{\leq i}$.
It is clear that this only doubles the size of bags in the tree decomposition.
This process is intuitively designed to ensure that, whenever we have rewritten a prefix
of the strata to a circuit $C'$ that emulates
it (using
Theorem~\ref{thm:rmcyclemono2}), the rewriting of the next stratum can proceed
using only the gates of~$C'$.

We now compute in linear time a normal-form tree decomposition $T$ of~$C$ as in
Lemma~\ref{lem:getnormal}: clearly this still applies to non-monotone cycluits,
as it does not depend on the semantics of gates. Clearly $T$ is still a tree
decomposition of any union $\bigcup_{p \leq i}$ of lower strata, and, if we
restrict the bags to the gates of this union, it is still a normal-form tree
decomposition.

We now describe the final rewriting process by induction on the strata of~$C$.
The invariant is that, once we have processed the strata $C_1, \ldots, C_i$, we
have obtained a circuit $C_i'$ which 
emulates $C_{\leq i}$,
  with a tree decomposition $T'_{\leq i}$ with same skeleton as~$T$ where,
for each bag~$b'$, letting $b$ be the corresponding bag in~$b$, $\card{\dom(b')}$ is within
some factor (depending only on~$k$, and in $O(2^{k^\alpha})$ for some
constant~$\alpha$) of $\card{\dom(b) \cap G_i}$, where $G_i$
are the gates of $C_{\leq i}$, and where each potential output gate $g$ of
$C_{\leq i}$ occurs at least in the same bags as it does in~$T$. To achieve the
overall linear running time, it suffices to ensure that each stratum $C_i$ is
processed in time linear in $C_i$ and in the union of connected subtrees of $T$ (which we can
precompute for each $i$) that contains gates of stratum~$C_i$. The reason why this
ensures linear running time overall is that each bag of~$T$ is visited a linear
number of times for each stratum that includes a gate in it, the number of which
is a constant depending only on~$k$. Note that, because of the fact that we only
visit for each stratum the union of connected subtrees of~$T$ containing gates of this
  stratum, the tree decomposition~$T'_{\leq i}$ that we compute at each stratum is
technically not fully materialized (we do not materialize some parts where all
bags have empty domain), but by a slight abuse we will nevertheless see it as
having same skeleton as~$T$.

The base case is that of the empty circuit, and of a tree decomposition with
same skeleton as~$T$ and empty bags (which of course we do not need to
materialize).

For the induction step, we consider stratum $C_i$, and let $T_i$ be the
normal-form tree
decomposition of~$C_i$ with same skeleton as~$T$ obtained by restricting the
bags of~$T$ to the gates of~$C_i$. By definition of a tree decomposition, the
bags with non-empty domains form connected subtrees of~$T_i$, which we can
process independently (as in this case the corresponding gates have no wires
connecting them in~$C_i$), and to which we can restrict our attention to ensure
the linear time requirement for stratum~$i$. 
To represent
in~$C_i$ the missing gates from lower strata (remembering that they are
guaranteed to be potential
output gates in $C_{\leq i-1}$), we add these gates as \emph{input gates},
putting them in~$C_i$, and putting each of them in~$T_i$ only at the leaf bags
which are $\phi$-images of a gate that uses one of them as input.

Now, we can apply
Theorem~\ref{thm:rmcyclemono2} to each disjoint part, the result of which is a
circuit $C'_i$ that 
emulates 
$C_i$, with a (not fully materialized) tree decomposition $T_i'$ having same
skeleton as~$T$; and clearly the bags that were empty in~$T_i$ remain empty
in~$T_i'$, and the size of the others depends on the size of the corresponding
bags in~$T_i$ by a constant factor depending only on~$k$ and satisfying the
prescribed bound, as can be observed
from the construction of Theorems~\ref{thm:rmcyclemono}
and~\ref{thm:rmcyclemono2}. (Note that, in this accounting, we ignore the input gates that
we added, as they will be accounted for by
their occurrence in the lower strata.) We now ensure that all potential output gates
of~$C_i$, which for now only appear in the bag of~$T'_i$ corresponding to their
$\phi$-image in~$T_i$ (i.e., in~$T$), occur in all bags where they appear
in~$T_i$; this does not affect the width requirement, as each gate thus added to
a bag of~$T_i'$ is added for a gate (namely, the same gate) that was in the
corresponding bag in~$T_i$, and for which no gate had been added in~$T_i$ yet.

What remains to be done is to connect this rewriting $C'_i$ of $C_i$ to the
rewriting $C'_{\leq i-1}$ obtained by induction for the lower strata, and produce
the circuit $C'_{\leq i}$ and tree decomposition $T'_{\leq i}$ satisfying the
conditions. We simply do this by unioning $C'_i$ with $C'_{\leq i-1}$, and
substituting to the input gates added to~$C'_i$ the actual gates of~$C'_{\leq
i-1}$: as $C'_{\leq i-1}$ emulates $C_{\leq i}$, it emulates these gates, as they
must be potential output gates of~$C_{\leq i}$ thanks to the 
rewriting that we perform analogously to Lemma~\ref{lem:introgates}. It is
clear by composition that $C'_{\leq i}$ has the correct semantics, 
and of course this kind of substitution does not affect the acyclicity
of~$C'_{\leq i}$. The tree decomposition $T'_{\leq i}$ is constructed by unioning
$T'_i$ with $T'_{\leq i}$ as computed in the induction (this only needs to
traverse the bags where $T'_i$ is materialized), and substituting the input
gates; the fact that the result is still a tree decomposition of the result is
clear for the occurrences of the gates of stratum $i$ (which are the
same as in~$T'_i$), and of the occurrences of the gates of lower strata, because each time a
fresh input gate is substituted to a gate $g'$ of a lower stratum, the
construction ensures that this is within a bag which is the $\phi$-image of a
gate $g$ of stratum $i$ with $g'$ as input, and the occurrence of $g'$ in the
corresponding bag in~$T$ ensures that $g'$ must occur in that bag in~$T'_{\leq
i - 1}$. Likewise, $T'_{\leq i}$ covers all wires, because this needs only be checked
for the wires across strata, but it is the case for the same reason as we just
explained. This concludes the proof of the induction case.

We now check that the claimed treewidth bound of $O(2^{k^\alpha})$ for
some~$\alpha$ is respected, by accounting for the transformations that we have
performed:

\begin{itemize}
  \item We have added one additional gate to the input cycluit at the very
    beginning.
  \item We have transformed it to arity-2 as in Lemma~\ref{lem:getarity2}, so
    squaring the bag size.
  \item We have applied a process analogous to that of Lemma~\ref{lem:introgates}, which doubles
    the bag size.
  \item We have chosen the tree decomposition to be regrouped as in
    Lemma~\ref{lem:getregrouped},
    which multiplies the bag size by 3,
    then normal-form as in
    Lemma~\ref{lem:getregrouped}, which does not change the bag size.
  \item Now, we have performed the induction on strata, multiplying the bag size
    by at most~$2^{k^{\alpha'}}$ for some~$\alpha'$, using
    Theorem~\ref{thm:rmcyclemono2}.
\end{itemize}

Hence, we can clearly respect the prescribed bound of $O(2^{k^\alpha})$ for
some $\alpha$.

The result of the induction is indeed a circuit $C'$ which 
emulates
$C_{\leq m} = C$ (in particular, it emulates the output gate
of~$C$, as it was a potential output gate), with a tree decomposition~$T'$ where
each bag size is within the requested bounds, and computed in overall linear
time.

This concludes the proof of Theorem~\ref{thm:rmcycle}.
\end{proof}

\end{toappendix}

\begin{proofsketch}
  The proof is technical and proceeds stratum by stratum, so it focuses on
  monotone cycluits. For such circuits, we rewrite tree decompositions to a normal form
  that ensures that gate definitions occur only in leaf bags, and with all
  relevant input gates in scope. We then perform our rewriting bottom-up on the
  tree decomposition, by creating gates at each bag to code the behavior of the
  sub-circuit in terms of the values of the yet-undefined internal gates propagated from the
  parent bag: this requires a fixpoint computation, coded by iterating in the
  circuits for the two subtrees rooted at the children of the current bag. We
  show it is equivalent to standard cycluit evaluation, and that the number of
  required iterations is bounded by the bag
  size. Further, as higher strata may depend on any gate (not just on one
  single output), we use a second top-down pass to perform the previous process
  for each rooting of the tree decomposition in overall linear time.
\end{proofsketch}

\subparagraph*{Probabilistic query evaluation.}
We can then apply the above result to the probabilistic query evaluation (PQE)
problem, which we now define:

\begin{definition}
A \emph{TID instance} is a relational instance $I$ and a function $\pi$
mapping each fact $F \in I$ to a rational probability $\pi(F)$. A TID instance
$(I, \pi)$ defines a probability distribution $\Pr$ on $I' \subseteq I$, where
$\Pr(I') \defeq \prod_{F \in I'} \pi(F) \times \prod_{F \in I \backslash I'} (1 -
  \pi(F))$.

The \emph{probabilistic query evaluation} (PQE) problem asks, given a Boolean
query $Q$ and a TID instance $(I, \pi)$, the probability that the query $Q$ is
  satisfied in the distribution $\Pr$ of~$(I, \pi)$. Formally, we want to compute
$\sum_{I' \subseteq I\text{~s.t.~}I' \models Q} \Pr(I')$. The \emph{data
complexity} of PQE is its complexity when $Q$ is fixed and the TID instance $(I,
\pi)$ is given as input. Its \emph{combined complexity} is its complexity when
both the query and TID instance are given as input.
\end{definition}

Earlier work \cite{dalvi2007management} showed that PQE has \#P-hard data complexity
even for some CQs of a simple form, but 
\cite{amarilli2015provenance,amarilli2015provenance_extended} shows that PQE is tractable in data complexity
for any Boolean query in monadic second-order (MSO) if the input instances are
required to be treelike.

We now explain how to use Theorem~\ref{thm:rmcycle} for PQE. Let
$P$ be an ICG-Datalog program of body size $\kp$. Given a TID instance
$(I,\pi)$ of
treewidth~$\ki$, we compute 
a provenance cycluit for $P$ on $I$ of treewidth $O(|P|)$
in FPT-linear time in $\card{I}\cdot\card{P}$
by Theorem~\ref{thm:mainprov}.
By Theorem~\ref{thm:rmcycle}, we 
compute in $O(2^{\card{P}^\alpha}\card{I}\card{P})$ an equivalent circuit of treewidth
$O(2^{\card{P}^\alpha})$. Now, by 
Theorem D.2 of~\cite{amarilli2015provenance_extended}, we can solve
PQE for $P$ and $(I,\pi)$ 
in $O(2^{2^{\card{P}^\alpha}}\card{I}\card{P}+\card{\pi})$ up to PTIME arithmetic
costs. Linear-time data complexity was known
from~\cite{amarilli2015provenance}, but 2EXPTIME combined complexity
is novel, as \cite{amarilli2015provenance} only gave
non-elementary combined complexity bounds.

\begin{toappendix}
\subsection{Hardness of PQE}
\end{toappendix}
\subparagraph*{Acyclic queries on tree TIDs.}
A natural question is then to understand whether better 
bounds are possible. 
In particular, is PQE tractable in \emph{combined complexity} on
treelike instances?
We show that, unfortunately, treewidth bounds are
\emph{not} sufficient to ensure this.
The proof draws some inspiration
from earlier work \cite{kimelfeld2008query} on the topic
of tree-pattern query evaluation in \emph{probabilistic XML}
\cite{kimelfeld2013probabilistic}.

\begin{propositionrep}\label{prp:tpqtpq}
There is a fixed arity-two signature on which PQE is \#P-hard even when imposing
  that the
input instances have treewidth~$1$ and the input queries are
$\alpha$-acyclic CQs.
\end{propositionrep}

\begin{proof}
  We reduce from the \#P-hard problem
  \#MONOTONE-2-SAT~\cite{Valiant:1979}, that asks, given a
conjunction $\Phi$ of disjunctions of positive literals over variables $x_1,
\ldots, x_n$, the number of
assignments that satisfy $\Phi$.

  We consider a signature $\sigma$ formed of a unary relation $R$, and of three binary
  relations~$V$, $K$, and~$C$.

Given $\Phi$, we encode it in PTIME to a TID instance $(I, \pi)$ which comprises the
following facts:

\begin{itemize}
  \item One fact $R(r)$.
\item For each variable $v_i$ of $\Phi$, the fact $V(r, v_i)$ with
  probability $1/2$ in $\pi$, which intuitively codes the valuation of
    variable~$v_i$.
\item For each clause $C_j$ of $\Phi$ that contains variable $v_i$, the following gadget:
  \begin{align*}\text{a length-$j$ path } K(v_i, c_{i,j,1}), K(c_{i,j,1}, c_{i,j,2}), \ldots, K(c_{i,j,j-1},
    c_{i,j,j})\phantom{,}\\
  \text{ and the fact } C(c_{i,j,j}, c_{i,j,j}'),\end{align*} each of these facts having
probability~$1$ in~$\pi$. Intuitively, to remain on a fixed signature, we write the clause
number in unary as the path length.
\end{itemize}

It is immediate that $I$ is a tree, so it has the prescribed treewidth.
We now construct in PTIME the query $Q$ comprising the following atoms:

\begin{itemize}
  \item One atom $R(x)$.
\item For each clause $C_j$, one atom $V(x, y_j)$, one length-$j$ $K$-path from
$y_j$ to a variable $y_{j,j}$ and the fact $C(y_{j,j}, y_{j,j}')$.
\end{itemize}

Again, it is immediate that $Q$ is acyclic.

We now claim that the probability of $Q$ on $(I, \pi)$ is exactly the number of
satisfying assignments of $\Phi$ divided by~$2^n$, so that the computation of
one reduces in PTIME to the computation of the other, concluding the proof. To
see why, we define a bijection between the valuations of $v_1, \ldots, v_n$ to
the possible worlds $J$ of $(I, \pi)$ in the expected way: retain fact
  $V(r, v_i)$
iff $v_i$ is assigned to true in the valuation. Now, observe that $J$ satisfies
$Q$ iff the corresponding valuation makes $\Phi$ true. Indeed, from the
satisfaction of $Q$ by $J$, for any~$j$, observing the element to which $y_{j,j}$ is matched,
and observing its ancestor $v_{i_j}$ in~$I$, we deduce that $v_{i_j}$ must be
  made true by the valuation, and, by construction of~$I$, the variable
  $v_{i_j}$ appears in $C_j$, hence $C_j$ is true in~$\Phi$. Hence, $\Phi$ is true because every clause is
true. Conversely, assuming that the valuation satisfies $\Phi$, we construct a
match of $Q$ on $I$ by mapping each branch of $Q$ via the witnessing variable
for that clause in the valuation of~$\Phi$. This concludes the proof.
\end{proof}

\subparagraph*{Path queries on tree TIDs.} We must thus restrict the
query language further to achieve combined tractability.
One natural restriction is to go from \emph{$\alpha$-acyclic queries} to \emph{path
queries}, i.e., Boolean CQs of the form $R_1(x_1, x_2),
  \ldots, \allowbreak R_n(x_{n-1}, x_n)$, where each $R_i$ is a binary relation of the
  signature.
For instance, $R(x, y), S(y, z), T(z, w)$ is a path query, but $R(x, y), S(z, y)$
is not 
(we do not allow inverse relations). We can strengthen the
previous result to show:

\begin{propositionrep}\label{prp:path}
There is a fixed arity-two signature on which PQE is \#P-hard even when imposing
that the input instances have treewidth~$1$ and the input queries
are path queries.
\end{propositionrep}

\begin{proof}
We adapt the previous proof. We 
  extend the signature to include one binary relation~$S_-$ for each binary relation $S$. We modify
the definition of the instance $I$ to add, whenever we created a fact $S(a, b)$,
the fact $S_-(b, a)$, each of these inverse facts being given 
probability~$1$ in~$\pi$. It is clear that $I$ still has treewidth~$1$, as we
  can just use the same tree decomposition as before. We now define the path query $Q'$ as follows, following a
traversal of the query $Q$ of the previous proof:

\begin{itemize}
\item $V(x_1, y_1), K(y_1, y_{1,1}),
C(y_{1,1}, y_{1,1}'), C_-(y_{1,1}', y_{1,1}''), K_-(y_{1,1}'',
y_1''), V_-(y_1'', x_2)$;
\item $V(x_2, y_2), \allowbreak
K(y_2, y_{2,1}),
K(y_{2,1}, y_{2,2}), \allowbreak
C(y_{2,2}, y_{2,2}'), C_-(y_{2,2}', y_{2,2}''), \allowbreak
K_-(y_{2,2}'', y_{2,1}''),$\\
$K_-(y_{2,1}'', y_2''),
V_-(y_2'', x_3)$;
\item etc.
\end{itemize}

It is straightforward to observe that, when inverse facts are added like we did,
$Q$ has a match $M$ on $I$ iff $Q'$ has this same match: constructing the match
of $Q'$ from that of $Q$ is trivial, and any match of $Q'$ is a match of $Q$
thanks to the fact that, $I$ being a tree, each element of $\dom(I)$ has at most
one ingoing inverse fact, so that each inverse fact must in fact be mapped to
the same element as the corresponding fact that was traversed earlier. This
concludes.
\end{proof}

\subparagraph*{Tractable cases.}
In which cases, then, could PQE be tractable
in combined
complexity? One example is in~\cite{cohen2009running}: 
PQE is tractable in combined complexity 
over \emph{probabilistic
XML}, when queries are written as \emph{deterministic} tree automata.
In this setting, 
that the edges of the XML document are
\emph{directed} (preventing, e.g., the inverse construction used in the proof of
Proposition~\ref{prp:path}). Further, as the result works on \emph{unranked}
trees, it is important that children of a node are \emph{ordered} as well (see
\cite{amarilli2014possibility} for examples 
where this matters).

We
leave open the question of whether there are some practical classes of
instances and of queries
for which such a deterministic tree automaton can be obtained
from the query in polynomial time to test the query for a given
treewidth. As we have shown, path queries and instances of
treewidth~1, even though very restricted, do not suffice to ensure
this.
Note that, in terms of data complexity, we have shown
in~\cite{amarilli2016tractable} that treelike instances are essentially the only
instances for which first-order tractability is achievable.

\section{Conclusion}
We introduced ICG-Datalog, a new stratified Datalog fragment whose
evaluation has FPT-linear complexity when parameterized by instance
treewidth and program body size. The complexity result is obtained via 
compilation to alternating two-way automata, and via the computation of a
provenance representation
in the form of stratified cycluits, a generalisation of
provenance circuits that we hope to be of independent interest.

We believe that  ICG-Datalog can be further improved by removing the guardedness
requirement on negated atoms, which would make it more expressive and step back from the world of
guarded negation logics. In particular, we conjecture that our FPT-linear
tractability result generalizes to \emph{frontier-guarded Datalog}, and its
extensions with clique-guards and stratified (but unguarded) negation, taking
the rule body size and instance treewidth as the parameters. We further hope that
our results could be used to derive PTIME combined complexity results
on instances of arbitrary treewidth, e.g., XP membership when parametrizing by
program size; this could in particular recapture the tractability of
bounded-treewidth queries.
Last, we intend to extend our
cycluit framework to support more expressive 
provenance semirings than Boolean provenance (e.g., formal power series~\cite{green2007provenance}).

We leave open the question of practical implementation of the methods we
developed, but we have good hopes that this approach can give efficient
results in practice, in part from our experience with a preliminary
provenance prototype~\cite{monet2016probabilistic}. Optimization is possible, for
instance by not representing the full automata but building them on the
fly when needed in query evaluation. Another promising direction
supported by our experience, to deal with real-world datasets that are
not treelike, is to use partial tree decompositions~\cite{maniu2014probtree}.

\subparagraph*{Acknowledgements.}
This work was partly funded by the Télécom ParisTech Research Chair on Big Data
and Market Insights.

\bibliographystyle{abbrv}
\bibliography{main}

\begin{toappendix}
  \section{Results from~\cite{unpublishedbenediktmonadic}}
  \label{app:mdl}
  This appendix contains some results
from~\cite{unpublishedbenediktmonadic}, an extended version
of~\cite{benedikt2012monadic} that deals with the containment of monadic
Datalog programs, but which is currently unpublished. Relevant
parts of this work are reproduced here for completeness.

\subsection{Treelike canonical set of instances of a Datalog program}

\begin{definition}
  A \emph{canonical set of instances} for a Datalog program $P$ is a
  (generally infinite) family~$\calI$ of instances that all satisfy $P$ and such that,
  for any CQ $Q$, if there is an instance $I$ satisfying $P \wedge
  \neg Q$, then there is an instance in $\calI$ with the same property.
\end{definition}

\newcommand{\ext}{\mathrm{ext}}

\begin{definition}
  The \emph{var-size} of a Datalog program~$P$ is the maximal number of variables used in a rule of $P$. 
\end{definition}

The following result from~\cite{unpublishedbenediktmonadic} shows that
any Datalog program has a bounded-treewidth set of canonical instances,
an encoding of which can be described by a bNTA constructible in
exponential time.
It heavily relies on notions introduced in~\cite{chaudhuri1997equivalence}.
\begin{lemma}
  \label{lem:caninst}
  For all $c \in \NN^*$,
  letting $\ki \defeq 2c-1$,
  for any signature $\sigma$, given a Datalog program $P$ of
  var-size bounded by $c$,
  we can compute in exponential time in $P$ a
  \mbox{$\Gamma_\sigma^{\ki}$-bNTA}~$A_P$
  such
  that the set of the instances obtained by decoding the trees in the language
  of~$A_P$ is canonical for~$P$.
\end{lemma}

\begin{proof}
This proof is based on the notion of \emph{unfolding expansion tree} of a
Datalog program~$P$ (Definition 2.4 of~\cite{chaudhuri1997equivalence}).

  An \emph{expansion tree} of a Datalog program is a ranked tree
  (non-binary in general) where each node is
labeled by an \emph{instantiated} rule~$r$ of $P$ (i.e., an homomorphic
  image of the rule by some \emph{one-to-one} mapping from the variables of the rule to
  some set of variables), and has a child for each
  intensional predicate
  atom~$A$ appearing in~$r$, whose label is a rule~$r'$ with $A$ for
  head, with the variables of $A$ mapped to the same variables as in~$r$.
  We further require that the root of the tree has a rule with
  the goal predicate in the head.
  
  An \emph{unfolding expansion tree} is an expansion tree with the
  additional condition that, for each node $n$ of the
  tree, each variable occurring in the
body of the rule labeling~$n$ but not in its head does not occur in the
  label of any ancestor of~$n$. In other words, the same variable is
  never re-used across rules unless the variable is propagated through
  the head.

 From an unfolding expansion tree $t$, it is possible to get an instance over the
  extensional signature satisfying $P$ as follows. Let $\nu$ be a one-to-one
  mapping from the variables in the rules labeling the nodes of~$t$ to
  constants.
  We denote by $\nu(t)$ the tree obtained by replacing each variable
  appearing in the labels of $t$ by its image by $\nu$.
  We note that $t$ and $\nu(t)$ are identical up to renaming the values
  in variables. Let $\Pi_{\ext}(\nu(t))$ be the tree obtained by keeping
  only the facts over the extensional signature in each label of a node
  of~$\nu(t)$. We denote by $I(\Pi_{\ext}(\nu(t)))$ the instance composed
  of the facts appearing in $\Pi_{\ext}(\nu(t))$. It is clear that
  $I(\Pi_{\ext}(\nu(t)))$ satisfies $P$: indeed, a simple bottom-up
  induction on $\nu(t)$ shows that it only contains
  intensional facts that are derivable by $P$ from $I(\Pi_{\ext}(\nu(t)))$.
  
  First note that, for $P$ a Datalog program, the
  set of instances of the form~$I(\Pi_{\ext}(\nu(t)))$, where $t$ is un
  unfolding expansion tree of $P$ and $\nu$ is some fixed one-to-one mapping to
  constants, is a canonical set of instances
  for~$P$. This comes from Proposition 2.6
  of~\cite{chaudhuri1997equivalence}: the Datalog program is equivalent
  to the infinite disjunction of the $\Pi_{\ext}(t)$, where $t$ is an
  unfolding expansion tree of $P$, each $\Pi_{\ext}(t)$ 
  being seen as a conjunctive query. So for any CQ $Q$, if $I\models
  P\land\lnot
  Q$, then in particular $I\models \Pi_{\ext}(t)$ for some
  unfolding expansion tree~$t$. But then $I(\Pi_{\ext}(\nu(t)))\not\models
  Q$, since $I(\Pi_{\ext}(\nu(t)))$ is a canonical model of $\Pi_\ext(t)$
  and $\Pi_\ext(t)$ has a model that does not satisfy $Q$.
  
  In Section 5.1 of~\cite{chaudhuri1997equivalence}, the notion of
  \emph{proof tree} is introduced. A proof tree for a Datalog program~$P$
  is defined as an expansion tree over a finite fixed set of variables
  $\{x_1,x_2,\dots,x_{2s}\}$ where $s$ is the var-size of $P$.

  In the proof of Proposition 5.6 of~\cite{chaudhuri1997equivalence}, it
  is shown that for every unfolding expansion tree~$t$, there is an
  associated proof tree $t'$ obtained by a mapping $\mu$ such that $t'
  =\mu(t)$. Now, observe that $\Pi_{\ext}(t')$ can be seen as a form of 
  a tree encoding
  (though not of the same form as our tree
  encodings: 
  the domain is a subset of the $x_i$'s, facts are
  extensional facts of the instance) of $I(\Pi_{\ext}(\nu(t)))$ for any one-to-one
  mapping~$\nu$ from variables to constants. This witnesses
  $I(\Pi_{\ext}(\nu(t))$ is of treewidth $\leq 2s-1$.

  For any Datalog program of var-size $c$,
  Proposition 5.9 of~\cite{chaudhuri1997equivalence} shows that there
  exists a tree automaton of size exponential in $P$ recognizing the
  proof trees of $P$ (and the proof of this result makes it clear that
  this automaton is computable in exponential time). From such a tree
  automaton (running on the ranked expansion trees), by projection, we can construct in polynomial time in the
  size of the automaton a bNTA $A$ recognizing the $\Pi_{\ext}(t')$ for
  any proof tree $t'$. We now apply the technique of the proof of
  Proposition B.1 in~\cite{amarilli2015provenance_extended} to transform $A$,
  in time polynomial in~$A$, into a $\Gamma_\sigma^{\ki}$-bNTA $A_P$ that
  recognizes $(\sigma,{\ki})$-tree encodings of the $I(\Pi_{\ext}(\nu(t)))$
  for an arbitrary unfolding expansion tree $t$ and one-to-one mapping $\nu$
  to constants. A technical detail imposed by the technique
  of~\cite{amarilli2015provenance_extended} is that we must ensure that each rule
  of the Datalog program has either 0 or 2 intensional facts; we can
  always do that by an initial transformation of the Datalog program,
  introducing intermediate intensional predicates for rules with more
  than 2 intensional facts or adding a trivial intensional predicate atom for
  rules with 1 intensional fact. Note that the var-size of the Datalog
  program is not affected by this transformation.
  
  We have thus shown that it is possible to compute $A_P$ in exponential
  time, and that the set of instances in the decoded language of $A_P$ is
  canonical for $P$.
\end{proof}

\subsection{2EXPTIME-Hardness of Treelike CQ Validity over Valid Trees}

\renewcommand{\bold}[1]{{\ifmmode\bm{#1}\else{\boldmath\bfseries#1}\fi}}

\newcommand{\Succ}{\mathsf{Succ}}
\newcommand{\ct}{\mathsf{CT}}
\newcommand{\nct}{\mathsf{NextCT}}
\newcommand{\ns}{\mathsf{TTCh}}
\newcommand{\content}{r}
\newcommand{\samecell}{\mathsf{SameCell}}
\newcommand{\asch}{\mathsf{Sch}}
\newcommand{\ssimple}{\mathcal{S}^{\binary}_{\first,\second}}
\renewcommand{\root}{\mathsf{Root}}
\newcommand{\leaf}{\mathsf{Leaf}}

We consider
``universality'' or ``validity'' problems for queries over trees:
given a schema describing a set of trees and a Boolean query over trees,
does every tree satisfy the query.

Let $\asch$ be a finite set of labels.
The \emph{relational signature of ordered, labeled, binary trees}, denoted
$\ssimple$, is made out of
the binary predicates $\firstchild$, $\secondchild$,
unary $\root$, $\leaf$ predicates,
and $\rlabel_\alpha$ predicates for all $\alpha\in\asch$.

We denote as $\schildself$
the relational signature containing all the relations
of $\ssimple$ together with binary
$\child$ and $\childself$ relations.

A tree $T$ over $\ssimple$ is
a relational instance such that:
\begin{compactenum}[(i)]
  \item the non-empty $\rlabel_\alpha^T$'s for $\alpha\in\asch$ form a partition of
    $\dom(T)$ (one can thus talk about the \emph{label} of a
    node~$n$, which is the $\alpha\in\asch$ such that $n\in\rlabel_\alpha^T$);
  \item $\firstchild^T$ and $\secondchild^T$ are one-to-one partial mappings with
    the same domain (the set of \emph{internal nodes}), whose
    complement is $\leaf^T$ (the set of \emph{leaves}), and with disjoint ranges;
  \item the inverses of $\firstchild^T$ and $\secondchild^T$ are
    one-to-one partial mappings;
  \item $\exists x\,\firstchild(x,x)\lor\secondchild(x,x)$ does not hold;
  \item $\root^T$ contains exactly one element (the \emph{root} $r$
    of~$T$), and the following formula does not hold for~$r$:
    $\exists x\, \firstchild(x,r)\lor\secondchild(x,r)$.
\end{compactenum}

A \emph{tree} $T$
over $\schildself$ is a relational instance that verifies the same
axioms as a tree over $\ssimple$, where
$\child^T$ is the disjoint union of $\firstchild^T$ and
$\secondchild^T$, and where the following formula holds:
$\childself(x,y)\leftrightarrow(\child(x,y)\lor x=y)$.

A Boolean query on one of the signatures above is \emph{valid} over a
bNTA if for all trees that satisfy the schema, the
query returns true.

We can now prove the following result, which closely tracks Theorem 6
of~\cite{bjoerklund_2008}.

\begin{theorem}
\label{th:DtDValid}
  Given a CQ $Q$ on $\schildself$ of treewidth $\leq 2$ 
  and a bNTA $A$, it is 2EXPTIME-hard to decide whether $Q$ is valid over~$A$.
\end{theorem}

\begin{proof}
 We adapt the proof of Theorem 6 of~\cite{bjoerklund_2008}, given in
 Appendix C.2 of~\cite{bjorklund_2008_extended},
 which states that validity with respect to a bNTA
 of a CQ with \emph{child and
 descendant} predicates over \emph{unranked trees} is 2EXPTIME-hard. We adapt it
 by moving from unranked trees to binary trees
 (with the changes that it implies in our definition of a bNTA), writing the
 reduction using $\childself$ instead of the descendant predicate, and
  proving that the resulting CQ has bounded treewidth.
  The final CQ will be defined through intermediate subformulae, and we will use
  the following immediate observation to bound its treewidth: the
  treewidth of a CQ of the form $Q=\exists \mathbf{x}\mathbf{y}\:
  Q_1(\mathbf{x})\land Q_2(\mathbf{x},\mathbf{y})$ is the maximum of the
  width of a decomposition of $\exists \mathbf{x}\: Q_1(\mathbf x)$ where
  all variables of $\mathbf{x}$ are in the same bag, and of the treewidth
  of $Q'=\exists\mathbf{x}\mathbf{y}\: R(\mathbf{x})\land
  Q_2(\mathbf{x},\mathbf{y})$ where
  $R(\mathbf{x})$ is a single atom.

  We give a self-contained presentation keeping the notation
  from~\cite{bjoerklund_2008} as much as possible, with notable departures 
  highlighted in
  \textbf{bold font} throughout the proof.

As in~\cite{bjoerklund_2008}, we reduce from the termination of
an alternating EXPSPACE Turing Machine $M$, a 2EXPTIME-hard
problem~\cite{ChandraKS81}. The next three paragraphs are taken in part
  from~\cite{bjoerklund_2008}, with some minor simplifications, as we need to introduce the same concepts.

\newcommand{\qa}{q_{\mathrm{a}}}
\newcommand{\qr}{q_{\mathrm{r}}}
\renewcommand{\L}{{\mathrm{L}}}
\newcommand{\R}{{\mathrm{R}}}
\renewcommand{\S}{{\mathrm{S}}}

An alternating Turing machine (ATM) is a tuple $M = (\Omega, \Gamma,
\Delta, q_0 )$
where $\Omega = \Omega_\forall \uplus \Omega_\exists \uplus\{\qa\}
\uplus\{\qr\}$ is a finite set of states partitioned into
universal
states from $\Omega_\forall$, existential states from $\Omega_\exists$, an
accepting state $\qa$, and a
rejecting
state $\qr$. The finite tape alphabet is $\Gamma$. 
The
initial
state of $M$ is $q_0 \in \Omega$. The transition relation $\Delta$ is a subset of 
$(\Omega \times \Gamma ) \times
(\Omega \times \Gamma \times
\{\L, \R, \S\})$. The letters $\L$, $\R$, and $\S$ denote the directions left, right,
and stay, according to
which the tape head is moved.

An \emph{accepting computation tree} for an ATM $M$ is a finite
\emph{unranked} tree
labeled by configurations (tape content, reading head position, and
internal state) of $M$ such that
\begin{inparaenum}[(1)]
\item if node $v$ is
  labeled by an existential configuration, then $v$ has one child,
  labeled by one of the possible successor configurations;
\item if $v$ is
  labeled by a universal configuration, then $v$ has one child for each
  possible successor configuration;
\item the root is labeled by the initial
  configuration (input word on the tape, head at the beginning of the
  word, initial state); and
\item all leaves are labeled by accepting
  configurations (and accepting configurations only appear as leaves).
\end{inparaenum}
An ATM~$M$ accepts a word $w\in\Gamma^*$ if there exists an accepting
computation tree for~$M$ with $w$ as initial tape content.

The overall idea of the proof of~\cite{bjoerklund_2008}, that we closely
adapt, is as follows. Given an ATM $M$ and a word $w$ of
length $n$, we construct, in polynomial time,
\begin{inparaenum}[(1)]
\item an ATM $M_w$ which accepts
the
empty word if and only if $M$ accepts $w$; and 
\item a bNTA $A$ that checks most
important properties of (suitably encoded) computation trees of $M_w$,
except their
consistency w.r.t.\ the transition relation of $M_w$.
\end{inparaenum}
The consistency is
tested by a query $Q$ that we define. To be precise, $Q$ is
satisfied by a tree $T$ in
$L(A)$
if and only if the transition relation of $M_w$ is not respected by $t$. This
means
that $Q$ is valid w.r.t.\ $A$, iff there does not exist a consistent,
accepting
computation tree for $M_w$. Since 2EXPTIME is closed under complementation,
we conclude that validity of CQs on $\schildself$ with respect to bNTAs
is 2EXPTIME-hard.

\bold{Without loss of generality, we assume that universal states of
$M_w$ have exactly two successors, whatever the symbol read --
if they have less, we can just add transition(s) to an accepting
configuration, and if they have more, we can introduce intermediary
states to encode the $n$-ary conjunction as a tree of binary
conjunctions, with no change to the tape.}

We do not give the non-deterministic tree automaton explicitly, but trees
in its language will have \bold{the shape represented by Figure
\ref{fig:binarytree}}.\footnote{\bold{Since our definition of bNTA
requires a binary tree to be full, we need to add dummy nodes where needed, with
labels distinct from real nodes. This technicality has no impact, and we will
ignore these nodes.}}
The bold nodes are the nodes added to the tree of the proof of Theorem 19
of~\cite{bjoerklund_2008}.
The labels of dashed edges indicate the number of nodes between a node and its ancestor.
This bNTA encodes trees that represent the
executions of $M_w$. 

Each configuration in an execution is encoded by a
subtree rooted by a node labeled $\ct$; the $\ct$-node for the initial configuration appears
as the unique child of a chain of $\ell$ nodes with dummy labels from the
root for some integer $\ell$ that we will define further (this chain of $\ell$
nodes is only needed for technical reasons). A $\ct$-node has \bold{two} children
labeled $\content$ and \bold{$\nct$}. The subtree rooted by the $\content$-node
represents the tape of the configuration \bold{and the subtree rooted by
  the $\nct$-node has zero, one, or two $\ct$-children representing dummy
  nodes where needed, that we will regard as zero, one, or two following
  configurations depending whether the current state is accepting,
existential, or universal.}

A tape (under a
node $\content$) is written as a complete binary tree of depth $n$, with
leaves of the tree containing information about the $2^n$ cells. This
complete binary tree, that we will refer to as the \emph{tape tree}, is
itself encoded for querying purposes in the following fashion: for $1\leq
i\leq n$, a node with label $s$ represents a node at depth~$i$ in the
tape tree; each such node representing a node at depth~$i$, with $1\leq i\leq n-1$, has for
first child a node of label $p$ that serves as a \emph{navigation widget}
indicating the position of this node in the tape tree, and \bold{as second child a node of
label $\ns$ (for \emph{tape tree children})
that has for children two nodes of labels $s$}, encoding
the two children of the current encoded node in the tape tree. The
navigation widget is a $p$-labeled node with a single $x$-child that has
itself a single $y$-child. If the current encoded node in the tape tree
was a left child, $x=0$ and $y=1$; otherwise, $x=1$ and $y=0$ (this
widget thus encodes the $i$-th bit of the address of a cell). The node
$r$ itself has two children, the two nodes labeled with $s$ encoding the first level
of the tape tree. Finally, nodes labeled with $s$ representing a node at
depth~$n$ in the tape tree also have a
$p$-labeled navigation widget, but have as second child a $c$-node that
encodes the content of the cell.

We encode the same information about configuration tapes as
in~\cite{bjorklund_2008_extended}: symbols on basic cells, symbol and
transition followed on the current cell, and current symbol, previous
state, previous symbol on previous tape cells. This means each cell is
virtually annotated with an element of
$\Gamma\cup(\Gamma\times\Delta)\cup(\Gamma\times\Omega\times\Gamma)$:
there are polynomially many such annotations, we refer to them in the
following as $1,\dots,k$ fixing an arbitrary order. As
in~\cite{bjorklund_2008_extended}, we want to impose a number of
horizontal constraints (constraints on the annotations of neighboring
cells in a given configuration) and of vertical constraints (constraints
on the annotation of the same cell in successive configurations). These
constraints can be written as two sets of pairs $H(M_w)$ and $V(M_w)$ of
integers $1\leq i,j\leq k$, respectively, indicating respectively whether
$j$ can appear to the right of $i$ in a configuration, and whether $j$ can
appear in the same cell as $i$ in a successive configuration. We refer
to~\cite{bjorklund_2008_extended} for the full set of constraints
required.

For each cell, the $c$-node has two
children, labeled with $m$ (for \emph{me}) and $f$ (for
\emph{forbidden}), each having as descendants a chain of $k$ nodes that
can have labels either $0$ or~$1$. Only one node has label~$1$ under~$m$: the one
whose depth gives the current content of the cell. Under~$f$, for a cell at
position~$i$ in the tape, node at depth $j$ has
label~$0$ if and only if $(i,j)\in V(M_w)$.

\begin{figure}
\centering
  \begin{tikzpicture}[
      level 2/.style={level distance=8mm,sibling distance=100mm},
      level 3/.append style={sibling distance=30mm},
      level 6/.append style={sibling distance=15mm},
    ]

  \node {$\bot$} child[level distance=20mm,dashed]{node{$\ct$}
    child[solid]{node{$\content$}
      child[level distance =20mm,dashed]{node{$s$} [level distance=8mm]
        child[solid]{node{$p$}
          child{node{$0$}
            child{node{$1$}}
          }
        }
        child[solid]{node{$c$}
          child{node{$f$}
            child[level distance=20mm, dashed]{node{$\cdots$}
            edge from parent node[left]{$k$}}
          }
          child{node{$m$}
            child[level distance =20mm,dashed]{node{$\cdots$}
            edge from parent node[left]{$k$}}
          }
        }
      edge from parent node[left]{\bold{$2n-1$}}}
    }
    child[solid]{node{$\bold{\nct}$}
      child{node{$\ct$}
        child{node{$\content$}
          child{node{$s$}
            child[solid]{node{$p$}
              child{node{$0$}
                child{node{$1$}}
              }
            }
            child[solid]{node{$\bold{\ns}$}	
              child{node{$s$}
                child{node{$\cdots$}}
              }
              child{node{$s$}
                child{node{$\cdots$}}
              }
            }
          }
          child{node{$s$}
            child[solid]{node{$p$}
              child{node{$1$}
                child{node{$0$}}
              }
            }
            child[solid]{node{$\bold{\ns}$} {
              child{node{$\cdots$}}
            }
            }
          }
        }
            child{node{$\bold{\nct}$}
              child{node{$\cdots$}}
            }
      }
          child{node{$\ct$}
            child{node{$\cdots$}}
          }
      }
      edge from parent node[left]{\bold{$\ell$}}};
    ;
      \end{tikzpicture}
    \caption{General structure of trees in proof of
      Theorem~\ref{th:DtDValid}; bold labels and counters highlight changes from the proof of
      Theorem~6 
    of~\cite{bjoerklund_2008}}
    \label{fig:binarytree}
  \end{figure}
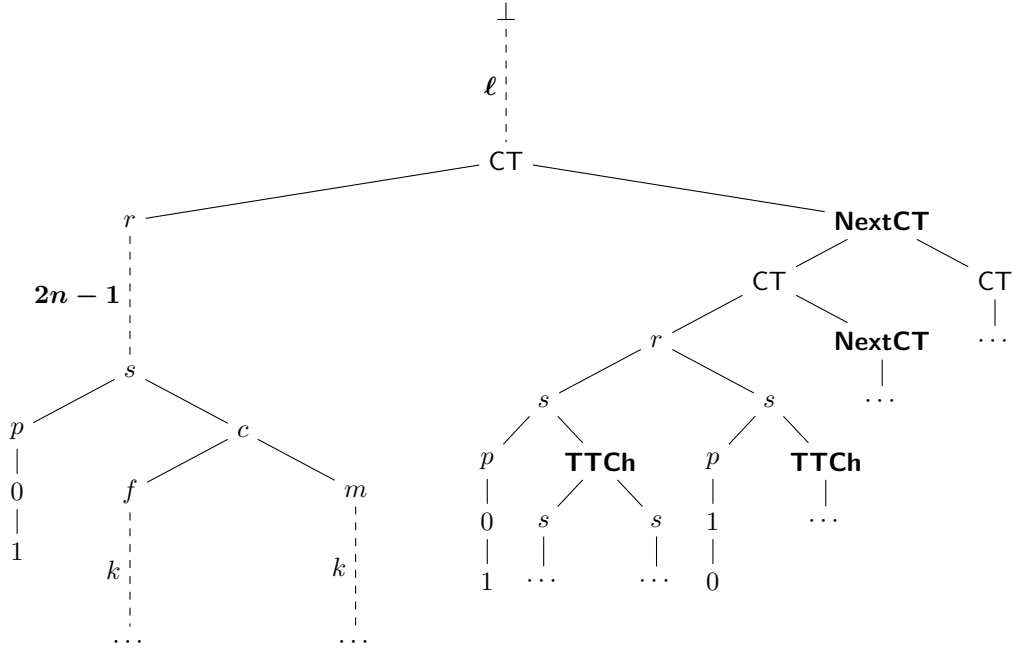

As in~\cite{bjorklund_2008_extended}, we can construct in polynomial time 
a bNTA
that enforces that all trees have the described form, including respect
of horizontal constraints, initial configuration at the root and
accepting configuration at the leaves, but \emph{excluding vertical
constraints}. Indeed, vertical constraints cannot (at least
straightforwardly) be imposed on the tree as they relate nodes of the
tree that are very far apart -- see~\cite{bjorklund_2008_extended} for
how to encode horizontal constraints and the general structure.
Modifications needed because of our binary setting are minor.
The language of this bNTA is exactly the codes of accepting computation
trees for $M_w$, except that vertical constraints may be violated.

We now construct a conjunctive query that holds if vertical
constraints are violated. In what follows, we denote by $R^i(x,y)$ the
chain $\exists x_1\dots x_{i-1}\,R(x,x_1)\land\dots\land R(x_{i-1},y)$
for $R$ a binary relation and $i\geq 1$. The query is built up from the
following subformulae:

\begin{itemize}
 \item A formula $\Succ(r_1,r_2)$ that expresses that $r_1$ and $r_2$ are
   roots
   of a tree encoding tape, with the configuration of $r_2$ being the successor of
   that of~$r_1$. Formally:
   \begin{align*}
     \Succ(r_1,r_2) \defeq \exists s_1s_2\, &\rlabel_r(r_1) \wedge
                              \rlabel_r(r_2) \\&\wedge \child(s_1,r_1)\wedge \child(s_2,r_2)
 \wedge 
\child^{\bold{2}}(s_1,s_2).\end{align*}

    \bold{There is a tree decomposition of width 2 of this subquery where both
    exported variables $r_1$ and $r_2$ are in the same bag.}

 \item A formula $\Phi_i(x,y)$ that expresses that $x$ and $y$ are
   $s$-nodes encoding a node at the $i$-th level of two tape trees,
   such that the configuration of $y$ is a successor of the configuration of $x$:
   \begin{align*}
     \Phi_i(x,y) \defeq  \exists r_1 r_2 \, &\rlabel_{s}(x)  \wedge
      \rlabel_{s}(y) \wedge \Succ(r_1,r_2) \\&\wedge \child^{\bold{ 2i-1}}(r_1,x) \wedge  
  \child^{\bold{2i-1}}(r_2,y).
\end{align*}
    \bold{There is a tree decomposition of width 2 of this subquery where both
    exported variables $x$ and $y$ are in the same bag.}
    
 \item A formula $\Psi_i(x,y)$ that expresses that $\Phi_i(x,y)$ holds
   and that, additionally, $x$ and $y$ are both first children or both
   second children of their parents; note that we could not use
   $\firstchild$ and $\secondchild$ here as it would require disjunction.
   We can, however, use the navigation widgets:
   \begin{align*}\Psi_i(x,y) \defeq  \exists
     p_xp_yt_xt_y\bold{t'_x}\bold{t'_y}z\:& \Phi_i(x,y)  \wedge
     \rlabel_p(p_x) \wedge \rlabel_p(p_y) \wedge \rlabel_1(t_x) \wedge
     \rlabel_1(t_y) \\&
        \wedge \child(x,p_x) \wedge \child(y,p_y) \wedge
        \bold{\child(p_x,t'_x)} \wedge \bold{\child(p_y,t'_y)}\\& \wedge \bold{\childself(t'_x,t_x)} \wedge 
                                  \bold{\childself(t'_y,t_y)} \\&\wedge \child^{\bold{(2i-1)+4}}(z,t_x) \wedge
\child^{\bold{(2i-1)+6}}(z,t_y)\end{align*}
Observe that when $x$ and $y$ are both first children, the $t_x$ and
$t_y$ are grandchildren of the $p$-node, and therefore \bold{at distance
$(2i-1)+3$ of the $r$-node}, so going up \bold{$(2i-1)+4$} times brings us to the
$\ct$-node of the current configuration, and going up \bold{$(2i-1)+6$} times
brings us to the $\ct$-node of the preceding configuration. Similarly, if
$x$ and $y$ are both second children, the $t_x$ and $t_y$ are children of
the $p$-node, so going up \bold{$(2i-1)+4$} times brings us to the 
parent of the $\ct$-node of the current configuration, and going up
\bold{$(2i-1)+6$} times 
brings us to the parent of the $\ct$-node of the preceding configuration.
    \bold{This is one of the two places we need the chain of $\ell$ nodes at the root:
    otherwise, since the initial configuration does not have a preceding
    configuration, we would not be able to go high enough up in the tree
    to find the~$z$ node. Taking $\ell\geq1$ suffices.}

    \bold{There is a tree decomposition of width 2 of this subquery where both
    exported variables $x$ and $y$ are in the same bag.}

 \item A formula $\samecell(s_1,s_2)$ that expresses that two $s$-nodes
   encoding a node at depth~$n$ in the tape tree (i.e., at the bottom of
    the tape tree) 
   correspond to the same cell of successive configuration tapes:
   \begin{align*}
     \samecell(s_1,s_2) \defeq \exists x_1\cdots x_{n-1}y_1\cdots
     y_{n-1} \, &\bigwedge_{1 \leq i <n-1} \big( \child^{\bold{2}}(x_i,x_{i+1}) 
     \wedge \child^{\bold{2}}(y_i,y_{i+1}) \big) \\&\wedge
     \child^{\bold{2}}(x_{n-1},s_1) \wedge \child^{2}(y_{n-1},s_2) \\&\wedge \Psi_n(s_1,s_2) \wedge \bigwedge_{1 \leq i <n}
\Psi_{i}(x_i,y_i).
\end{align*}
    \bold{There is a tree decomposition of width $2$ of this subquery where both
    exported variables $s_1$ and $s_2$ are in the same bag.}
\end{itemize}

We can now use these subformulae in the following sentence, that
expresses the final conjunctive query $Q$.
It checks whether the two same cells $s_1$ and $s_2$ of successive
configurations violate vertical constraints. Remember that the value of a
cell is encoded under the $m$-node, while vertical constraints are
encoded under the $f$-node. A vertical constraint occurs when the
(unique) position of a $1$-node under the $m$-descendant of $s_2$ is
equal to 
the position of a $1$-node under the $f$-descendant of $s_1$.
\begin{align*}
  Q \defeq \exists s_1s_2t_1t_2f_1m_2p_1p_2z \,& \samecell(s_1,s_2) \wedge \child(s_1,t_1) \wedge \child(s_2,t_2)
 \\
  &\wedge \child(t_1,f_1) \wedge \child(t_2,m_2)\\
 &\wedge \rlabel_f(f_1) \wedge \rlabel_m(m_2) \wedge \rlabel_1(p_1)
  \wedge \rlabel_1(p_2) 
                        \\&\wedge \bold{(\childself)^{k}(f_1,p_1)} \wedge 
\bold{(\childself)^{k}(m_2,p_2)} \\&\wedge
  \child^{\bold{(2n-1+3)}+k}(z,p_1) \wedge
  \child^{\bold{(2n-1+5)}+k}(z,p_2).
\end{align*}
    \bold{This is the other place we need the chain of $\ell$ nodes at the root:
    otherwise, again, since the initial configuration does not have a preceding
    configuration, we would not be able to go high enough up in the tree
    to find the~$z$ node. Taking $\ell\geq k-1$ suffices.}

The query $Q$ can be constructed in polynomial time, and $Q$ is
valid over
the bNTA previously constructed if and only if the Turing machine $M_w$
has no
accepting (EXPSPACE) computation tree. \bold{$Q$ has treewidth~$2$.}
\end{proof}
\subsection{2EXPTIME-Hardness of MDL containment in a Treelike
CQ}

We show a 2EXPTIME lower bound for the problem of checking the containment of a monadic Datalog
program in a CQ of treewidth $\leq 2$. This matches the general upper
bound for the containment of a Datalog query within a union of CQs.

\begin{theorem}
  \label{thm:containment}
  The following containment problem is
  2EXPTIME-hard over the arity-two signature $\schildself$:
  given a monadic Datalog program $P$ 
  with var-size $\leq 3$
  and a conjunctive query $Q$ of treewidth $\leq 2$,
  decide whether there exists some instance $I$ satisfying $P \wedge \neg Q$.
\end{theorem}

\begin{proof}
 We reduce from the problem of validity of a CQ on $\schildself$ of treewidth
  $\leq 2$ over a bNTA, which is 2EXPTIME-hard by
  Theorem~\ref{th:DtDValid}.

  Let $A=(\calQ, F, \iota, \Delta)$ be a bNTA satisfying the two
  requirements in the previous paragraph,
  and $Q$ a conjunctive query of treewidth $\leq 2$.
  We build a monadic Datalog program $P$ as follows, which obeys the desired
  bound on var-size:
  \begin{itemize}
    \item For every $q\in\calQ$, we have an intensional monadic predicate
      $P_q$.
    \item For every $q\in F$, we have a rule:
      \[
        \text{Goal}() \leftarrow \root(r) , P_q(r).
        \]
    \item For every symbol $\alpha\in\asch$, for every
      $q\in\iota(\alpha)$, we have a rule:
      \[
        P_q(l) \leftarrow \leaf(l) ,
        \rlabel_\alpha(l),\childself(l,l).
        \]
    \item For every symbol $\alpha\in\asch$, for every $q_1,q_2,q'\in\calQ$ such
      that $q'\in\Delta(\alpha,q_1,q_2)$, we have a rule:
      \begin{align*}
        P_{q'}(n) \leftarrow~ &\rlabel_\alpha(n) , P_{q_1}(n_1) ,
        P_{q_2}(n_2),\\&\firstchild(n,n_1),\child(n,n_1),\childself(n,n_1),\\&\secondchild(n,n_2),\child(n,n_2),\childself(n,n_2),\\&\childself(n,n)
      \end{align*}
  \end{itemize}

  Now, by construction, for every unfolding expansion tree $t$ of $P$
  (see proof of Lemma~\ref{lem:caninst}), $I(\Pi_{\ext}(t))$ is a tree over
  $\schildself$. Again by the proof of Lemma~\ref{lem:caninst}, the set
  of the instances of the form $I(\Pi_{\ext}(t))$ where $t$ is an
  unfolding expansion tree of $P$ is a canonical set of instances. In
  particular, there exists an instance satisfying $P$ and not satisfying
  $Q$ if and only if $Q$ is valid over $A$.
\end{proof}

\end{toappendix}

\end{document}